%% file: popl21.tex
\PassOptionsToPackage{final}{microtype} % avoid surprises in final layout
\documentclass[acmsmall,screen]{acmart}
\setcopyright{none}
\renewcommand\footnotetextcopyrightpermission[1]{}
\usepackage{listings}
\usepackage{framed}
\usepackage{pervasives}
\usepackage{graphicx}
\usepackage{wrapfig}
\usepackage{array}
\usepackage{pifont}
\usepackage{makecell}
\usepackage{enumitem}
\usepackage{multirow}
\usepackage{booktabs}
\graphicspath{ {graphs/} }

\begin{document}
\pagestyle{empty}

\author{Ryan Doenges}
\orcid{0000-0002-6899-4529}
\affiliation{%
  \institution{Cornell University}%
  \city{Ithaca, NY}
  \country{USA}%
}
\email{rhd89@cornell.edu}
\author{Mina Tahmasbi Arashloo}
\affiliation{%
  \institution{Cornell University}%
  \city{Ithaca, NY}
  \country{USA}%
}
\email{mt822@cornell.edu}
\author{Santiago Bautista}
\orcid{0000-0003-2129-897X}
\authornote{Work performed at Cornell University.}
\affiliation{%
  \institution{ENS Rennes}%
  \city{Bruz}
  \country{France}%
}
\email{santiago.bautista@ens-rennes.fr}
\author{Alexander Chang}
\affiliation{%
  \institution{Cornell University}%
  \city{Ithaca, NY}
  \country{USA}%
}
\email{apc73@cornell.edu}
\author{Newton Ni}
\affiliation{%
  \institution{Cornell University}%
  \city{Ithaca, NY}
  \country{USA}%
}
\email{cn279@cornell.edu}
\author{Samwise Parkinson}
\affiliation{%
  \institution{Cornell University}%
  \city{Ithaca, NY}
  \country{USA}%
}
\email{stp59@cornell.edu}
\author{Rudy Peterson}
\affiliation{%
  \institution{Cornell University}%
  \city{Ithaca, NY}
  \country{USA}%
}
\email{rnp39@cornell.edu}
\author{Alaia Solko-Breslin}
\affiliation{%
  \institution{Cornell University}%
  \city{Ithaca, NY}
  \country{USA}%
}
\email{ajs644@cornell.edu}
\author{Amanda Xu}
\affiliation{%
  \institution{Cornell University}%
  \city{Ithaca, NY}
  \country{USA}%
}
\email{ax49@cornell.edu}
\author{Nate Foster}
\affiliation{%
  \institution{Cornell University}%
  \city{Ithaca, NY}
  \country{USA}%
}
\email{jnfoster@cs.cornell.edu}

\renewcommand{\shortauthors}{Doenges, Arashloo, Bautista, Chang,
Ni, Parkinson, Peterson, Solko-Breslin, Xu and Foster}

\title{Petr4: Formal Foundations for P4 Data Planes}

\begin{abstract}
P4 is a domain-specific language for programming and specifying
packet-processing systems. It is based on an elegant design with
high-level abstractions like parsers and match-action pipelines that
can be compiled to efficient implementations in software or hardware.
Unfortunately, like many industrial languages, P4 has developed
without a formal foundation. The P4 Language Specification is a
160-page document with a mixture of informal prose, graphical
diagrams, and pseudocode. The P4 reference implementation is a complex
system, running to over 40KLoC of C++ code. Clearly neither of these
artifacts is suitable for formal reasoning.

This paper presents a new framework, called \petra, that puts P4 on a
solid foundation. \petra consists of a clean-slate definitional
interpreter and a calculus that models the semantics of a core
fragment of P4. Throughout the specification, some aspects of program
behavior are left up to targets. Our interpreter is parameterized over a
target interface which collects all the target-specific behavior
in the specification in a single interface.

The specification makes ad-hoc restrictions on the nesting of certain
program constructs in order to simplify compilation and avoid the
possibility of nonterminating programs. We captured the latter
intention in our core calculus by stratifying its type system, rather
than imposing unnatural syntactic restrictions, and we proved that all
programs in this core calculus terminate.

We have validated the interpreter against a suite of over 750 tests
from the P4 reference implementation, exercising our target interface
with tests for different targets. We established termination for the
core calculus by induction on the stratified type system. While
developing \petra, we reported dozens of bugs in the language
specification and the reference implementation, many of which have
been fixed.
\end{abstract}

\maketitle

\thispagestyle{empty}

\section{Introduction}
\label{sec:introduction}

Most networks today are designed and operated without the use of
formal methods. The philosophy of the Internet Engineering Task Force
(IETF), which manages the standards for protocols like TCP and IP, can
be summarized by David Clark's slogan: ``we believe in rough consensus
and running code.'' Likewise, Jon Postel's famous dictum to ``be
conservative in what you do, be liberal in what you accept from
others,'' advocates for a kind of robustness that is achieved not by
adhering to precise logical specifications, but rather by designing
systems that can tolerate minor deviations from perfect behavior.

But while it is hard to argue with the success of modern networks, one
only has to glance at the recent headlines to see that operating a
network correctly is becoming a huge challenge, especially at
scale~\cite{century-link}. Hardware and software bugs frequently rear
their heads, causing service outages, performance degradations, and
security incidents.

Given this context, it is natural to ask whether formal methods may
assist in building networks that behave as intended. Indeed, a number
of recent tools including Header Space Analysis (HSA)~\cite{hsa},
Anteater~\cite{anteater}, NetKAT~\cite{netkat},
Batfish~\cite{batfish}, Minesweeper~\cite{minesweeper},
ARC~\cite{arc}, and others enable operators to automatically verify a
variety of network-wide properties. Startup companies like Forward
Networks, Veriflow Systems, and Intentionet offer commercial products
based on these tools, and even large companies like
Amazon~\cite{amazon-inspector}, Cisco~\cite{cisco-assurance}, and
Microsoft~\cite{secguru,crystalnet} are making substantial investments
in network verification.

Despite significant progress, there is a widening gap between the
simple models used by network verification tools, and the growing set
of features supported on modern routers and switches. Early tools like
HSA and VeriFlow were based on OpenFlow, a stateless packet-forwarding
model that handles about a dozen basic protocols. However, today a
typical data center switch supports 40 or more conventional protocols
(e.g., Ethernet, ARP, VLAN, IPv4, TCP, and UDP), and new protocols
(e.g., VXLAN, Segment Routing, and ILA) are rapidly emerging.
Moreover, even when the protocols are well understood, it can be
difficult to collect the inputs that verification tools require
because device configurations are usually written in idiosyncratic,
vendor-specific formats.

\paragraph*{P4 language}

A promising idea for addressing these challenges is to encode the
behavior of each device in a common representation that is amenable to
analysis. In particular, the P4 language~\cite{bosshart14,p4-spec}
provides a collection of domain-specific abstractions (e.g., header
types, packet parsers, match-action tables, and structured control
flow) that can be used to describe the functionality of a wide range
of packet-processing systems. P4 can be used to model conventional
protocols~\cite{Heule19}, but it is flexible enough to specify
completely new forwarding behavior---e.g., in-band
telemetry~\cite{INT} or in-network computing~\cite{NetCache,netchain}.

Unfortunately, although P4 has been gaining momentum in
industry---companies like Arista, Cisco, NVIDIA, and Xilinx all offer
P4-programmable devices---the language lacks a solid semantic
foundation. The official definition of the language is an informal
document maintained by a language design committee. Parts of the
document are vague, so it is not always clear what a given program
construct means. Turning to the open-source reference implementation
of P4 does not provide clear guidance either, because it is complex,
contains bugs, and occasionally diverges from the specification.
Besides understanding individual programs, the absence of a clear
formal foundation for P4 has made it difficult to understand and evolve
the language itself. For instance, bounded polymorphism has been a
topic of discussion in the language design committee for over three
years, but without the type system written down anywhere it is
difficult to see how such an extension would interact with existing
language features. Overall, in its current form, P4 does not provide a
suitable foundation for reasoning formally about network behavior.

\paragraph*{The \petra framework}

This paper presents \petra (pronounced ``petra''---i.e., Greek for
stone), a new framework that puts P4\footnote{In this paper, when we
  refer to P4, we mean \pfoursix, and not earlier versions of the
  language.} on a solid foundation. \petra is based on two distinct
contributions: (i) a clean-slate definitional interpreter for P4, and
(ii) a calculus that models the formal semantics of a core fragment of
P4. The two artifacts were designed to be consistent---we developed
the calculus after building the interpreter---but they are not
formally related. Our implementation offers a front-end, type checker,
interpreter, and test harness, as well as command-line and web-based
user interfaces. Our calculus models the meaning of a simple P4
program in terms of standard typing and operational semantics
judgments.

\petra builds on standard techniques developed by the programming
languages community over several decades (e.g., definitional
interpreters, type systems, and operational semantics) and applies
these tools to a large, industrial language in a new domain. In
building \petra we had to overcome several challenges. First, as has
already been mentioned, the official definition of P4 is a 160-page
specification document containing informal prose, graphical diagrams,
snippets of code, and a grammar. But while the document is generally
well written, there are some surprising inconsistencies and
omissions---e.g., it does not define P4's lexical syntax or its type
system precisely. Second, P4 is a low-level language with a variety of
constructs for bit-level manipulation of packet data. There are subtle
issues that arise with undefined values, casts, and exceptional
control flow that require a careful treatment. Third, P4 is not really
a single language but a family of languages---there is one dialect for
each architecture that it supports. Hence, to fully understand the
meaning of a P4 program, one must also understand the semantics of the
underlying target devices.

To address these challenges, we first studied the language
specification, reporting dozens of bugs, ambiguities, and
inconsistencies to the language design committee. We then built a
clean-slate definitional interpreter, carefully following the
specification rather than adapting code from the open-source
implementation (i.e., to avoid replicating bugs). One unusual aspect
of our interpreter is that it is parameterized on the choices
delegated to architectures---e.g., what happens when reading or
writing an invalid packet header. We developed a ``plug-in'' model
that enables semantics to be instantiated for many different
architectures, often just by writing a few hundred lines of OCaml
code. We validated our semantics against the test suite for the
open-source implementation, which uncovered additional bugs. Finally,
we extracted a core calculus from our implementation and proved key
properties including type soundness and termination.

\begin{figure}[t]
\begin{minipage}{.475\textwidth}
\centerline{\includegraphics[width=.95\textwidth]{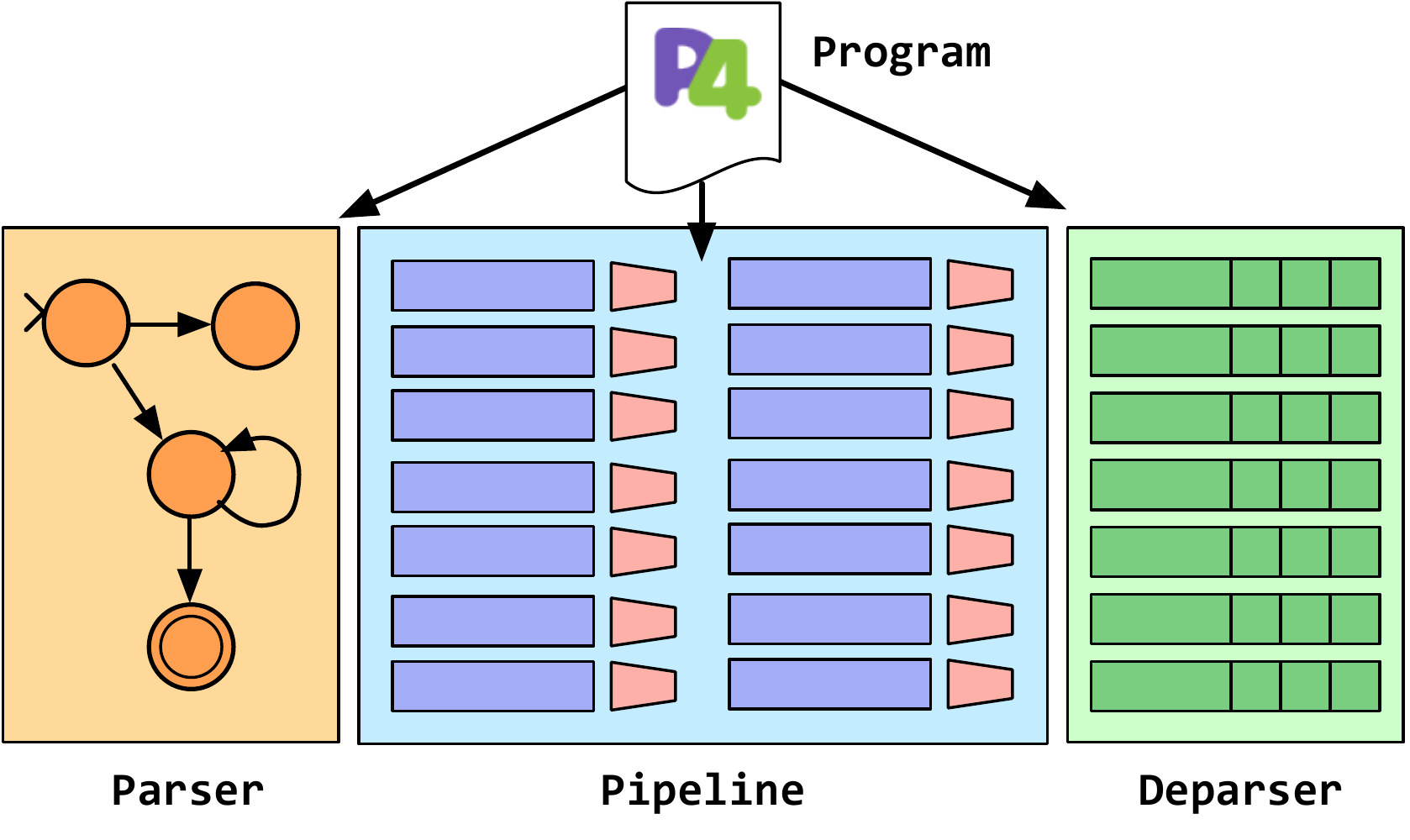}}
\centerline{\footnotesize(a)}
\begin{snugshade}
\begin{lstlisting}
// Architecture externs for packet I/O
extern packet_in {
  void extract<T>(out T hdr);
}
extern packet_out {
  void emit<T>(in T hdr);
}

// Architecture metadata
struct std_meta {
  // metadata initialized on ingress
  bit<8> ingress_port;
  bit<32> packet_length;
  // metadata controlling egress
  bit<8> egress_port;
}

// Architecture types
parser Parse<H>(packet_in pkt,
                out H hdrs,
                inout std_meta meta);
control Pipeline<H>(inout H hdrs,
                    inout std_meta meta);
control Deparse<H>(packet_out pkt,
                   in H hdrs);

// Architecture package
package Switch<H>(Parse<H> parse,
                  Pipeline <H> pipe,
                  Deparse<H> deparse);
\end{lstlisting}
\end{snugshade}
\centerline{\footnotesize (b)}
\end{minipage}\quad\begin{minipage}{.475\textwidth}
\medskip
\begin{snugshade}
\begin{lstlisting}
// Programmer-defined types
header hop {
  bit<7> port;
  bit<1> bos;
}
struct headers {
  hop[9] hops;
}
// Programmer-defined components
parser MyParse(packet_in pkt,
               out headers hdrs,
               inout std_meta meta) {
  state start {
    pkt.extract(hdrs.hops.next);
    transition select(hdrs.hops.last.bos) {
      1: accept;
      default: start;
    }
  }
}
control MyPipe(inout headers hdrs,
               inout std_meta meta) {
  action allow() { }
  action deny() { meta.egress_port = 0xFF; }
  table acl {
    key = { meta.ingress_port : exact;
            meta.egress_port : exact; }
    actions = { allow; deny; }
    default_action = deny();
  }
  apply {
    meta.egress_port =
      (bit<8>)hdrs.hops[0].port;
    hdrs.hops.pop_front(1);
    acl.apply();
  }
}
control MyDeparse(packet_out pkt,
                  in headers hdrs) {
  apply { pkt.emit(hdrs.hops); }
}
Switch(MyParse(),MyPipe(),MyDeparse()) main;
\end{lstlisting}
\end{snugshade}
\centerline{\footnotesize(c)}
\end{minipage}
\caption{Example: (a) Diagram of three-stage architecture; (b) P4
  definition of three-stage architecture; (c) simple P4 program that
  implements source routing with access control in the three-stage
  architecture.}
\label{fig:example}
\end{figure}

\paragraph*{Contributions}
Overall, this paper makes the following contributions:

\begin{itemize}
\item We develop a clean-slate, definitional interpreter for the P4
  language~(\S~\ref{sec:implementation}).
\item We define a calculus that models the semantics of a core
  fragment of P4 in terms of standard typing and operational semantics
  judgments.~(\S~\ref{sec:corep4}).
\item We prove type soundness and termination (\S~\ref{sec:corep4})
  for our calculus.
\item We develop an extension to the language~(\S~\ref{sec:unions}) as
  a case study.
\item We validate our implementation against hundreds of tests from
  the test suite for P4's reference compiler and classify some of the
  bugs we found~(\S~\ref{sec:evaluation}).
\end{itemize}

\noindent Overall, \petra represents a promising first step toward the
vision of formally verified systems built using P4. In particular, we
are optimistic that \petra will not only provide a rigorous foundation
for current language-based verification tools, but will also serve as
a catalyst for future efforts that target higher layers of the
networking stack.

\section{Background}
\label{sec:background}

This section introduces the P4 language using a simple example and
motivates the need for formal foundations by highlighting some of the
opportunities and challenges related to formal reasoning about P4
programs.

\paragraph*{Targets and architectures.}

P4 is a domain-specific language designed for programming a range of
packet-processing targets, including high-speed routers, software
switches, and network interface cards (NICs). Although the details of
these targets vary, they tend to have a few features in common,
including a programmable parser that maps input packets into typed
representations for processing and a pipeline that interleaves
reconfigurable tables and fixed-function blocks. Some targets offer
limited forms of persistent state that can be read and written by each
packet, but they typically do not support general recursion---looping
would require sending the packet through the pipeline multiple times,
which degrades throughput.

The core constructs in P4 capture what these targets have in common.
To accommodate their differences, it also provides the notion of an
architecture, which exposes the structure and capabilities of the
underlying target while abstracting away implementation details. For
example, \Cref{fig:example}~(a) depicts the structure of a simple
architecture that processes packets in three stages: the input packet
is first parsed into a typed representation using a finite-state
machine, then the parsed representation is transformed using a
sequence of match-action tables and arithmetic units, and finally the
parsed representation is serialized into the output packet.

\paragraph*{Example program: three-stage architecture.}

\Cref{fig:example}~(b) shows how this three-stage architecture
can be defined as a part of a P4 program. The architecture definitions
should be read like a Java interface or ML module signature---they
specify the structure and types of each component, but do not define
their implementation. The first few declarations define the types
of \lstinline{extern} objects that can be used to map between raw
packets and typed representations. For instance, the \lstinline{packet_in} object's
\lstinline{extract} method reads from the input packet and
populates the header passed as an argument. The next few declarations
define a \lstinline{struct} type for the metadata associated with the
architecture, including the \lstinline{ingress_port}, which is
initialized by the target when a packet is received; and
the \lstinline{egress_port}, which specifies the port to
use when emitting the packet. The last few declarations define
the P4-programmable components of the architecture: a parser, a
pipeline, and a deparser, as well as the \lstinline{package} that
models the device itself.

\paragraph*{Example program: source routing with access control.}

\Cref{fig:example}~(c) defines a P4 program written against the
three-stage architecture that implements a simple form of source
routing and access control. Here source routing means that each packet
carries a stack of values that encodes the series of ports the packet
should be forwarded out on as it traverses the network, while access
control means the control plane can install filtering rules in a
match-action table to drop certain packets. More formally, each packet
has a fixed-length array (or ``stack'') of byte-sized \lstinline{hop}
headers. Each header is initially ``invalid'' but becomes ``valid''
when it is populated by the parser. For the \lstinline{hop} header,
the first 7 bits encode the output port and the 8th ``bottom of
stack'' bit is \lstinline{1} if it is the last element in the stack.
The \lstinline{MyParse} parser uses a finite state machine abstraction
to map raw input packets into this typed representation. The parser
has a single state that repeatedly extracts \lstinline{hop} headers
from the packet until the \lstinline{bos} marker is \lstinline{1}.
Note that because packets are finite and the loop extracts some bits
from the packet on each iteration, the parser is guaranteed to
terminate. Next, the \lstinline{MyPipe} control defines
an \lstinline{apply} method that specifies how packets are processed.
This method sets the \lstinline{egress_port} metadata field to the
port encoded in the top element of the stack, pops the stack, and then
executes the \lstinline{acl} table, which matches
the \lstinline{ingress_port} and \lstinline{egress_port} metadata
fields against filtering rules (not shown) installed by the
control-plane. The rules either \lstinline{allow} or \lstinline{deny}
the packet, defaulting to \lstinline{deny} if no matching rule can be
found. Finally, the \lstinline{MyDeparse} control serializes the
parsed data back into an output packet.

\paragraph*{Formal methods opportunities and challenges.}

At first glance, P4 appears to be a relatively simple language. So it
seems like it should be possible to use P4 to reason formally about a
range of network scenarios, such as the following:

\begin{itemize}
\item{\textit{Executable specifications of protocols:}} Rather than
  specifying protocols using informal documents, like IETF RFCs, we
  could use P4 to create executable protocol specifications that
  precisely specify packet formats and allowed behaviors. For example,
  the program in \Cref{fig:example}~(c) might serve as the definition
  of the source routing scheme it realizes. Whereas current efforts to
  standardize protocols rely on informal ASCII documents, the P4
  program would provide an unambiguous, mechanized, executable
  reference that could be used to design and validate other
  implementations.

\item{\textit{Program verification:}} P4 programs are expected to
  satisfy various properties---e.g., an IPv4 router should correctly
  decrement the \lstinline{ttl} field and also unambiguously specify
  the forwarding behavior of each packet. Generally speaking,
  verification is simpler than in many other languages because P4
  lacks complex data types and iteration. But current P4 verification
  tools~\cite{p4v,vera} rely on existing front-ends such as the
  open-source reference implementation, which is known to deviate from
  the specification and has bugs. Hence, the results of verification
  are potentially compromised.

\item{\textit{Verified compilers:}} P4 compilers must generate low-level code
  for hardware devices such as programmable switches and FPGAs. This
  process transforms the input program in complex ways---e.g.,
  unrolling parser state machines, eliminating common sub-expressions,
  and extracting parallelism for hardware pipelines. Many of these
  transformations rely on intricate side conditions that are easy to
  get wrong~\cite{ruffy-osdi2020}. A verified compiler for P4, either
  using static verification or translation validation, could eliminate
  bugs in compilers and make it possible to obtain implementations
  that are guaranteed to be correct.

\item{\textit{Proof-carrying code:}} Today, cloud platforms
  allow customers to customize the network infrastructure to suit
  their needs---e.g., they can obtain an isolated virtual network
  slice that they can configure however they like. In the near future,
  cloud providers are likely to go further and allow customers to
  customize the low-level behavior of devices such as routers and
  smart NICs. Techniques such as proof-carrying code~\cite{pcnc} could
  be used to allow P4 programs written by different customers to
  collaborate to implement new features without interfering with the
  functionality of the network as a whole.
\end{itemize}

Unfortunately, while these examples represent some exciting
applications of formal methods to networks, realizing them today would
be difficult. The key challenge is that P4 lacks a formal foundation,
so it is difficult to reason about the language and its programs. More
specifically, we identify three challenges that any formalization of P4
must overcome:

\begin{itemize}
\item{\textit{Incomplete specification}:} The language specification
  is generally well-written but does not fully specify the meaning of
  each language construct. For example, the type system is only
  described at a high level, and important questions such as the
  precise semantics of implicit casts and the definition of type
  equivalence are left unanswered. There are also tricky interactions
  between features that have apparently never been considered, such as
  whether extern objects can have recursive types.

\item{\textit{Undefined values}:} To ease compilation to
  resource-limited targets, P4 makes certain tradeoffs between safety
  and efficiency. For example, P4 allows programs to manipulate
  uninitialized or invalid headers; reading or writing an invalid
  header yields an undefined value. For example, the forwarding
  behavior of the program in \Cref{fig:example}~(c) is undefined in
  cases where \lstinline{hops[0]} is invalid.

\item{\textit{Architecture-specific behaviors}:} The P4 specification
  also delegates many key decisions to architectures, making the meaning
  of a P4 program architecture-dependent. To give one example, the
  behavior of the program in \Cref{fig:example}~(c) depends on whether
  malformed packets---i.e., with more than 9 \lstinline{hops}
  headers---are automatically dropped by the parser or propagated to
  the pipeline. Other architecture-specific behaviors and restrictions
  include the matches and actions supported in match-action tables and
  the availability of certain arithmetic operations such as division.
\end{itemize}

To reason precisely about the behavior of a P4 program today, a
programmer has two main options: they can consult the language
specification or they can execute the program using an existing
implementation. Of course, there are serious issues with either
choice. The specification is incomplete and contradictory, and any
implementation restricts programs to target-specific behavior.

\paragraph*{Our approach.}

Our primary goal in developing \petra was to produce a reusable,
realistic formal semantics for P4.
In particular, we wanted to support executing programs in
a manner that precisely follows the existing specification (to the extent
possible), and facilitate doing formal proofs about programs as well
as the language as a whole. To this end, we developed a clean-slate
definitional interpreter for P4 in OCaml, and we also designed a
calculus that models the type system and operational semantics for a
core fragment of the language. Working carefully from the
specification, our implementation was designed to be independent of
the existing open-source implementation. To resolve situations where
the specification was vague or delegated decisions to architectures,
we parameterized our development, allowing each target to make a
different choice. For example, our calculus models undefined values
using an oracle, and our interpreter is an OCaml functor that can be
instantiated to realize the behavior of different architectures.
Overall, we believe that \petra represents a promising first step
toward our vision of verified data planes, offering a rigorous
foundation as well as running code.

\section{Core P4}
\label{sec:corep4}
\input{syntax.tex}

\section{Implementation}
\label{sec:implementation}

This section presents \petra's definitional interpreter. Unlike the
mathematical semantics for \pcore developed in the last section, which
only models a subset of the language, our implementation is designed
to handle the full \pfoursix language, with a few caveats and
limitations discussed below.

\paragraph*{Overview.}

\Cref{fig:interpreter}~(a) depicts the architecture of the \petra
implementation, as well as the way that programs and packets flow
through it. We implemented \petra in OCaml, using the Menhir parser
generator, the Jane Street \texttt{Core} library, and the \texttt{js\_of\_ocaml}
OCaml-to-Javascript compiler. In total, the \petra implementation runs
13KLoC (as reported by \texttt{cloc}) of which 1.5KLoC implements
lexing and parsing, 1.5KLoC defines syntax, 4KLoC implements
typechecking, and 4.5KLoC implements evaluation/interpretation. The
remaining 1.5KLoC is miscellaneous utility code.

\paragraph{Lexer and parser}
The \pfoursix specification defines the syntax of the language with an
EBNF grammar. Unfortunately the grammar cannot be parsed by any
LALR(1) parser due to a conflict between generics and bit shifts over
the symbols `\texttt{<}' and `\texttt{>}' following identifiers. As a
workaround, the specification separates the tokens for identifiers
into two categories:

\begin{quote}
\textit{The grammar is actually ambiguous, so the lexer and the parser
  must collaborate for parsing the language. In particular, the lexer
  must be able to distinguish two kinds of identifiers: type names
  previously introduced (\texttt{TYPE\_IDENTIFIER} tokens) [and]
  regular identifiers (\texttt{IDENTIFIER} token).}
\end{quote}

\noindent Hence, the parser must keep track of rudimentary type
information as well as lexical scope, so that the lexer can produce
the correct tokens. We follow Jourdan and Pottier's approach for
implementing a parser for C11 in Menhir~\cite{jourdan2017simple}: the
parser maintains a simple context to keep track of the set of type
names, and we wrap a simple lexer that produces \texttt{NAME} tokens
with a second lexer that uses the context to rewrite those tokens into
\texttt{IDENTIFIER} or \texttt{TYPE\_IDENTIFIER} as appropriate.

\paragraph{Type checker}

P4 surface syntax leaves much to the imagination. Function calls may
omit type arguments which have to be inferred. Expressions may be used
at the ``wrong'' type, omitting implicit casts which have to be
inserted by the typechecker. Widths in numeric types, as in \pcore,
may be expressions which have to be evaluated. The \petra type checker
addresses all these issues, converting programs written in an
ambiguous surface syntax into an unambiguous internal syntax. In the
typed internal syntax tree, all nodes are tagged with their type and
all casts and type arguments are made explicit. Compile-time known expressions
are replaced with their values. The \pcore language is closer to this fully
elaborated and typed syntax, although it does retain an account of compile-time
evaluation.

The \pfoursix specification does not precisely define a type system
for the language. Key questions such as how type inference works,
where casts may be automatically inserted, and whether type equivalence is
nominal or structural are not addressed. As an example, the
specification uses the following text to introduce ``don't care''
types:

\begin{quote}
\textit{The ``don't care'' identifier (\texttt{\_}) can only be used for an
\texttt{out} function/method argument, when the value of [sic] returned in
that argument is ignored by subsequent computations. When used in
generic functions or methods, the compiler may reject the program if
it is unable to infer a type for the don't care argument.}
\end{quote}

\noindent However, aside from a brief mention of the Hindley-Milner
inference algorithm~\cite{damas1984type}, there is no explanation of
when the compiler should, if ever, be able to infer a missing type
argument. In practice, \pfourc does use a full Hindley-Milner
implementation to infer type arguments and check type equality, which
has been the source of surprising typechecking bugs~\cite{p4cbug2036}.
What is more surprising is that Hindley-Milner is unnecessary for
\pfoursix. Without solid metatheory available, the language specification
restricts type abstraction to only a few language constructs. In this
simple setting, we found that a much simpler inference algorithm can
get the job done.

The \petra inference algorithm is inspired by local type
inference~\cite{pierce-turner00}, but even LTI is a little heavyweight
for the present state of P4 generics. Where LTI collects type-type
constraints of the form $\tau_1=\tau_2$, \petra is able to stick to
variable-type constraints of the form $X=\tau$. At a call site with
missing type arguments, \petra collects constraints by checking
function arguments, solves those constraints, and then descends back
into the arguments to insert casts where appropriate. The resulting
AST contains no hidden casts or missing type arguments, which makes
life easier for the interpreter.

P4 allows implicit casts between some types. For example, the variable
initialization \lstinline{bit<8> x = 4} will typecheck even though
\lstinline{4} is an \lstinline{int} and not a \lstinline{bit<8>}. The
\petra typechecker inserts a cast and emits a type safe initialization
\lstinline{bit<8> x = (bit<8>)4}. This requires changes to the
inference algorithm to address the combination of implicit casts and
missing type arguments, since two apparently irreconcilable
constraints may be solvable with implicit casts inserted in the right spots.

P4 also includes overloading of functions and extern methods. Here the
specification restricts potential type system complexity by requiring
overloads to be resolvable by just looking at the number or names of
arguments and not their types. Our implementation handles overloading
in the code for checking function calls.

\begin{figure}
\includegraphics[width=.95\textwidth]{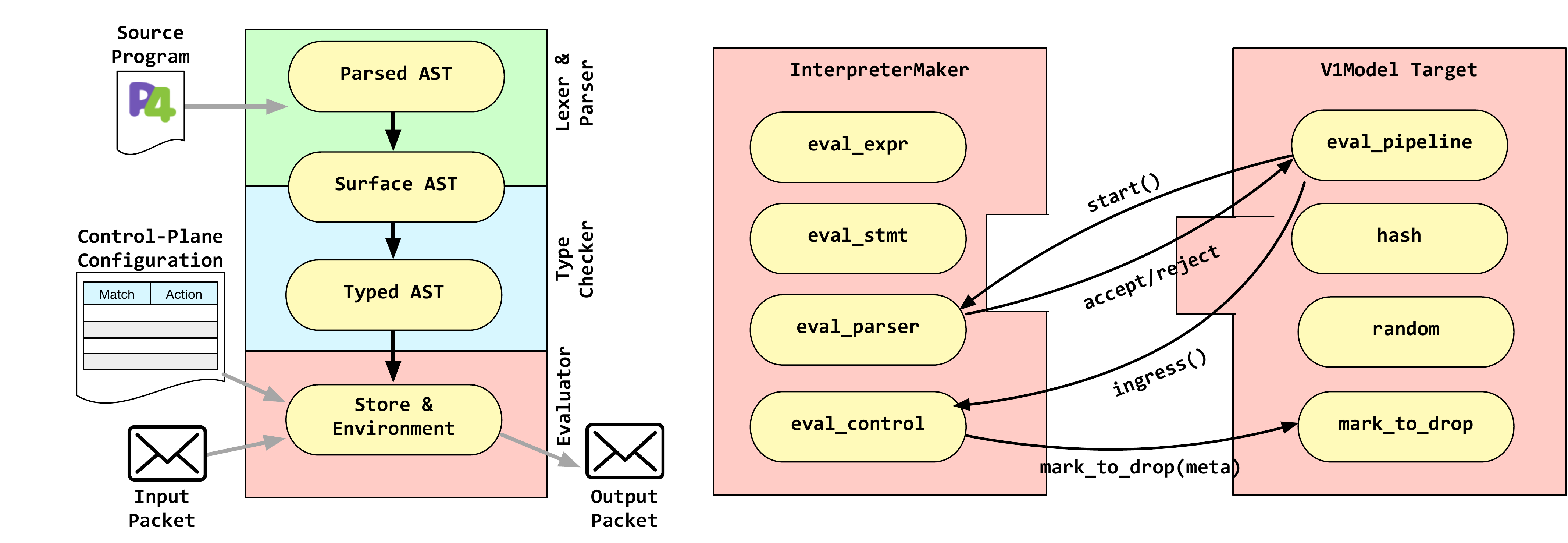}
\centerline{\null\hfill(a)\hfill\hfill(b)\hfill\null}
\caption{\petra implementation: (a) interpreter data flow; (b) architecture support via plug-ins.}
\label{fig:interpreter}
\end{figure}

\paragraph{Interpreter}

The \petra interpreter implements a big-step evaluator, following the
same basic approach as the \pcore evaluation relation
(\Cref{sec:corep4}). However, whereas \pcore uses nondeterminism to
overapproximate possible target-specific behaviors, the \petra
interpreter uses a ``plugin'' approach. The interpreter is an OCaml
functor with the following signature:
\[\texttt{functor (T : Target) $\rightarrow$ Interpreter}\]
The \texttt{Interpreter} module signature includes
functions analogous to \pcore evaluation judgments:
\texttt{eval\_declaration}, \texttt{eval\_statement}, and
\texttt{eval\_expression}.

The \texttt{Target} signature passed into the interpreter functor
defines the interface between a \pfoursix program and the architecture
it runs on. Targets offer a list of externs:
\[
\begin{array}{rl}
\texttt{extern} : & \texttt{env} \to \texttt{state} \to \texttt{type list} \to (\texttt{value} \times \texttt{type}) \texttt{list} \to \\
& \texttt{env} \times \texttt{state} \times \texttt{value}
\end{array}
\]
Each extern is modeled as an OCaml function that takes as input the
environment (\texttt{env}), store (\texttt{state}), type arguments
(\texttt{type list}), and arguments ((\texttt{value} $\times$
\texttt{type}) \texttt{list}), and returns an updated environment
(\texttt{env}), updated store (\texttt{state}), and result
(\texttt{value}). This expansive type reflects how \pfoursix externs
are allowed to do practically anything (short of modifying their
caller's local variables).

Targets must also define the implementation of the packet-processing
pipeline.
\[
\begin{array}{rl}
\texttt{eval\_pipeline} : & \texttt{ctrl} \to \texttt{env} \to \texttt{state} \to \texttt{buf} \to \texttt{apply} \to \\
& \texttt{state} \times \texttt{env} \times \texttt{pkt option}
\end{array}
\]
The pipeline evaluator takes as arguments the control-plane
configuration (\texttt{ctrl}), environment (\texttt{env}), store
(\texttt{state}), input packet (\texttt{buf}), and a hook for
interpreting parsers and controls (\texttt{apply}) and produces an
updated store (\texttt{state}), environment (\texttt{env}) and output
packet (\texttt{pkt option}). As can be seen from this type, \petra
does not currently support multicast, but adding it would be a
relatively straightforward extension.

\Cref{fig:interpreter}~(b) shows how the \texttt{Target} and
\texttt{Interpreter} pass control back and forth during execution,
using the \texttt{V1Switch} architecture as a concrete example.

The output of the \texttt{InterpreterMaker} functor is an
\texttt{Interpreter}, which defines a function for evaluating entire P4
programs:
\[
\begin{array}{rl}

 \texttt{eval\_program} : & \texttt{ctrl} \to \texttt{env} \to \texttt{state} \to \texttt{buf} \to \texttt{int} \to \texttt{prog} \to \\
& \texttt{state} \times (\texttt{buf} \times \texttt{int}) \; \texttt{option}
\end{array}
\]
It takes an initial control-plane configuration, environment, store, packet
buffer and port, along with a program, and produces an updated state and
(optional) modified packet as output.

We have used \petra to construct interpreters for two \pfoursix
architectures: V1Model and eBPF. V1Model is the most widely-used
architecture in open-source \pfoursix code. It includes a variety of
features that exercise the full range of the abstraction boundary
separating interpreter and target. The V1Model pipeline consists of 6
programmable blocks with some fixed-function compenents in between.
The eBPF architecture supports running P4 on the Linux kernel's
packet filter infrastructure. Packet filters have a simpler structure
than V1Model pipelines and support a different collection of
externs. These implementations show that our abstraction effectively
supports multiple architectures.

Adding a new architecture to \petra means writing a few OCaml
functions and datatypes. The implementer has to provide the function
\texttt{eval\_pipeline} above, which defines how control flow passes
between stages of the packet-processing pipeline. The implementer must
also provide data types to represent any \texttt{extern} objects
provided by the architecture and implement their methods. Our current
functor does not model everything left up to architectures in the
specification, but it does cover the most important points. We discuss
this further in \textit{Limitations} and leave a more
precise definition of architecture-dependent behavior to future work.

\paragraph{Control-plane APIs}

The control plane plays an important role in the execution of most
\pfoursix programs by dynamically populating the match-action tables
with forwarding entries. \petra exposes two different control-plane
APIs: one based on a serialization of table entries into JSON, and the
other the ASCII interface supported by the Simple Test Framework (STF)
tool bundled with \pfourc.

For example, the following STF test checks that sending a
packet containing a stack with a single \lstinline{hop} header whose
\lstinline{port} field and \lstinline{bos} fields are both
\lstinline{1} will cause the packet to be forwarded out on port
\lstinline{1}, provided the \lstinline{acl} table is configured to
allow the packet:
\begin{snugshade}\begin{lstlisting}[basicstyle=\footnotesize\ttfamily\color{black}]
add MyPipe.acl MyPipe.acl.ingress_port:0 MyPipe.acl.egress_port:1 MyPipe.allow()
packet 0 03FF
expect 1 FF
\end{lstlisting}\end{snugshade}

\begin{figure}
\fbox{\includegraphics[width=\textwidth]{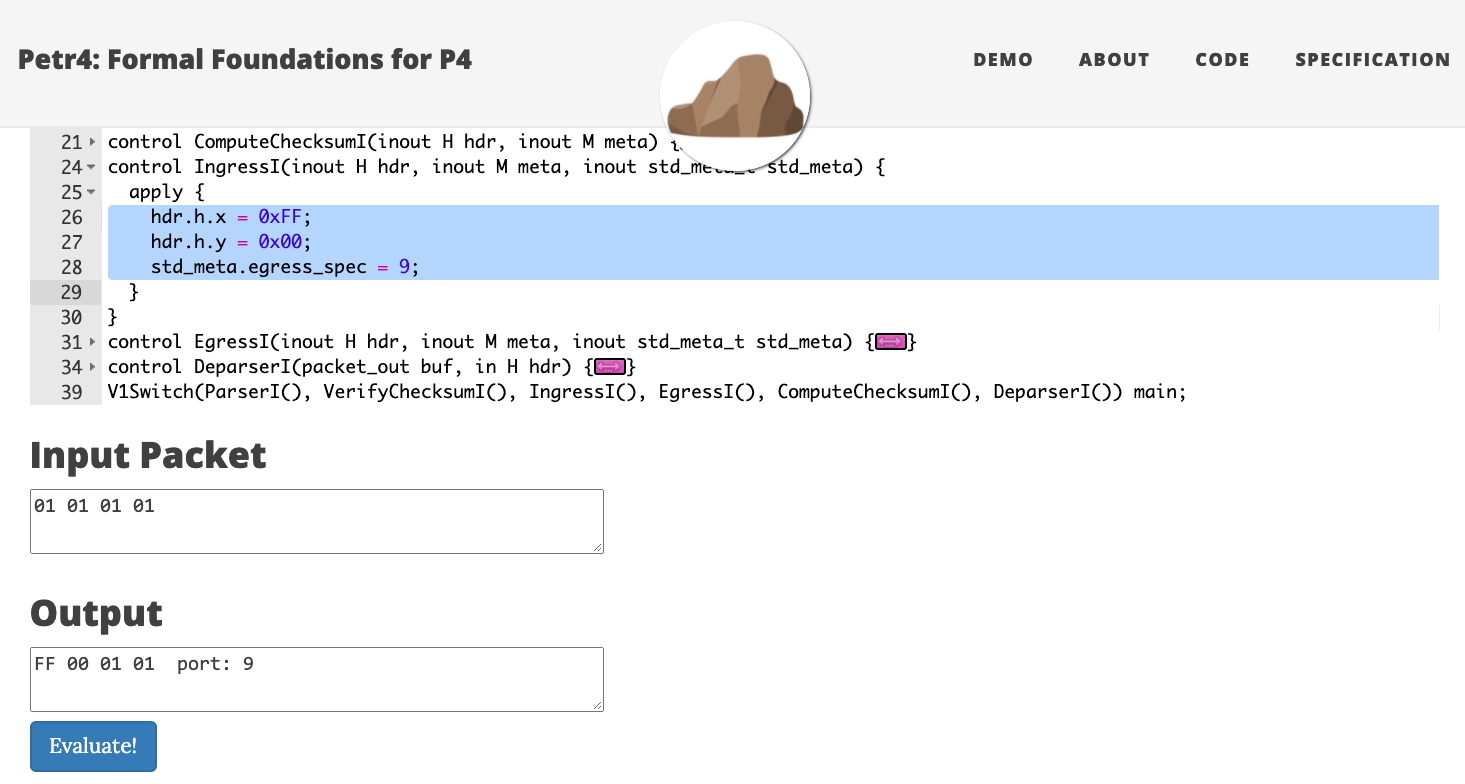}}
\caption{Petr4 interpreter running in a web browser.}
\label{fig:web}
\end{figure}

\paragraph{User interfaces}
We have equipped \petra with two user interfaces. The first provides a
simple command-line interface for \petra that supports several modes
of operation including parsing, type checking, and interpreting a P4
program. The second provides a web-based front-end that runs a P4
program directly in a browser, as shown in \Cref{fig:web}. The
web-based interface is implemented using \texttt{js\_of\_ocaml},
allowing \petra to run directly in the browser. Compared to the
open-source reference implementation, which requires compiling the
program with \pfourc to an intermediate JSON representation that can
then executed on \texttt{bmv2}, a software switch, \petra is
dramatically simpler to use. We expect that both user interfaces will
be useful in teaching P4, as they eliminate much of the overhead and
complexity associated with using \pfourc and \texttt{bmv2}---e.g.,
setting up virtual machines, installing dependencies, hooking into the
Linux networking stack, and coordinating behavior across multiple
stand-alone binaries.

\subsection{Limitations}
\label{sec:limitations}

\petra implements the vast majority of features discussed in the
\pfoursix specification. However, our current prototype does have some
important limitations. \petra implements a sequential model of
computation: this is more restrictive than the specification, which
allows for certain forms of concurrency. \petra also lacks support for
cloning and packet replication. Adding support for both of these
features should be straightforward, but will require additional
engineering both in the formalization and the implementation. \petra
largely ignores annotations, including annotations that can affect
packet processing on some architectures. \petra does not support
\texttt{abstract} externs or user-defined initialization blocks---two
recent additions to the language. \petra's implementation of the
V1Model target omits some externs, including direct-mapped objects.
Finally, while the \texttt{Target} signature exposes hooks that would
allow an implementation to customize behaviors left up to
architectures (e.g., semantics of reads/writes to invalid headers),
some behaviors have not yet been parameterized (e.g., custom
properties for match-action tables).

\section{Evaluation}
\label{sec:evaluation}

To evaluate \petra, we want to determine the correctness and utility
of its semantics and corresponding implementation. That is, we need to
determine how well they capture P4, both as it is really used by the
open-source community, and as it is described in the specification. To
this end, we first explore the results of running \petra against the
same test suite as the reference implementation. We next describe bugs
and ambiguities we discovered and addressed during development, both
in the reference implementation and in the language specification.

\paragraph{Parser and typechecker}

We have imported 792 test cases from \pfourc for the parser and
typechecker. These consist of ``good'' tests (those the typechecker
should accept) and ``bad'' tests (those the typechecker should
reject). Currently, \petra's parser passes all 792 of these. The
typechecker, on the other hand, passes 782 of these, with 10 failures
due to the following issues:

\begin{itemize}

\item A bug in the grammar requiring an additional ``lexer hack'' only recently fixed in \pfourc.

\item The \texttt{@optional} annotation for arguments, which \petra
  does not support.

\item Type casts that discard significant bits of bitstrings should
  emit a warning, but not fail. \pfourc's suite expects them to fail.

\item \pfourc rejects programs that shadow names (control-plane and
  local scope) of functions, actions, controls, tables, and parsers,
  whereas \petra is more permissive. The specification says
  the compiler ``may provide a warning if multiple resolutions are
  possible for the same name'' for some situations but does not
  require it to be a type error.

\item Implicit casts from signed to unsigned integers may turn a bad
  (negative) operand for division into a good (positive)
  value. Division with negative values is not allowed by the
  specification, so the difference with \pfourc in this case is only
  because \petra checks the sign after doing implicit casts rather
  than before.

\item Several restrictions on the structure of programs are imposed by
  \texttt{V1Model} but not enforced by \petra. For
  example, \texttt{V1Model} requires deparser code to be free of
  conditionals, but Petr4 does not enforce this kind of
  architecture-specific syntactic restriction yet.

\end{itemize}

There are an additional 110 tests imported from \pfourc which are
unsupported by our typechecker. For more detail
see \Cref{sec:limitations} above.

\paragraph{Interpreter}

Of the good checker tests, 121 are accompanied by corresponding \stf
files used to test the correctness of \pfourc's back end. As described
in section 4, our control plane API allows us to run these same tests
on our interpreter. We currently pass 95 of these \stf tests with 26
failures. Most of our failures (20 tests) are \texttt{P4} programs
written in architectures unimplemented by \petra (\texttt{PSA} and
\texttt{UBPF}). The remaining 6 utilize externs in the \texttt{EBPF}
and \texttt{V1Model} architectures that \petra also leaves
unsupported, such as multicast and the \texttt{crc16} checksum
algorithm. Some of the more interesting tests imported from \pfourc
are described in detail in \Cref{fig:eval2}. We also provide 40 of our
own custom \stf tests accumulated during test-driven development of
the interpreter that address difficult edge cases of the language we
felt the \pfourc suite did not sufficiently exercise. \petra passes
all 40 custom tests, a sample of which are described in detail in
\Cref{fig:eval1}.

\begin{figure}[t]
	\footnotesize
	\input{evaltable2.tex}
	\caption{Selection from \pfourc's \stf test suite.}
	\label{fig:eval2}
\end{figure}
\begin{figure}[t]
	\footnotesize
	\input{evaltable1.tex}
	\caption{Selection from \petra's custom \stf test suite.}
	\label{fig:eval1}
\end{figure}

In developing \petra, we uncovered bugs in \pfourc, ambiguities in the
informal \pfoursix spec, and issues with \pfourc arising from choices
it made to resolve these ambiguities. We describe some bugs here, but
see \Cref{fig:p4cbugs} and \Cref{fig:p4specbugs} for a full list. All
bugs have been reported to either the \pfoursix specification
repository or the \pfourc repository on Github.

\begin{figure}[t]
	\footnotesize
	\input{p4cbugs.tex}
	\caption{\pfourc Issues}
	\label{fig:p4cbugs}
\end{figure}
\begin{figure}[t]
	\footnotesize
	\input{p4specbugs.tex}
	\caption{\pfoursix Specification Issues}
	\label{fig:p4specbugs}
\end{figure}

\paragraph*{Grammar and parser.}

The \pfoursix grammar allowed annotations to either take an expression
list or a list of key-value pairs for their arguments. This approach
introduced an ambiguity into the grammar: there was no way to
discern whether an empty list was an expression list or a key-value
pair list. We eliminated this ambiguity by allowing
the annotations to take a non-terminal in the grammar called
\texttt{argumentList} in which each argument could be either an
expression or a key-value pair. Additionally, this simplification
allowed for more flexible behavior---mixing expression and key-value
pair arguments---that strictly contained the previous behavior without
introducing any new problems.

Even before support for top-level functions was implemented, we
discovered a conflict between function declarations and newtype
declarations. (A newtype declaration \lstinline{type name_t old_t}
creates an opaque type alias \lstinline{name_t} for the
type \lstinline{old_t}, like \lstinline{newtype} in Haskell.)
The P4 grammar begins both function and newtype declarations with a
token sequence \texttt{TYPE TYPE\_IDENTIFIER}.
The \texttt{TYPE} token corresponded to not only a \texttt{nonTypeName}
in the function declaration, but also the \texttt{type} keyword in
the newtype declaration. Thus, the parser could either reduce a
\texttt{nonTypeName} out of \texttt{TYPE} or shift to recognize a
newtype declaration.

\paragraph*{Type checker.}

We found multiple discrepancies between \pfourc and the \pfoursix specification
with respect to typechecking. \pfourc rejected any program with
headers containing multiple \texttt{varbit} fields. However, the specification
only requires that such headers cannot be used in
\texttt{extract}. Since P4 implementations are allowed to provide
extensions that could make such headers useful, \pfourc's
restriction was too strong compared to the spec. Conversely, the specification
is too restrictive in permitting division and modulo between
positive \texttt{int} values only, whereas \pfourc relaxes this constraint
to permit the same operations between \texttt{bit} values.

\paragraph*{Semantics of P4 constructs.}

The P4 specification originally imposed a restriction on inserting
implicit casts for the arguments to methods or functions. Implicit
casts were intended to reduce the friction for the programmer and
allow her to use constants naturally. Thus, the restriction was rather
undesirable and was lifted when the specification was amended to allow implicit
casts on \texttt{in} arguments for methods and functions. We also
influenced another amendment to define explicitly the
concatenation operator's behavior for both signed and
unsigned bitstrings. Prior to this change, it was unclear whether the
concatenation operator was supported for signed bitstrings until we
found that \pfourc allowed it.

\paragraph*{Inconsistent type system, buggy inference, untyped constant folding.}

These bugs emphasize subtle ways in which handling types can be tricky
to get right. First, we discovered that the type system was
inconsistent in that it was sometimes nominal and sometimes
structural. For example, controls were structural in architecture
definitions, and nominal elsewhere. The solution we implemented was to
distinguish a control from the type of the control. Consequently, a
control could no longer be used as the type of a parameter. In
conjunction, a tuple cannot be assigned to a struct where list
expressions can, and in fact have a tuple type. This subtlety allowed
tuples to be assigned to structs. This was fixed by checking for a
tuple type by means of a struct-like type conversion and introducing a
new type internally for list expressions. Another example of a subtle
type issue was in \pfourc's constant folding. Because it was not
paying enough attention to types, it transformed an ill-typed program
into a well-typed one in some cases.

\section{Case study: adding type-safe unions to P4}
\label{sec:unions}
\input{unions.tex}

\section{Related work}

The problem of formalizing the semantics of a language is one of the
oldest problems in our field, and it remains an active and relevant
area of research today. This section briefly reviews some of the most
closely related work.

\paragraph{Semantics for industry languages}

Formal models have recently been developed for a growing number of
practical languages used in industry. Pioneering work by Milner,
Tofte, Harper, and MacQueen developed a formal definition of Standard
ML, one of the first languages to be given such a
treatment~\cite{standard-ml}. More recently, a number of prominent
efforts have developed semantics for languages as complex and diverse
as JavaScript~\cite{guha:js,park-stefanescu-rosu-2015-pldi},
WebAssembly~\cite{wasm17}, C~\cite{leroy2009formal},
x86-TSO~\cite{sewell2010x86}, and the POSIX shell~\cite{Greenberg20Smoosh}.
Like \petra, these efforts build on decades of foundational work
in semantics~\cite{scott1971toward,plotkin1981structural,kahn87} and
semantics engineering~\cite{JFP10}. Recent work by Ruffy~et~al.\ has
profitably combined fuzzing and translation validation to find
numerous bugs in \pfourc~\cite{ruffy-osdi2020}. Their Gauntlet
translation validator defines the behavior of P4 programs by an
SMT-LIB encoding, making program equivalence checkable with a single
Z3 query. Our work focused on building a reusable semantics which
could, for example, verify the translation used in Gauntlet. Of
course, the translation in Gauntlet could also be productively applied
to fuzzing the Petr4 interpreter.

\paragraph{Semantics for networks.}

In the networking context, Sewell et al. developed mechanized formal
models of TCP and UDP~\cite{netsem} using HOL4. A key challenge was
designing a ``loose'' semantics that could accommodate the
implementation choices made by different network stacks. The same
issue arises in \petra when modeling architecture-specific features,
such as read and write operations to invalid headers. Guha, Reitblatt,
and Foster developed a verified compiler from NetCore, a high-level
policy language, to OpenFlow, an early software-defined networking
standard~\cite{guha13}. Another line of recent work has focused on
eBPF, the packet-processing framework supported in the Linux kernel.
The JitK compiler~\cite{jitk} uses a machine-verified just-in-time
compiler to generate code that is guaranteed to satisfy the safety
conditions enforced by the kernel verifier, while JitSynth leverages
program synthesis~\cite{jitsynth}.

\paragraph{Network verification}

As mentioned in \Cref{sec:introduction}, there is a growing body of
work focused on data plane and control plane verification, including
Header Space Analysis (HSA)~\cite{hsa}, Anteater~\cite{anteater},
NetKAT~\cite{netkat}, Batfish~\cite{batfish},
Minesweeper~\cite{minesweeper}, and ARC~\cite{arc}, to name a few. Other
tools have applied techniques such as predicate
transformers~\cite{p4v}, symbolic execution~\cite{vera,p4pktgen}, or
translation into another language, such as
Datalog~\cite{mckeown2016automatically}, to verify P4 programs.
However, none of these tools are based on a foundational semantics
like \petra---they either rely on ad hoc models or rely on an existing
implementation such as \pfourc. Kheradmand and Rosu developed an
operational model for P4 in the K
framework~\cite{kheradmand2018p4k}. The P4K project implemented
P4$_{14}$, which has substantially different syntax and semantics from
\pfoursix, and provided an interpreter without an accompanying type
system. The interpreter is implemented in the K framework, which was
able to produce verification and translation validation tools
automatically from the interpreter definition. However, the encoding
in K limits its reusability outside of the K framework.

\section{Conclusion and future work}

This paper introduced \petra, a formal framework that models the
semantics of P4. We developed a clean-slate definitional interpreter
for P4 as well as a formal calculus that models the essential features
of the language. The implementation has been validated against over
750 tests from the reference implementation and the calculus was
validated with a proof of type-preserving termination.

In the future, we would like to extend our calculus to model the full
language and publish it as the official specification of the language
for the P4 community. We believe it would be a valuable resource for
designers, compiler writers, and application programmers alike.
Concretely, we would like to close the gap between our definitional
interpreter and calculus, obtaining a formal semantics that covers the
entire language. It would also be attractive to have a mechanized
semantics so the reference interpreter can be extracted from the
formalization. Toward this end, we have begun porting our definitional
interpreter to Coq. We do not foresee any major technical challenges
and believe it should be possible to complete this task
quickly---i.e., in a matter of weeks or months---though porting our
type soundess and termination proofs will take longer. The biggest
obstacles are likely to be related to architectures and extern
functions, which are straightforward to handle in principle but
somewhat tedious to implement in practice. Looking further ahead, we
eventally hope to use our Coq formalization to develop a verified
compiler for P4. We are also interested in using \petra to guide
development of further enhancements to P4---e.g., designing a smaller
core language to streamline development of tools, and adding full
support for generics and a module system to the language.

\begin{acks}
We are grateful to the POPL '21 reviewers for their feedback and many
suggestions for improving this paper. We wish to thank Chris Sommers
for many discussions on formalizing P4, and Michael Greenberg for
advice on presenting this work. Our work has been supported in part by
the National Science Foundation under grant FMiTF-1918396, the Defense
Advanced Research Projects Agency (DARPA) under Contract
HR001120C0107, and gifts from Keysight, Fujitsu, and InfoSys.
\end{acks}

\bibliography{popl21}

\appendix
\input{typesafety.tex}
\input{unions-proof.tex}

\end{document}

%% file: syntax.tex
This section presents syntax and semantics for \pcore, a simple
language that models the essential features of P4 in a core calculus.
P4 is a large and idiosyncratic language and while our definitional
interpreter handles nearly all of its features, formalizing the full
language in a paper would be unwieldy. We offer here a selective
transcription of the semantics realized in the \petra
implementation. The semantics is sufficiently rich to capture the
feature interactions that make P4 tricky to reason about, while
avoiding the notational clutter of the full language. The most
significant omission from \pcore is parsers. Hence, \pcore models the
essential packet-processing done by \lstinline{control} blocks but
omits recursion, which allows us to prove a termination result.

If desired, parsers that have been unrolled to eliminate recursion can
be emulated using \pcore's functions, with one function for each
state.  This retains the termination theorem and is often done in
practice on resource-constrained targets. Indeed, the P4 specification
states that compilers ``may reject parsers containing loops that
cannot be unrolled at compilation time.''

\subsection{Syntax and examples}

\pcore is a mostly standard imperative language, with separate
syntactic classes for expressions, statements, and declarations. It
also includes mutable variables, generic functions, and standard types
such as booleans, enumerations, and records.
Architecture-specific functionality is modeled using ``native''
functions. For example, the following \pcore program models
the \lstinline{apply} block from \Cref{fig:example}~(c):
\[
\begin{array}{l}
\Passign{\kw{meta}.\kw{egress\_port}}{\Pcast{\Pbittype{8}}{\kw{hops}[0]}};\\
\Pcall{\kw{pop\_front}}{\kw{hops}, 1};\\
\Pcall{\kw{acl}}{};\\
\end{array}
\]
Note that with a few exceptions, such as the use of function calls to
model header stack operations (\kw{pop\_front}) and match-action
tables (\kw{acl}), the original program and the \pcore program are
nearly identical.

The P4 specification imposes a multitude of restrictions on type
nesting, parameter types, locations of instantiations, and
other language constructs. While we stratify the \pcore type system
to prevent higher-order phenomena, we avoid modeling the remainder of
the specification's restrictions in \pcore. Nonetheless, \pcore is
type safe. The restrictions aim to simplify compiling P4 programs to
sometimes idiosyncratic and resource-limited hardware targets. They
are not fundamental and could be lifted if P4 compilers developed new
resource allocation strategies and optimizations.

\subsubsection{Notational conventions}

We typeset metavariables in $\mathit{italics}$ and keywords and other
concrete identifiers in $\kw{sans~serif}$. We avoid explicit indexing
of sequences by writing a line over the term we would otherwise index.
For instance, $\Pmany{x}$ represents a list $x_1, x_2, \dots, x_n$. We
write $x$ for ordinary variables and $X$ for type variables and names.
We write $f$ for fields of records or members of enumeration and
``open enumeration'' types. There are two open enums, which have the
reserved type names $\kw{error}$ and $\kw{match\_kind}$. Locations
$\ell$ appear in the dynamic semantics. We write $\Pfresh{\ell}$ to
obtain a new location $\ell$.

\begin{figure}\footnotesize
\[
\begin{array}{ll}
\begin{array}{rcll}
\Psyndef{\rho}{\Pbool}{booleans}
\Psyncase{\Pinttype}{integers}
\Psyncase{\Pbittype{\Pexp}}{bitstrings}
\Psyncase{\Perrortype{f}}{errors}
\Psyncase{\Pmatchtype{f}}{match kinds}
\Psyncase{\Penumtype{X}{\Pmany{f}}}{enums}
\Psyncase{\Prectype{\Pmany{f:\rho}}}{records}
\Psyncase{\Pheader{\Pmany{f:\rho}}}{headers}
\Psyncase{\Pstacktype{\rho}{n}}{stacks}
\Psyncase{X}{type variables}
\end{array} &
\begin{array}{rcll}
\Psyndef{\tau}{\rho}{data types}
\Psyncase{\Ptabletype}{tables}
\Psyncase{\Pfunctiontype{\Pmany{X}}{\Pmany{d~x:\rho}}{\rho_{\mathit{ret}}}}{functions}
\Psyncase{\Pctortype{\Pmany{x:\tau}}{\tau_{\mathit{ret}}}}{constructors}
\\
\Psyndef{d}{\Pin}{copy-in}
\Psyncase{\Pout}{copy-out}
\Psyncase{\Pinout}{copy-in-out}
\\
\\
\\
\end{array}
\end{array}
\]
\caption{\pcore types and directions.}
\label{fig:types}
\end{figure}

\subsubsection{Types (\cref{fig:types}).}
\pcore types are separated into function types $\tau$ and
base types $\rho$, with generics only allowed to range over base types.

Numeric datatypes in P4 are flexible. Consider this header type
representing an IPv4 option:
\begin{align*}
&\Pheader{
\\
&\quad\kw{copyFlag}:\Pbittype{1}
\\
&\quad\kw{optClass}:\Pbittype{2}
\\
&\quad\kw{option}:\Pbittype{5}
\\
&\quad\kw{optionLength}:\Pbittype{8}
\\
&}
\end{align*}
Each field is an unsigned integer (a \kw{bit} type) with its width
specified in angle brackets. This is convenient for network
programming, where network protocols can involve 1 or 5-bit values
packed into the packet without padding. P4 allows the width of a
numeric type to be an expression, provided it can be evaluated at
compile time. The presence of expressions in types complicates type
equality, as can be seen in this short example.
\begin{align*}
&\Pconst{\Pinttype}{\kw{w}}{8};
\\
&\Pvardeclinst{\Pbittype{\kw{w}}}{\kw{x}}{1};
\\
&\Pvardeclinst{\Pbittype{8}}{\kw{y}}{\kw{x}};
\end{align*}
The type $\Pbittype{\kw{w}}$ is not syntactically equal to $\Pbittype{8}$, but
the type checker should permit the assignment. The \pcore type system handles
this by reducing types to a normal form before comparing the normal forms with
syntactic equality (modulo $\alpha$-equivalence for generics). The
implementation of this equality check will in some situations impose type
equality by inserting casts, but we do not model implicit casts in \pcore.

The $\kw{match\_kind}$ and $\kw{error}$ types are ``open
enumerations,'' comparable to the extensible exception type in
Standard ML~\cite{standard-ml}.  Repeated declarations extend the open
enumeration with new members without replacing the old members or
shadowing the existing type.

\begin{figure}\footnotesize
\centering\(
\begin{array}[t]{ll}
\begin{array}{rcll}
\Psyndef{\Pexp}{b}{booleans}
\Psyncase{\Pint{n}{w}}{integers}
\Psyncase{x}{variables}
\Psyncase{\Pexp_1[\Pexp_2]}{array accesses}
\Psyncase{\Pexp_1[\Pexp_2{:}\Pexp_3]}{bitstring slices}
\Psyncase{\ominus\,\Pexp}{unary ops}
\Psyncase{\Pexp_1 \oplus \Pexp_2}{binary ops}
\Psyncase{\Pcast{\rho}{\Pexp}}{casts}
\Psyncase{\Prec{\Pmany{f=\Pexp}}}{records}
\Psyncase{\Pexp.f}{fields}
\Psyncase{X.f}{type members}
\Psyncase{\Pgencall{\Pexp}{\Pmany{\rho}}{\Pmany{\Pexp}}}{function call}
\end{array}
&
\begin{array}{rcll}
\Psyndef{\Pstmt}{\Pgencall{\Pexp}{\Pmany{\rho}}{\Pmany{\Pexp}}}{method call}
\Psyncase{\Passign{\Pexp}{\Pexp}}{assignment}
\Psyncase{\Pif{\Pexp}{\Pstmt}{\Pstmt}}{conditional}
\Psyncase{\Pblk}{sequencing}
\Psyncase{\Pexit}{exit}
\Psyncase{\Preturn{\Pexp}}{return}
\Psyncase{\Pvdecl}{variable declaration}
\\[-\bigskipamount]
\Psyndef{\Plval}{x}{local variables}
\Psyncase{\Plval.f}{fields}
\Psyncase{\Plval[n]}{array elements}
\Psyncase{\Plval[n_1:n_2]}{bitstring slices}
\end{array}
\end{array}
\)
\caption{\pcore expression, statement, and l-value syntax.
The expression on the left of an assignment is not an l-value to allow
(for example) a computed array index, which evaluates to an l-value
with a fixed index.}
\label{fig:expstmtsyntax}
\end{figure}

\subsubsection{Expressions (\cref{fig:expstmtsyntax})}

\pcore offers a rich set of expressions for manipulating packet 
contents. For example, the following program extracts the 6th byte of
a bitstring \kw{bits}:
\begin{align*}
&\Pconst{\Pbittype{48}}{\kw{bits}}{\dots};\\
&\Pconst{\Pinttype}{\kw{n}}{6};\\
&\Pvardeclinst{\Pbittype{8}}
              {\kw{nth\_byte}}
              {\kw{bits}[n*8-1{:}(n-1)*8]}
\end{align*}
The bitstring slice operator
$exp[exp_{\mathit{hi}}{:}exp_{\mathit{lo}}]$ computes a slice of the
bits of $exp$ from the high bit at $exp_{\mathit{hi}}$ down to the low
bit $exp_{\mathit{lo}}$ (inclusive). Since the slice endpoints appear
in the type ($\Pbittype{exp_{\mathit{hi}} - exp_{\mathit{lo}} + 1}$)
they must be known at compile-time.

Unary operations $\ominus$ and binary operations $\oplus$ are drawn
from a set of symbols including standard arithmetic and bitwise
operations as well as comparisons and equality. Casts are permitted
between numeric types and from record types to header types.

\subsubsection{Statements (\cref{fig:expstmtsyntax})}
\pcore's statement language is small and mostly standard.

Constants are available for use at the type level, so their
initializers must themselves be known at compile-time.

Exit statements abort an entire computation. For example, if we pass
an invalid IP header to $g$ in the following program, the \Pexit
statement in $f$ causes the second call to never happen:
\begin{align*}
&\Psimplefunctiondecl{\{\}}{f}{\Pin~\kw{hdr\_t}~h}{
  \Pif{!\kw{isValid}(\kw{h}.\kw{ip})}{\{\Pexit\}}{\{\dots\}}
}\\
&\Psimplefunctiondecl{\{\}}{g}{\Pin~\kw{hdr\_t}~h}{
  f(h); f(h)
}
\end{align*}

An instantiation takes the form $\Pinst{X}{\Pmany{\Pexp}}{x}$ and
creates an object named $x$ by invoking the constructor for the type
$X$.  In full P4 there are restrictions on what kinds of objects can
be instantiated where, but we do not reproduce these rules in \pcore.

\begin{figure}\footnotesize
\centering
\(
\begin{grammar}
\Psyndef{\Pdecl}{\Pvdecl}{variables}
\Psyncase{\Podecl}{objects}
\Psyncase{\Ptdecl}{types}
\Psyndef{\Pvdecl}{\Pconst{\tau}{x}{\Pexp}}{constants}
\Psyncase{\Pvardeclinst{\tau}{x}{\Pexp}}{local variables (initialized)}
\Psyncase{\Pvardecl{\tau}{x}}{local variables (uninitialized)}
\Psyncase{\Pinst{X}{\Pmany{\Pexp}}{x}}{instantiations}
\Psyndef{\Ptdecl}{\Ptypedef{\tau}{X}}{typedefs}
\Psyncase{\Penumdecl{X}{\overline{f}}}{enums}
\Psyncase{\Perrordecl{\Pmany{f}}}{errors}
\Psyncase{\Pmatchkinddecl{\Pmany{f}}}{match kinds}
\Psyndef{\Podecl}{\Ptable{x}{\Pmany{\Pkey}~\Pmany{\Pact}}}{tables}
\Psyncase{\Pcontroldecl{X}{\Pmany{d~x:\tau}}{\Pmany{x:\tau}}{\Pmany{\Pdecl}}{\Pstmt}}{controls}
\Psyncase{\Pfunctiondecl{\tau}{x}{\Pmany{X}}{\Pmany{d~x:\tau}}{\Pstmt}}{functions}

\Psyndef{\Pkey}{\Pexp:x}{table keys}
\Psyndef{\Pact}{x(\Pmany{\Pexp},\Pmany{x:\tau})}{actions}
\Psyndef{\Pprog}{\Pmany{\Pdecl}}{programs}
\end{grammar}
\)
\caption{\pcore declarations and programs.}
\label{fig:declsyntax}
\end{figure}

\subsubsection{Declarations and programs (\cref{fig:declsyntax})}

Declarations are partitioned into variable declarations, object
declarations, and type declarations. Variable declarations are part of
statements, which have already been introduced, and type declarations
are essentially types, which are discussed above. This leaves object
declarations: tables, controls, and functions.

P4 tables can be thought of as generalizing routing tables and switch
statements. Like routing tables on specialized network hardware, they
store a list of pattern-matching rules that can be edited at run time.
Like switch statements, they may run different code depending on the
value of an expression.

In the following example, a table inspects the packet's destination
Ethernet address and either sets its egress (output) port or drops the
packet.
\begin{align*}
&\Psimplefunctiondecl{\{\}}{\kw{set\_port}}{\Pbittype{9}~\kw{port}}{
\kw{meta}.\kw{egress\_port}=\kw{port};
}
\\
&\Psimplefunctiondecl{\{\}}{\kw{drop}}{}{
\kw{meta}.\kw{drop}=\Ptrue;
}
\\
&\Ptable{\kw{forward}}{\{\kw{hdr}.\kw{eth}.\kw{dstAddr:\kw{exact}}\}\ \ \{\kw{set\_port}; \kw{drop};\}}
\end{align*}

The $\kw{meta}$ struct contains metadata about the packet, while
$\kw{hdr}$ holds the parsed contents of the packet. The $\kw{exact}$
annotation on the key indicates that patterns in rules should be
matched exactly, as opposed to ranges, longest prefixes, or any other
$\kw{match\_kind}$ supported by the architecture.

We do not model table rules in \pcore and instead overapproximate them
by assuming a ``control plane'' $\mathcal{C}$ that deterministically
selects an action given an identifier for a table and values for its
keys. The identifier is an internal location rather than a table name
so that distinct table instantiations arising from the same
declaration can have separate rules.

Control declarations include a list of parameters (with directions)
and a list of constructor parameters (without directions). The body of
the declaration includes a list of declarations followed by a
statement, which is typically a block containing several statements.
While \pcore does not impose this restriction, the full P4 language
requires tables and other stateful objects to be declared within
controls rather than at the top level.

Functions are standard, although recursion is not permitted and type
parameters can only be instantiated with base types ($\rho$).

\subsubsection{L-values (\cref{fig:expstmtsyntax})}

An l-value is an expression that can appear on the left-hand side of
an assignment statement. They are built up from variables, array
indexing, field lookup, and bitslices. A syntactic distinction between
expressions and l-values is not enough in general because of function
calls, which require arguments for \kw{out} or \kw{inout} parameters
to be l-values but make no such imposition on their \kw{in} arguments.
To address this the type system checks whether expressions are
assignable (see \cref{sec:static-semantics}).

\begin{figure}\footnotesize
\centering
\(
\begin{grammar}
\Psyndef{\Pval}{b}{booleans}
\Psyncase{\Pint{n}{w}}{integers}
\Psyncase{\Prec{\Pmany{f=\Pval}}}{records}
\Psyncase{\Pheaderval{\Pmany{f:\tau=\Pval}}}{headers}
\Psyncase{X.f}{type members}
\Psyncase{\Pstackval{\tau}{\Pmany{\Pval}}}{header stacks}
\Psyncase{\Pclos{\Penv}{\Pmany{X}}{\Pmany{d~x:\tau}}{\tau}{\Pmany{\Pdecl}~\Pstmt}}{closures}
\Psyncase{\Pnative{x}{\Pmany{d~x:\tau}}{\tau}}{built-in functions}
\Psyncase{\Ptableval{\ell}{\Penv}{\Pmany{\Pkey}}{\Pmany{\Pact}}}{table values}
\Psyncase{\Pcclos{\Penv}{\Pmany{d~x:\tau}}{\Pmany{x_c:\tau_c}}{\Pmany{\Pdecl}}{\Pstmt}}{constructor closures}
\\
\Psyndef{\Psig}{\Pcont}{continue normally}
\Psyncase{\Psigreturn{\Pval}}{return value}
\Psyncase{\Psigexit}{exit/reject all enclosing calls}
\end{grammar}
\)
\caption{\pcore values and signals. A value, naturally, is the result of evaluating an expression. A signal is the result of evaluating a statement or declaration.}
\label{fig:valuesyntax}
\end{figure}

\subsubsection{Values and signals (\cref{fig:valuesyntax})}

Record values are standard. A header value augments a record with a
validity tag, marking whether the header has been initialized. When
parsing a packet into a header, a $\kw{valid}$ tag is added if it does
not already exist. Native functions are available to check for,
remove, or add a tag. Header and stack values include their field and
element types to facilitate our treatment of undefined reads
(see \cref{sec:expreval}).

A single closure construct is used to represent function closures and
constructed controls, so a closure can contain declarations. A closure
includes an environment, but not a store, so that closure calls see
updates to mutable variables that were in scope when the closure was
created.

Native functions are provided by the architecture in the initial
program environment. They always include common operations for
manipulating header validity bits and the like, but may also include
architecture-specific functionality, for example, hash functions.
\begin{align*}
\Pbittype{16}~\kw{hash\_crc16}\langle\kw{T}\rangle(\kw{T}\ \kw{data});
\end{align*}

Table closures include an environment for evaluating key expressions
and a list of actions. They include a location $\ell$ used as an
identifier for the control plane to disambiguate between different
instances of the same table declaration.

Signals are used to encode normal and exceptional control-flow:
continuing normally, returning a value, or exiting.

\begin{figure}\footnotesize
\[\begin{array}{rcll}
\Psyndef{\Gamma}{\Gamma,\ x_1:\tau_1}{typing context}
\Psyncase{\Gamma,\ X_1:\tau_1}{constructor type}
\Psyncase{[]}{}
\Psyndef{\Delta}{\Delta,\ X_1~\kw{var}}{type variable and definition context}
\Psyncase{\Delta,\ X_1=\tau_1}{type definition}
\Psyncase{[]}{}
\Sigma&:&\kw{Var}\to\kw{Value}&\text{constant context} \\
\Pstore&:&\kw{Loc} \to \kw{Value}&\text{store} \\
\Penv&:&\kw{Var} \to\kw{Loc} &\text{environment} \\
\Xi&:&\kw{Loc} \to\kw{Type}&\text{store typing context}\\
\mathcal{C}&:&\kw{Loc} \times \kw{Value} \times \Pmany{\kw{PartialActRef}} \to\kw{ActRef} &\text{control plane} \\
\end{array}\]
\caption{Evaluation and typechecking contexts and environments.
All function spaces in this figure are restricted to finite partial
maps. Stores associate values with locations. Evaluation environments
associate locations with variables. A $\kw{PartialActRef}$ is a
function call expression with missing parameters, while an
$\kw{ActRef}$ is an ordinary function call expression.}
\label{fig:evalenvs}
\end{figure}

\subsubsection{Typing and evaluation contexts (\cref{fig:evalenvs})}

There are four kinds of context used in typechecking \pcore programs:
typing contexts $\Gamma$, type definition contexts $\Delta$, store
typing contexts $\Xi$, and constant contexts $\Sigma$. Typing
contexts are lists of bindings, giving types to variable names and
type names $X$. In particular, if $X$ is a type with a
constructor, the type of the constructor will be recorded in $\Gamma$
under the name $X$. Type definition contexts include freely mixed
definitions $X=\tau$ and variable markers
$X~\kw{var}$. Store typings are finite partial maps from locations to
types. Constant contexts are finite partial maps from variable names
to compile-time values.

\subsection{Static semantics}
\label{sec:static-semantics}

The static semantics for \pcore takes care of copy-in copy-out
typechecking, compile-time computation in types, generics, type
definitions, casts, open enumerations, and extern (native) functions.
Surface concerns like type argument inference and implicit cast
insertion are handled in the \petra interpreter but omitted here
(see \Cref{sec:implementation} for details).

Typing judgments are given in \cref{fig:staticjudgments}. The first
three judgments are the top-level program typing judgments. Store,
environment, and value typing are not used to typecheck programs but
are necessary in order to formulate our type safety theorem. The type
simplification judgment replaces type variables in $\tau$ with their
definitions in $\Delta$ and performs compile-time evaluation on any
expressions that appear in $\tau$. The compile-time evaluation
judgment only needs a constant environment and an expression.

\begin{figure}\footnotesize
\[
\begin{array}{ll}
\begin{array}{ll}
\Pjudgment{\Sigma,\Gamma,\Delta \vdash \Pexp:\tau\mathrel{\kw{goes}}d}
& \text{Expression typing}
\\

\Pjudgment{\vphantom{p}\Sigma,\Gamma,\Delta \vdash \Pstmt \dashv \Sigma',\Gamma'}
& \text{Statement typing}
\\

\Pjudgment{\vphantom{p}\Sigma,\Gamma,\Delta \vdash \Pdecl \dashv \Sigma',\Gamma',\Delta'}
& \text{Declaration typing}
\\

\Pjudgment{\vphantom{p}\Xi,\Sigma,\Delta \vdash \Pstore}
& \text{Store typing}
\\
\Pjudgment{\Pcteval{\Sigma, \Pexp}{v}}
& \text{Compile-time evaluation}
\end{array} &
\begin{array}{ll}
\Pjudgment{\vphantom{p} \Xi,\Delta \vdash \Penv:\Gamma}
& \text{Environment typing}
\\

\Pjudgment{\vphantom{\Gamma p}\Xi,\Sigma,\Delta \vdash \Pval:\tau}
& \text{Value typing}
\\

\Pjudgment{\vphantom{p}\Delta \vdash \rho \preceq \rho'}
& \text{Legal casts}
\\
\Pjudgment{\Ptypeeval{\Sigma, \Delta}{\tau}{\tau'}}
& \text{Type simplification}\\[\medskipamount]
\end{array}
\end{array}
\]
\caption{Selected judgment signatures from the static semantics.}\label{fig:staticjudgments}
\end{figure}

The expression typing judgment produces a direction
indicating whether the expression is assignable ($\kw{goes}~\Pinout$)
or not ($\kw{goes}~\Pin$.) Sometimes we need the type of an expression
but do not care about its direction. In such a situation the
expression typing judgment may be written $\Sigma,\Gamma,\Delta \vdash \Pexp:\tau$, 
leaving off the direction annotation $\kw{goes}~d$.

Statement typechecking produces a new constant context and a new
typing context. Declaration typechecking produces new constant and
typing contexts, as for statements, but it also produces an updated
type variable context to hold any new type definitions.

Our type soundness proof assumes that function and control bodies
always return a value. In the implementation, a simple static analysis
integrated into statement typechecking ensures that this is the case.
We omit it here in order to avoid cluttering up the typing rules.

The type simplification judgment replaces type variables with their
definitions in $\Delta$ and evaluates expressions occurring in
types. Here is an example of it substituting a definition for the type
variable $C$, including recursive substitutions for the type variable
$B$ and the expression $c+1$.
\[
c:=7,~B=\Pbool,~C=\Psimplefunctiontype{\Pin~x:\Pbittype{c+1}}{B} \quad \vdash \quad
          C
          \rightsquigarrow
          (\Psimplefunctiontype{\Pin~x:\Pbittype{8}}{\Pbool})
\]

\subsubsection{Expression typing (\cref{fig:staticjudgments})}
\begin{figure}
\footnotesize
\begin{mathpar}
\Pinfer{T-Var}{
x \notin \kw{dom}(\Sigma)
\\
\Gamma(x)=\tau
}{
\Sigma,\Gamma,\Delta \vdash x:\tau \Pgoes{\Pinout}
}
\and
\Pinfer{T-Var-Const}{
x \in \kw{dom}(\Sigma) \\
\Gamma(x)=\tau
}{
\Sigma,\Gamma,\Delta \vdash x:\tau \Pgoes{\Pin}
}
\and
\Pinfer{T-Bit}{w \neq \infty}{
\Sigma,\Gamma,\Delta \vdash \Pint{n}{w}:\Pbittype{w} \Pgoes{\Pin}
}
\and
\Pinfer{T-Bool}{ }{
\Sigma,\Gamma,\Delta \vdash b:\Pbool \Pgoes{\Pin}
}
\and
\Pinfer{T-Integer}{ }{
\Sigma,\Gamma,\Delta \vdash \Pint{n}{\infty}:\Pinttype \Pgoes{\Pin}
}
\and
\Pinfer{T-Index}{
\Sigma,\Gamma,\Delta \vdash \Pexp_1:\tau[n]\Pgoes{d} \\\\
\Sigma,\Gamma,\Delta \vdash \Pexp_2:\Pbittype{32}
}{
\Sigma,\Gamma,\Delta \vdash \Pexp_1[\Pexp_2]:\tau \Pgoes{d}
}
\and
\Pinfer{T-Enum}{
\Delta(X)=\Penumtype{X}{\Pmany{f}}
}{
\Sigma,\Gamma,\Delta \vdash X.f_i: \Penumtype{X}{\Pmany{f}} \Pgoes{\Pin}
}
\and
\Pinfer{T-Err}{
\Perrortype{f} \in \Delta(\Perror) \\
f_i \in \Pmany{f}
}{
\Sigma,\Gamma,\Delta \vdash \Perror.f_i:\Perror \Pgoes{\Pin}
}
\and
\Pinfer{T-Match}{
\Pmatchkindtype{f} \in \Delta(\Pmatchkind) \\
f_i \in \Pmany{f}
}{
\Sigma,\Gamma,\Delta \vdash \Pmatchkind.f_i:\Pmatchkind \Pgoes{\Pin}
}
\and
\Pinfer{T-Cast}{
\Sigma,\Gamma,\Delta \vdash \Pexp:\rho_0\Pgoes{d} \\
\Ptypeeval{\Sigma,\Delta}{\rho}{\tau'} \\
\Delta \vdash \rho_0 \preceq \tau'
}{
\Sigma,\Gamma,\Delta \vdash \Pcast{\rho}{\Pexp}:\tau'\Pgoes{d}
}
\and
\Pinfer{T-UOp}{
\Puoptype{\Ptypedefs,\ominus}{\rho_1}{\rho_2} \\
\Sigma,\Gamma,\Delta \vdash \Pexp:\rho_1
}{
\Sigma,\Gamma,\Delta \vdash {\ominus}\Pexp:\rho_2\Pgoes{\Pin}
}
\and
\Pinfer{T-BinOp}{
\Pbinoptype{\Ptypedefs,\oplus}{\rho_1}{\rho_2}{\rho_3} \\
\Sigma,\Gamma,\Delta \vdash \Pexp_1:\rho_1 \\
\Sigma,\Gamma,\Delta \vdash \Pexp_2:\rho_2
}{
\Sigma,\Gamma,\Delta \vdash \Pexp_1\oplus\Pexp_2:\rho_3\Pgoes{\Pin}
}
\and
\Pinfer{T-MemHdr}{
\Sigma,\Gamma,\Delta \vdash \Pexp: \Pheader{\Pmany{f:\tau}}\Pgoes{d}
}{
\Sigma,\Gamma,\Delta \vdash \Pexp.f_i : \tau_i\Pgoes{d}
}
\and
\Pinfer{T-MemRec}{
\Sigma,\Gamma,\Delta \vdash \Pexp:\Prectype{\Pmany{f:\tau}}\Pgoes{d}
}{
\Sigma,\Gamma,\Delta \vdash \Pexp.f_i : \tau_i\Pgoes{d}
}
\and
\Pinfer{T-Record}{
\Sigma,\Gamma,\Delta \vdash \Pmany{\Pexp:\tau}
}{
\Sigma,\Gamma,\Delta \vdash \Prec{\Pmany{f=\Pexp}}:\Prectype{\Pmany{f:\tau}}\Pgoes{\Pin}
}
\and
\Pinfer{T-Slice}{
\Sigma,\Gamma,\Delta \vdash \Pexp_1:\Pbittype{w}\Pgoes{d} \\\\
\Sigma,\Gamma,\Delta \vdash \Pexp_2:\Pinttype \\
\Sigma,\Gamma,\Delta \vdash \Pexp_3:\Pinttype \\\\
\Pcteval{\Sigma, \Pexp_2}{n_2} \\
\Pcteval{\Sigma, \Pexp_3}{n_3} \\\\
w > n_2 \geq n_3 \geq 0
}{
\Sigma,\Gamma,\Delta \vdash \Pexp_1[\Pexp_2{:}\Pexp_3]:\Pbittype{n_2-n_3+1}\Pgoes{d}
}
\and
\Pinfer{T-Call}{
\Sigma,\Gamma,\Delta \vdash \Pexp:\Pfunctiontype{\Pmany{X}}{\Pmany{d~x:\tau}}{\tau_{\mathit{ret}}} \\\\
\Ptypeeval{\Sigma,\Delta[\Pmany{X=\rho}]}{\Pmany{\tau}}{\Pmany{\tau'}} \\
\Sigma,\Gamma,\Delta \vdash \Pmany{\Pexp:\tau'\mathrel{\kw{goes}}d} \\\\
\Ptypeeval{\Sigma,\Delta[\Pmany{X=\rho}]}{\tau_{\mathit{ret}}}{\tau_{\mathit{ret}}'}
}{
\Sigma,\Gamma,\Delta \vdash \Pgencall{\Pexp}{\Pmany{\rho}}{\Pmany{\Pexp}}:\tau_{\mathit{ret}}'\Pgoes{\Pin}
}
\end{mathpar}
\caption{Expression typing rules.}
\label{fig:expr-typing}
\end{figure}

The expression typing judgment is defined in \cref{fig:expr-typing}.
It is designed to only ever output types in a canonical form with no
unevaluated expressions and no free variables except the ones declared
with $X~\kw{var}$ in $\Delta$. This works if the contents of $\Gamma$
are also in canonical form for $\Delta$.

The typing rules for l-values (arrays, bitslices, fields) check the
direction $d$ of their ``root'' subexpression. The only rule that
produces $d \neq \Pin$ is \textsc{T-Var}, which requires $x$ to not be
in the constant context. The types of unary and binary operators are
determined by a type interpretation function $\mathcal{T}$. Array
indexes are not required to be compile time known and are not bounds
checked. By contrast, the endpoints of a bit slice receive both
treatments because the type of a slice depends on the values of its
endpoints and types should only depend on compile-time values. Bounds
checking is a bonus, since the endpoints are already evaluated.

The function call rule uses type simplification to substitute type arguments
into parameter types and return types.

\subsubsection{Statement typing (\cref{fig:stmt-typing})}
\begin{figure}
\footnotesize
\begin{mathpar}
\Pinfer{TS-Empty}{ }{
\Sigma,\Gamma,\Delta \vdash \Pbraces{} \dashv \Sigma,\Gamma
}
\and
\Pinfer{TS-Exit}{ }{
\Sigma,\Gamma,\Delta \vdash \Pexit \dashv \Sigma,\Gamma
}
\and
\Pinfer{TS-Block}{
\Sigma,\Gamma,\Delta \vdash \Pstmt \dashv \Sigma_1, \Gamma_1\\
\Sigma_1,\Gamma_1,\Delta \vdash \Pbraces{\Pmany{\Pstmt}} \dashv \Sigma_2,\Gamma_2
}{
\Sigma,\Gamma,\Delta \vdash \Pbraces{\Pstmt; \Pmany{\Pstmt}} \dashv \Sigma,\Gamma
}
\\
\Pinfer{TS-Decl}{
\Sigma_0,\Gamma_0,\Delta_0 \vdash \Pvdecl \dashv \Sigma_1,\Gamma_1,\Delta_1
}{
\Sigma_0,\Gamma_0,\Delta_0 \vdash \Pvdecl \dashv \Sigma_1, \Gamma_1
}
\and
\Pinfer{TS-Assign}{
\Sigma,\Gamma,\Delta \vdash \Pexp_1:\tau\mathrel{\kw{goes}}\kw{inout}
\\\\
\Sigma,\Gamma,\Delta \vdash \Pexp_2:\tau
}{\Sigma,\Gamma,\Delta \vdash \Passign{\Pexp_1}{\Pexp_2} \dashv \Sigma,\Gamma}
\and
\Pinfer{TS-Ret}{
\Sigma,\Gamma,\Delta \vdash \Pexp:\tau \\
\Sigma,\Delta \vdash \Gamma(\Preturnvoid) \rightsquigarrow \tau \\
}{
\Sigma,\Gamma,\Delta \vdash \Preturn{\Pexp} \dashv \Sigma,\Gamma
}
\and
\Pinfer{TS-If}{
\Sigma,\Gamma,\Delta \vdash \Pexp:\Pbool \\\\
\Sigma,\Gamma,\Delta \vdash \Pstmt_1 \dashv \Sigma_1, \Gamma_1 \\
\Sigma,\Gamma,\Delta \vdash \Pstmt_2 \dashv \Sigma_2, \Gamma_2
}{
\Sigma,\Gamma,\Delta \vdash \Pif{\Pexp}{\Pstmt_1}{\Pstmt_2} \dashv \Sigma,\Gamma
}
\and
\Pinfer{TS-TblCall}{
\Sigma,\Gamma,\Delta \vdash \Pexp:\Ptabletype
}{
\Sigma,\Gamma,\Delta \vdash \Pcall{\Pexp}{} \vdash \Sigma, \Gamma
}
\and
\Pinfer{TS-Call}{
\Sigma,\Gamma,\Delta \vdash \Pgencall{\Pexp}{\Pmany{\rho}}{\Pmany{\Pexp}}:\tau
}{\Sigma,\Gamma,\Delta \vdash
\Pgencall{\Pexp}{\Pmany{\rho}}{\Pmany{\Pexp}}
\dashv \Sigma, \Gamma}
\end{mathpar}
\caption{Statement typing rules.}
\label{fig:stmt-typing}
\end{figure}

The typing rules for statements, defined in \cref{fig:stmt-typing},
are largely standard. The relation
$\Sigma,\Gamma,\Delta \vdash \Pstmt \dashv \Sigma,\Gamma$ holds when a
statement executed in the contexts on the left side will produce a
final state satisfying the contexts on the right side. The constant
context $\Sigma$ appears on the right because constants can be
declared in statements, while $\Gamma$ appears because variables can
be declared in statements.

The assignment rule \textsc{TS-Assign} checks that the expression on
the left side has direction \Pinout, which means (as we saw in the
expression typing rules) that it is an l-value. The return
rule \textsc{TS-Ret} checks that the type of the value being returned
agrees with the type of the special identifier
$\kw{return}$. Declaration typing rules
(see \cref{fig:object-decl-typing}) insert a type for \kw{return}
before typechecking the bodies of functions and controls.

\subsubsection{Variable declaration rules (\cref{fig:var-typing})}

Variable declarations introduce new variables and can be used as
statements. Their typing relation includes an output type context for
uniformity with other declarations but they do not bind new types.

\subsubsection{Object declaration rules (\cref{fig:object-decl-typing})}
Table typechecking checks keys and match kinds.  An action
$\mathit{act}$ is a partial application of a function, so the
auxiliary judgment \kw{act\_ok} checks the action like a function call
but allows any number of arguments to be left off. Omitted arguments
are the responsibility of the control plane.

The typing rules for controls and functions use a special \kw{return}
identifier to check return statements within the body of the
declaration. The $\{\}$ is an empty record type standing in
for what P4 calls \lstinline{void}.

\begin{figure}
\footnotesize
\begin{mathpar}
\Pinfer{Type-Const}{
\Ptypeeval{\Sigma,\Delta}{\tau}{\tau'} \\\\
\Sigma,\Gamma,\Delta \vdash \Pexp:\tau'\\
\Pcteval{\Sigma, \Pexp}{v}
}{
\Sigma,\Gamma,\Delta \vdash \Pconst{\tau}{x}{\Pexp} \dashv \Sigma[x=v],\Gamma[x:\tau'],\Delta
}
\and
\Pinfer{Type-Var}{
\Ptypeeval{\Sigma,\Delta}{\tau}{\tau'}
}{
\Sigma,\Gamma,\Delta \vdash
\Pvardecl{\tau}{x} \dashv
\Sigma,\Gamma[x:\tau'],\Delta
}
\\
\Pinfer{Type-VarInit}{
\Ptypeeval{\Sigma,\Delta}{\tau}{\tau'} \\
\Sigma,\Gamma,\Delta \vdash \Pexp:\tau'
}{
\Sigma,\Gamma,\Delta \vdash
\Pvardeclinst{\tau}{x}{\Pexp} \dashv
\Sigma,\Gamma[x:\tau'],\Delta
}
\and
\Pinfer{Type-Inst}{
\Sigma,\Gamma,\Delta \vdash C:\Pctortype{\Pmany{x:\tau}}{\tau_{\mathit{inst}}} \\
\Sigma,\Gamma,\Delta \vdash \Pmany{\Pexp:\tau}
}{
\Sigma,\Gamma,\Delta \vdash
\Pinst{X}{\Pmany{\Pexp}}{x} \dashv
\Sigma,\Gamma[x:\tau_{\mathit{inst}}],\Delta
}
\end{mathpar}
\caption{Variable declaration typing rules.}
\label{fig:var-typing}
\end{figure}

\begin{figure}
\footnotesize
\begin{mathpar}
\Pinfer{T-TableDecl}{
\Sigma,\Gamma,\Delta \vdash \Pmany{\Pexp_k : \tau_k} \\
\Sigma,\Gamma,\Delta \vdash \Pmany{x_k:\Pmatchkind} \\
\Sigma,\Gamma,\Delta \vdash \Pmany{act~\kw{act\_ok}}
 }{
\Sigma,\Gamma,\Delta \vdash
\Ptable{x}{\Pmany{\Pexp_k:x_k}~\Pmany{\Pact}} \dashv
\Sigma,\Gamma[x : \Ptabletype],\Delta
}
\and
\Pinfer{T-CtrlDecl}{
\Ptypeeval{\Sigma,\Delta}{\Pmany{\tau_c}}{\Pmany{\tau_c'}} \\
\Ptypeeval{\Sigma,\Delta}{\Pmany{\tau}}{\Pmany{\tau'}} \\
\Sigma,\Gamma[\Pmany{x_c:\tau_c}][\Pmany{x:\tau}],\Delta \vdash
\Pmany{\Pdecl} \dashv \Sigma_1, \Gamma_1, \Delta_1 \\
\Sigma_1, \Gamma_1[\Preturnvoid:\{\}], \Delta_1 \vdash \Pstmt \dashv \Sigma_2, \Gamma_2
}{
\Sigma,\Gamma,\Delta \vdash
\Pcontroldecl{X}{\Pmany{d~x:\tau'}}{\Pmany{x_c:\tau_c'}}{\Pmany{\Pdecl}}{\Pstmt} \dashv
\Sigma,\Gamma[X : \Pctortype{\Pmany{x_c:\tau_c'}}{
\Psimplefunctiontype{\Pmany{d~x:\tau'}}{\{\}}
}],\Delta
}
\and
\Pinfer{T-FuncDecl}{
\Gamma_1 = \Gamma[\Pmany{x_i:\tau_i'}, \kw{return}: \tau'] \\
\Delta_1 = \Delta[\Pmany{X~\kw{var}}] \\
\Ptypeeval{\Sigma, \Delta_1}{\Pmany{\tau_i}}{\Pmany{\tau_i'}} \\\\
\Ptypeeval{\Sigma, \Delta_1}{\tau}{\tau'} \\
\Sigma, \Gamma_1, \Delta_1 \vdash \Pstmt \dashv \Sigma_2, \Gamma_2
}{
\Sigma,\Gamma,\Delta \vdash
\Pfunctiondecl{\tau}{x}{\Pmany{X}}{\Pmany{d~x_i:\tau_i}}{\Pstmt} \dashv
\Sigma,\Gamma[x : \Pfunctiontype{\Pmany{X}}{\Pmany{d~x_i:\tau_i'}}{\tau'}],\Delta
}
\end{mathpar}
\caption{Object declaration typing rules.}
\label{fig:object-decl-typing}
\end{figure}

\subsection{Dynamic semantics}

The dynamic semantics for \pcore is defined in a big-step style.
\Cref{fig:dynamic-judgments} gives the types of the main judgments. 
Local state is split into a store and an environment to implement the
scoping of mutable variables. The environment maps names of variables
to store locations, and the store maps locations to values. This
decoupling allows closures to witness updates to mutable variables
saved in their environments.

Morally speaking, P4 programs are deterministic. The semantics of
\pcore does introduce nondeterminism in a few places to simplify the
presentation or to model architecture-dependent behavior. For example,
the result of reading an invalid header is an undefined value, which
may vary from target to target and even from read to read within a
single program. In the semantics, we write $\kw{havoc}(\tau)$ to
indicate an operation producing an arbitrary value of type $\tau$.
Match-action evaluation uses the control plane $\mathcal{C}$ to select
from the table's actions (rather than defining an algorithm for
selecting the matching entry from a list of forwarding rules). As
mentioned previously, we give tables unique identifiers for control
plane use by reusing locations $\ell$, which are also generated
non-deterministically, although this is not essential.

Statements evaluate to signals, which indicate how control flow should
proceed. Expressions evaluate to signals as well but with values
$\Pval$ in place of the $\Pcont$ signal. The signals are how \pcore
models non-standard control flow. To save space, we elide the
``unwinding'' rules for handling signals other than $\Pcont$ or
$\Pval$ in most places. For each intermediate computation with outputs
$\Pstore$ and $\Penv$ if that computation terminates in \Pexit\xspace
or $\Preturn{\Pval}$, the overall computation freezes the state at
$\langle \Pstore, \Penv\rangle$ and propagates the signal.

\begin{figure}\footnotesize
\centering
\(
\centering
\begin{array}{ll}
\Pjudgment{
\langle \Delta, \sigma,\epsilon,\tau \rangle \Downarrow_{\tau} \tau'
}
& \text{Type simplification}
\\[\medskipamount]

\Pjudgment{
\langle\Pcp,\Ptypedefs,\Pstore,\Penv,\Pmany{d~x:\tau:=\Pexp}\rangle \Downarrow_{\mathit{copy}}
\langle
\Pstore',
\Pmany{\Penvset{x}{\ell}},
\Pmany{\Passign{\Plval}{\ell}}
\rangle
}
& \text{Copy-in copy-out}
\\[\medskipamount]

\Pjudgment{
\Passignlval{\Pcp,\Delta,\Pstore}{\Penv}{\Passign{\Plval}{\Pval}}{\Pstore'}
}
& \text{L-value assignment}
\\[\medskipamount]

\Pjudgment{
\Plvaleval{\Pcp}{\Ptypedefs}{\Pstore}{\Penv}{\Pexp}
          {\Pstore'}{\Plval}
}
& \text{L-value evaluation}
\\[\medskipamount]

\Pjudgment{
\Pexpreval{\Pcp}{\Ptypedefs}{\Pstore}{\Penv}{\Pexp}
          {\Pstore'}{\Pval}
}
& \text{Expression evaluation}
\\[\medskipamount]

\Pjudgment{
\langle\Pcp,x,\Pmany{\Pval:x}\rangle \Downarrow_{\mathit{match}}
x(\Pmany{\Pexp})}
& \text{Match-action evaluation}
\\[\medskipamount]

\Pjudgment{
\Peval{\Pcp,\Ptypedefs,\Pstore,\Penv,\Pstmt}
{\Pstore',\Penv',\Psig}}
& \text{Statement evaluation}
\\[\medskipamount]

\Pjudgment{
\Pstmtconf{\Pcp,\Ptypedefs}{\Pstore,\Penv}{\Pdecl}
\Downarrow
\Pstmtconf{\Ptypedefs'}{\Pstore',\Penv'}{\Psig}
}
& \text{Declaration evaluation}
\end{array}
\)
\caption{Selected judgment signatures from the dynamic semantics.}
\label{fig:dynamic-judgments}
\end{figure}

\subsubsection{Copy-in copy-out rules (\cref{fig:copy})}

This example shows how copy-in copy-out handles aliasing of function
arguments. The $\{\}$ before $f$ is its return type, a record with no
fields (i.e., unit).
\begin{align*}
&\Psimplefunctiondecl{\{\}}{f}{\Pinout~\Pbittype{8}~\mathit{src},
\Pinout~\Pbittype{8}~\mathit{dst}}{
  \Passign{dst}{src+1};\ 
  \Passign{src}{0}
}\\
&\Passign{x}{1};
\\
&f(x,x);
\end{align*}
In a call-by-reference language $x$ would be $0$ after the call to
$f$. In a call-by-value language, it would still be $1$. In P4,
however, $x$ will be $2$. A function call creates temporaries for
storing its arguments for each call and copies the temporaries back,
in order, after the body of the function finishes. In the example
$dst$ comes last in the parameter list of $f$, so $x$ ends up with the
$dst$ value ($2$) overwriting the $src$ value ($0$).

The calling convention guarantees that distinct variable names within
a function refer to distinct storage locations. This means P4
compilers and static analysies never have to account for aliasing.

\begin{figure}
\footnotesize
\begin{mathpar}
\Pinfer{CopyIn}{
\Pstmtconf{\Pcp,\Ptypedefs}{\Pstore,\Penv}{\Pexp}
\Downarrow
\langle \Pstore',\Pval \rangle \\
\Pfresh{\ell}
}{
\langle\Pcp,\Ptypedefs,\Pstore,\Penv,\kw{in}~x:\tau:=\Pexp\rangle \Downarrow_{\mathit{copy}}
\langle
\Pstore'[\Pheapset{\ell}{\Pval}],
\Penvset{x}{\ell},
[]
\rangle
}
\\
\Pinfer{CopyOut}{
\Pstmtconf{\Pcp,\Ptypedefs}{\Pstore,\Penv}{\Pexp}
\Downarrow_{\mathit{\Plval}}
\langle
\Pstore',\Plval
\rangle \\
\Pfresh{\ell}
}{
\langle\Pcp,\Ptypedefs,\Pstore,\Penv,\kw{out}~x:\tau:=\Pexp\rangle \Downarrow_{\mathit{copy}}
\langle
\Pstore[\Pheapset{\ell}{\kw{init}_\Ptypedefs~\tau}],
\Penvset{x}{\ell},
[\Passign{\Plval}{\ell}]
\rangle
}
\\
\Pinfer{CopyInOut}{
\Pstmtconf{\Pcp,\Ptypedefs}{\Pstore,\Penv}{\Pexp}
\Downarrow_{\mathit{\Plval}}
\langle
\Pstore_1,\Plval
\rangle\\
\Pstmtconf{\Ptypedefs}{\Pstore_1,\Penv}{\Plval}
\Downarrow
\langle
\Pstore_2,
\Pval
\rangle\\
\Pfresh{\ell}
}{
\langle\Pcp,\Ptypedefs,\Pstore,\Penv,\kw{inout}~x:\tau:=\Pexp\rangle \Downarrow_{\mathit{copy}}
\langle
\Pstore_2[\Pheapset{\ell}{\Pval}],
\Penvset{x}{\ell},
[\Passign{\Plval}{\ell}]
\rangle
}
\end{mathpar}
\caption{Copy-in and copy-out operations. We define them for single arguments and they are lifted to lists of arguments in the obvious way.}
\label{fig:copy}
\end{figure}

\subsubsection{Expression evaluation (\cref{fig:expression-eval-i,fig:expression-eval-ii}).}
\label{sec:expreval}

Unary operations, binary operations, and casts are axiomatized. Rather
than spell out all the legal casts or arithmetic expressions, we
assume we have typing and evaluation oracles for each of them which
agree. For unary and binary operations, this means that there is a
typing function $\mathcal{T}$ and an evaluation function
$\mathcal{E}$. For casts, this means there is agreement between a
casting check $\Delta \vdash \tau \preceq \tau'$ and a casting
function $\Pcasteval{\Sigma}{\Pval}{\tau}$.

The P4 specification allows programs to produce ``undefined values''
in certain situations. This is substantially more restrictive than the
concept of ``undefined behavior'' in C, which has
notoriously confusing semantics~\cite{wang2012undefined}.
Our \textsc{E-HdrMemUnref} rule introduces an undefined (havoc'd)
value when a program attempts to read from an invalid header, but does not affect any other program state.

The full P4 expression language includes built-in functions for
operations such as accessing header validity bits. \pcore models these
functions using native functions, which we assume are already in the
context at the start of program execution and which are evaluated by
appealing to an interpretation $\mathcal{N}$.

\begin{figure}
\footnotesize
\begin{mathpar}
\Pinfer{E-Int}{ }{
\Pexpreval{\Pcp}{\Ptypedefs}{\Pstore}{\Penv}{\Pint{n}{w}}
          {\Pstore}{\Pint{n}{w}}
}
\and
\Pinfer{E-Bool}{ }{
\Pexpreval{\Pcp}{\Ptypedefs}{\Pstore}{\Penv}{b}
          {\Pstore}{b}
}
\and
\Pinfer{E-TypMem}{
}{
\Pexpreval{\Pcp}{\Ptypedefs}{\Pstore}{\Penv}{X.f}
          {\Pstore}{X.f}
}
\\
\Pinfer{E-Var}{
\Penv(x)=\ell \and
\Pstore(\ell)=\Pval
}{
\Pexpreval{\Pcp}{\Ptypedefs}{\Pstore}{\Penv}{x}
          {\Pstore}{\Pval}
}
\and
\Pinfer{E-Cast}{
\Pexpreval{\Pcp}{\Ptypedefs}{\Pstore}{\Penv}{\Pexp}
          {\Pstore'}{\Pval} \\\\
\langle \Delta,\Pstore,\Penv,\tau\rangle \Downarrow_{\tau} \tau'
}{
\Pexpreval{\Pcp}{\Ptypedefs}{\Pstore}{\Penv}{(\tau)\Pexp}
          {\Pstore'}{\Pcasteval{\Delta}{\Pval}{\tau'}}
}
\\
\Pinfer{E-Uop}{
\Pexpreval{\Pcp}{\Ptypedefs}{\Pstore}{\Penv}{\Pexp}
          {\Pstore'}{\Pval} \\
}{
\Pexpreval{\Pcp}{\Ptypedefs}{\Pstore}{\Penv}{{\ominus}\Pexp}
          {\Pstore'}{\Puopeval{\ominus}{\Pval}}
}
\and
\Pinfer{E-BinOp}{
\Pexpreval{\Pcp}{\Ptypedefs}{\Pstore}{\Penv}{\Pexp_1}
          {\Pstore_1}{\Pval_1} \\\\
\Pexpreval{\Pcp}{\Ptypedefs}{\Pstore_1}{\Penv}{\Pexp_2}
          {\Pstore_2}{\Pval_2}
}{
\Pexpreval{\Pcp}{\Ptypedefs}{\Pstore}{\Penv}{\Pexp_1\oplus\Pexp_2}
          {\Pstore_2}{\Pbinopeval{\oplus}{\Pval_1}{\Pval_2}}
}
\\
\Pinfer{E-Rec}{
\Pexpreval{\Pcp}{\Ptypedefs}{\Pstore}{\Penv}{\Pmany{\Pexp}}
          {\Pstore'}{\Pmany{\Pval}} \\
}{
\Pexpreval{\Pcp}{\Ptypedefs}{\Pstore}{\Penv}{\Prec{\Pmany{f=\Pexp}}}
          {\Pstore'}{\Prec{\Pmany{f=\Pval}}}
}
\and
\Pinfer{E-RecMem}{
\Pexpreval{\Pcp}{\Ptypedefs}{\Pstore}{\Penv}{\Pexp}
          {\Pstore'}{\Prec{\Pmany{f:\tau=\Pval}}}
}{
\Pexpreval{\Pcp}{\Ptypedefs}{\Pstore}{\Penv}{\Pexp.f_i}
          {\Pstore'}{\Pval_i}
}
\\
\Pinfer{E-Slice}{
\Pexpreval{\Pcp}{\Ptypedefs}{\Pstore}{\Penv}{\Pexp_1}
          {\Pstore_1}{\Pint{n}{w}} \\\\
\Pexpreval{\Pcp}{\Ptypedefs}{\Pstore_1}{\Penv}{\Pexp_2}
          {\Pstore_2}{\Pint{p}{\infty}} \\\\
\Pexpreval{\Pcp}{\Ptypedefs}{\Pstore_2}{\Penv}{\Pexp_3}
          {\Pstore_3}{\Pint{q}{\infty}}
}{
\Pexpreval{\Pcp}{\Ptypedefs}{\Pstore}{\Penv}{\Pexp_1[\Pexp_2{:}\Pexp_3]}
          {\Pstore_3}{\Pint{n}{w}[p{:}q]}
}
\and
\Pinfer{E-Index}{
\Pexpreval{\Pcp}{\Ptypedefs}{\Pstore}{\Penv}{\Pexp_1}
          {\Pstore_1}{\Pstackval{\tau}{\Pmany{\Pval}}} \\\\
\Pexpreval{\Pcp}{\Ptypedefs}{\Pstore_1}{\Penv}{\Pexp_2}
          {\Pstore_2}{\Pint{n}{32}} \\
0 \leq n < \kw{len}(\Pmany{\Pval})
}{
\Pexpreval{\Pcp}{\Ptypedefs}{\Pstore}{\Penv}{\Pexp_1[\Pexp_2]}
          {\Pstore_2}{\Pval_n}
}
\\
\Pinfer{E-IndexOOB}{
\Pexpreval{\Pcp}{\Ptypedefs}{\Pstore}{\Penv}{\Pexp_1}
          {\Pstore_1}{\Pstackval{\tau}{\Pmany{\Pval}}} \\\\
\Pexpreval{\Pcp}{\Ptypedefs}{\Pstore_1}{\Penv}{\Pexp_2}
          {\Pstore_2}{\Pint{n}{32}} \\
n \geq \kw{len}(\Pmany{\Pval})
}{
\Pexpreval{\Pcp}{\Ptypedefs}{\Pstore}{\Penv}{\Pexp_1[\Pexp_2]}
          {\Pstore_2}{\Phavoc{\tau}}
}
\end{mathpar}
\caption{Semantics for expressions I.}
\label{fig:expression-eval-i}
\end{figure}

\begin{figure}
\footnotesize
\begin{mathpar}
\Pinfer{E-HdrMem}{
\Pexpreval{\Pcp}{\Ptypedefs}{\Pstore}{\Penv}{\Pexp}
          {\Pstore'}{\Pheaderval{\Pmany{f:\tau=\Pval}}}
}{
\Pexpreval{\Pcp}{\Ptypedefs}{\Pstore}{\Penv}{\Pexp.f_i}
          {\Pstore'}{\Pval_i}
}
\and
\Pinfer{E-HdrMemUndef}{
\Pexpreval{\Pcp}{\Ptypedefs}{\Pstore}{\Penv}{\Pexp}
          {\Pstore'}{\Pheaderinval{\Pmany{f:\tau=\Pval}}} \\
}{
\Pexpreval{\Pcp}{\Ptypedefs}{\Pstore}{\Penv}{\Pexp.f_i}
          {\Pstore'}{\Phavoc{\tau_i}}
}
\and
\Pinfer{E-Call-DeclExit}{
\Peval{\Pcp,\Ptypedefs,\Pstore,\Penv,\Pexp}
      {\Pstore_1,\Pclos{\Penv_c}{\Pmany{X}}{\Pmany{d~x:\tau}}{\tau}{\Pmany{\Pdecl}~\Pstmt}}
\\\\
\langle \Ptypedefs[\Pmany{X=\rho}], \Pstore, \Penv, \Pmany{\tau} \rangle \Downarrow_\tau \Pmany{\tau'}
\\\\
\langle \Pcp,\Ptypedefs,\Pstore_1,\Penv,\Pmany{\Passign{d~x:\tau'}{\Pexp}}\rangle
\Downarrow_{\mathit{copy}}
\langle \Pstore_2,\Pmany{\Penvset{x}{\ell}},\Pmany{\Passign{\Plval}{\ell}}\rangle
\\\\
\langle
\Pcp,\Ptypedefs[\Pmany{X=\rho}],\Pstore_2,\Penv_c[\overline{\Penvset{x}{\ell}}],\Pmany{\Pdecl}
\rangle
\Downarrow
\langle
\Ptypedefs_2,\Pstore_3,\Penv_2,\Pexit
\rangle
\\\\
\Passignlval{\Pcp,\Delta,\Pstore_3}{\Penv}{\Pmany{\Passign{\Plval}{\Pstore_3(\ell)}}}{\Pstore_4}
}{
\Pexpreval{\Pcp}{\Ptypedefs}{\Pstore}{\Penv}{\Pgencall{\Pexp}{\Pmany{\rho}}{\Pmany{\Pexp}}}
      {\Pstore_4}{\Pexit}
}
\and
\Pinfer{E-Call-StmtExit}{
\Peval{\Pcp,\Ptypedefs,\Pstore,\Penv,\Pexp}
      {\Pstore_1,\Pclos{\Penv_c}{\Pmany{X}}{\Pmany{d~x:\tau}}{\tau}{\Pmany{\Pdecl}~\Pstmt}}
\\\\
\langle \Ptypedefs[\Pmany{X=\rho}], \Pstore, \Penv, \Pmany{\tau} \rangle \Downarrow_\tau \Pmany{\tau'}
\\\\
\langle \Pcp,\Ptypedefs,\Pstore_1,\Penv,\Pmany{\Passign{d~x:\tau}{\Pexp}}\rangle
\Downarrow_{\mathit{copy}}
\langle \Pstore_2,\Pmany{\Penvset{x}{\ell}},\Pmany{\Passign{\Plval}{\ell}}\rangle
\\\\
\langle
\Pcp,\Ptypedefs[\Pmany{X=\rho}],\Pstore_2,\Penv_c[\overline{\Penvset{x}{\ell}}],\Pmany{\Pdecl}
\rangle
\Downarrow
\langle
\Ptypedefs_2,\Pstore_3,\Penv_2,\Pcont
\rangle
\\\\
\langle
\Pcp,\Ptypedefs_2,\Pstore_3,\Penv_2,\Pstmt
\rangle
\Downarrow
\langle
\Pstore_4,
\Penv_3,
\Pexit
\rangle
\\\\
\Passignlval{\Pcp,\Delta,\Pstore_4}{\Penv}{\Pmany{\Passign{\Plval}{\Pstore_4(\ell)}}}{\Pstore_5}
}{
\Pexpreval{\Pcp}{\Ptypedefs}{\Pstore}{\Penv}{\Pgencall{\Pexp}{\Pmany{\rho}}{\Pmany{\Pexp}}}
      {\Pstore_5}{\Pexit}
}
\and
\Pinfer{E-CallN}{
\Peval{\Pcp,\Ptypedefs,\Pstore,\Penv,\Pexp}
      {\Pstore_1,\Pnative{x}{\Pmany{d~x:\tau}}{\tau}}
\\\\
\langle \Pcp,\Ptypedefs,\Pstore_1,\Penv,\Pmany{\Passign{d~x:\tau}{\Pexp}}\rangle
\Downarrow_{\mathit{copy}}
\langle \Pstore_2,\Pmany{\Penvset{x}{\ell}},\Pmany{\Passign{\Plval}{\ell}}\rangle
\\\\
\Pnativeeval{x}{\Pstore_2}{[\Pmany{\Penvset{x}{\ell}}]}{\Pstore_3}{\Pval}
\\\\
\Passignlval{\Pcp,\Delta,\Pstore_3}{\Penv}{\Pmany{\Passign{\Plval}{\Pstore_3(\ell)}}}{\Pstore_4}
}{
\Pexpreval{\Pcp}{\Ptypedefs}{\Pstore}{\Penv}{\Pcall{\Pexp}{\Pmany{\Pexp}}}
      {\Pstore_4}{\Pval}
}
\and
\Pinfer{E-Call}{
\Peval{\Pcp,\Ptypedefs,\Pstore,\Penv,\Pexp}
      {\Pstore_1,\Pclos{\Penv_c}{\Pmany{X}}{\Pmany{d~x:\tau}}{\tau}{\Pmany{\Pdecl}~\Pstmt}}
\\\\
\langle \Ptypedefs[\Pmany{X=\rho}], \Pstore, \Penv, \Pmany{\tau} \rangle \Downarrow_\tau \Pmany{\tau'}
\\\\
\langle \Pcp,\Ptypedefs,\Pstore_1,\Penv,\Pmany{\Passign{d~x:\tau'}{\Pexp}}\rangle
\Downarrow_{\mathit{copy}}
\langle \Pstore_2,\Pmany{\Penvset{x}{\ell}},\Pmany{\Passign{\Plval}{\ell}}\rangle
\\\\
\langle
\Pcp,\Ptypedefs[\Pmany{X=\rho}],\Pstore_2,\Penv_c[\overline{\Penvset{x}{\ell}}],\Pmany{\Pdecl}
\rangle
\Downarrow
\langle
\Ptypedefs_2,\Pstore_3,\Penv_2,\Pcont
\rangle
\\\\
\langle
\Pcp,\Ptypedefs_2,\Pstore_3,\Penv_2,\Pstmt
\rangle
\Downarrow
\langle
\Pstore_4,
\Penv_3,
\Preturn{\Pval}
\rangle
\\\\
\Passignlval{\Pcp,\Delta,\Pstore_4}{\Penv}{\Pmany{\Passign{\Plval}{\Pstore_4(\ell)}}}{\Pstore_5}
}{
\Pexpreval{\Pcp}{\Ptypedefs}{\Pstore}{\Penv}{\Pgencall{\Pexp}{\Pmany{\rho}}{\Pmany{\Pexp}}}
      {\Pstore_5}{\Pval}
}
\end{mathpar}
\caption{Semantics for expressions II.}
\label{fig:expression-eval-ii}
\end{figure}

\subsubsection{Variable declaration evaluation (\cref{fig:var-decl-semantics})}

The next collection of formal rules handles variable declarations.
Constants and regular values are not distinguished at run time. The
most interesting rule is \textsc{E-Inst} for instantiations. It
produces a closure without executing any additional code, saving the
constructor arguments in the store and placing pointers to those
arguments in the closure's environment.

\subsubsection{Object declaration evaluation (\cref{fig:object-decl-semantics}}
The object declarations create closures from declarations of tables,
controls, and functions. All closures save a copy of the environment,
but do not save a copy of the store. Control and function closures are
standard.  Table closures save the fresh location of the table for use
by the control plane in disambiguating multiple tables instantiated
from a single declaration. Table closures also save the table key
expressions and the the list of actions available to the table for use
in matching.

\begin{figure}
\begin{mathpar}\footnotesize
\Pinfer{E-Const}{
\Peval{\Pcp,\Ptypedefs,\Pstore,\Penv,\Pvardeclinst{\tau}{x}{\Pexp}}
      {\Ptypedefs,\Pstore_1,\Penv_1,\Pok}
}{
\Peval{\Pcp,\Ptypedefs,\Pstore,\Penv,\Pconst{\tau}{x}{\Pexp}}
      {\Ptypedefs,\Pstore_1,\Penv_1,\Pok}
}
\and
\Pinfer{E-VarDecl}{
\Pfresh{\ell} \\
\langle \Delta,\Pstore,\Penv,\tau \rangle \Downarrow \tau'
}{
\Peval{\Pcp,\Ptypedefs,\Pstore,\Penv,\Pvardecl{\tau}{x}}
      {\Ptypedefs,\Pstore[\Passign{\ell}{\kw{init}_\Ptypedefs~\tau'}],\Penv[\Penvset{x}{\ell}],\Pok}
}
\and
\Pinfer{E-VarInit}{
\Pfresh{\ell} \\
\Peval{\Pcp,\Ptypedefs,\Pstore,\Penv,\Pexp}
      {\Pstore_1,\Pval}
}{
\Peval{\Pcp,\Ptypedefs,\Pstore,\Penv,\Pvardeclinst{\tau}{x}{\Pexp}}
      {\Ptypedefs,\Pstore_1[\Passign{\ell}{\Pval}],\Penv[\Penvset{x}{\ell}],\Pok}
}
\and
\Pinfer{E-Inst}{
\Peval{\Pcp,\Ptypedefs,\Pstore,\Penv,X}
      {\Pstore_1,\Pcclos{\Penv_{\mathit{cc}}}{\Pmany{d~x:\tau}}{\Pmany{x_c:\tau_c}}{\Pmany{\Pdecl}}{\Pstmt}}
\\\\
\Peval{\Pcp,\Ptypedefs,\Pstore_1,\Penv,\Pmany{\Pexp}}
      {\Pstore_2,\Pmany{\Pval_c}}
\\\\
\overline{\ell_c},\Pfresh{\ell}
\\
\Pval = \Pclos{\Penv_{\mathit{cc}}[\Pmany{\Penvset{x_c}{\ell_c}}]}{\langle \rangle}{\Pmany{d~x:\tau}}{\{\}}{\Pmany{\Pdecl}~\Pstmt}
}{
\Peval{\Pcp,\Ptypedefs,\Pstore,\Penv,\Pinst{X}{\Pmany{\Pexp}}{x}}
      {\Ptypedefs,\Pstore_2[\Pmany{\Pheapset{\ell_c}{\Pval_c}}][\Pheapset{\ell}{\Pval}],\Penv[\Penvset{x}{\ell}],\Pok}
}
\end{mathpar}
\caption{Semantics for variable declarations.}
\label{fig:var-decl-semantics}
\end{figure}

\begin{figure}
\footnotesize
\begin{mathpar}
\Pinfer{E-TableDecl}{
\Pfresh{\ell} \\
\Pval = \Ptableval{\ell}{\Penv}{\Pmany{\Pkey}}{\Pmany{\Pact}} \\
}{
\Peval{\Pcp,\Ptypedefs,\Pstore,\Penv,\Ptable{x}{\Pmany{\Pkey}~\Pmany{\Pact}}}
      {\Ptypedefs,\Pstore[\Pheapset{\ell}{\Pval}],\Penv[\Penvset{x}{\ell}],\Pok}
}
\and
\Pinfer{E-CtrlDecl}{
\Pfresh{\ell} \\
\langle\Delta,\Pstore,\Penv, \Pmany{\tau_c}\rangle \Downarrow_\tau \Pmany{\tau_c'} \\
\langle\Delta,\Pstore,\Penv, \Pmany{\tau}\rangle \Downarrow_\tau \Pmany{\tau'} \\
\Pval = \Pcclos{\Penv}{\Pmany{d~x:\tau'}}{\Pmany{x_c:\tau_c'}}{\Pmany{\Pdecl}}{\Pstmt} \\
}{
\Peval{\Pcp,\Ptypedefs,\Pstore,\Penv,\Pcontroldecl{X}{\Pmany{d~x:\tau}}{\Pmany{x_c:\tau_c}}{\Pmany{\Pdecl}}{\Pstmt}}
      {\Ptypedefs,\Pstore[\Pheapset{\ell}{\Pval}],\Penv[\Penvset{X}{\ell}],\Pok}
}
\and
\Pinfer{E-FuncDecl}{
\Pfresh{\ell} \\
\langle\Delta[\Pmany{X~\kw{var}}],\Pstore,\Penv, \Pmany{\tau_i}\rangle \Downarrow_\tau \Pmany{\tau_i'} \\
\langle\Delta[\Pmany{X~\kw{var}}],\Pstore,\Penv, \tau\rangle \Downarrow_\tau \tau' \\
\Pval = \Pclos{\Penv}{\Pmany{X}}{\Pmany{d~x_i:\tau_i'}}{\tau'}{\Pstmt} \\
}{
\Peval{\Pcp,\Ptypedefs,\Pstore,\Penv,\Pfunctiondecl{\tau}{x}{\Pmany{X}}{\Pmany{d~x_i:\tau_i}}{\Pstmt}}
      {\Ptypedefs,\Pstore[\Pheapset{\ell}{\Pval}],\Penv[\Penvset{x}{\ell}],\Pok}
}
\end{mathpar}
\caption{Semantics for object declarations.}
\label{fig:object-decl-semantics}
\end{figure}

\subsubsection{Statement evaluation (\cref{fig:stmt-semantics}).}

The rules for statement evaluation are mostly standard. The most
interesting rule is \textsc{E-Call-Table}, which handles table
invocation. It first evaluates the key and then uses the control-plane
to locate a matching action, and executes the body of the action to
obtain the final result. For simplicity, in \pcore, we assume that
tables have a default action, so they cannot ``miss.''

\begin{figure}
\footnotesize
\begin{mathpar}
\Pinfer{E-Empty}{
}{
\Peval{\Pcp,\Ptypedefs,\Pstore,\Penv,\Pbraces{}}
      {\Pstore,\Penv,\Pcont}
}
\and
\Pinfer{E-Exit}{ }{
\Peval{\Pcp,\Ptypedefs,\Pstore,\Penv,\Pexit}
      {\Pstore,\Penv,\Psigexit}
}
\\
\Pinfer{E-IfT}{
\Peval{\Pcp,\Ptypedefs,\Pstore,\Penv,\Pexp}
      {\Pstore_1,\Ptrue}
\\\\
\Peval{\Pcp,\Ptypedefs,\Pstore_1,\Penv,\Pstmt_1}
      {\Pstore_2,\Penv_2,\Psig}
}{
\Peval{\Pcp,\Ptypedefs,\Pstore,\Penv,\Pif{\Pexp}{\Pstmt_1}{\Pstmt_2}}
      {\Pstore_2,\Penv,\Psig}
}
\and
\Pinfer{E-IfF}{
\Peval{\Pcp,\Ptypedefs,\Pstore,\Penv,\Pexp}
      {\Pstore_1,\Pfalse}
\\\\
\Peval{\Pcp,\Ptypedefs,\Pstore_1,\Penv,\Pstmt_2}
      {\Pstore_2,\Penv_2,\Psig}
}{
\Peval{\Pcp,\Ptypedefs,\Pstore,\Penv,\Pif{\Pexp}{\Pstmt_1}{\Pstmt_2}}
      {\Pstore_2,\Penv,\Psig}
}
\\
\Pinfer{E-Block}{
\Peval{\Pcp,\Ptypedefs,\Pstore,\Penv,\Pstmt}
      {\Pstore_1,\Penv_1,\Pcont}
\\\\
\Peval{\Pcp,\Ptypedefs,\Pstore_1,\Penv_1,\Pbraces{\Pmany{\Pstmt}}}
      {\Pstore_2,\Penv_2,\Psig}
}{
\Peval{\Pcp,\Ptypedefs,\Pstore,\Penv,\Pbraces{\Pstmt,\Pmany{\Pstmt}}}
      {\Pstore_2,\Penv,\Psig}
}
\and
\Pinfer{E-Return}{
\Peval{\Pcp,\Ptypedefs,\Pstore,\Penv,\Pexp}
      {\Pstore_1,\Pval}
}{
\Peval{\Pcp,\Ptypedefs,\Pstore,\Penv,\Preturn{\Pexp}}
      {\Pstore_1,\Penv,\Preturn{\Pval}}
}
\\
\Pinfer{E-Assign}{
\Pstmtconf{\Pcp,\Ptypedefs}{\Pstore,\Penv}{\Pexp_1}
\Downarrow_{\mathit{\Plval}}
\langle\Pstore_1,\Plval\rangle
\\\\
\Peval{\Pcp,\Ptypedefs,\Pstore_1,\Penv,\Pexp_2}
      {\Pstore_2,\Pval}
\\\\
\Passignlval{\Pcp,\Delta,\Pstore_2}{\Penv}{\Passign{\Plval}{\Pval}}{\Pstore_3}
}{
\Peval{\Pcp,\Ptypedefs,\Pstore,\Penv,\Passign{\Pexp_1}{\Pexp_2}}
      {\Pstore_3,\Penv,\Pcont}
}
\and
\Pinfer{E-Call-Table}{
\Peval{\Pcp,\Ptypedefs,\Pstore,\Penv,\Pexp}
      {\Pstore_1,\Ptableval{\ell}{\Penv_c}{\Pmany{\Pexp_{\mathit{key}}:x}}{\Pmany{x_{\mathit{act}}(\Pmany{\Pexp_s},\Pmany{x_c:\tau})}}}
\\\\
\Peval{\Pcp,\Ptypedefs,\Pstore_1,\Penv_c,\Pmany{\Pexp_{\mathit{key}}}}
      {\Pstore_2,\Pmany{\Pval_{\mathit{key}}}}
\\\\
\langle \Pcp, \ell, \Pmany{\Pval_{\mathit{key}}:x},\Pmany{x_{\mathit{act}}(\Pmany{x_c:\tau})}\rangle \Downarrow_{\mathit{match}}
x_{\mathit{act}}(\Pmany{\Pexp_c})
\\\\
\Peval{\Pcp,\Ptypedefs,\Pstore_2,\Penv_c,x_{\mathit{act}}(\Pmany{\Pexp_s},\Pmany{\Pexp_c})}
      {\Pstore_3,\Penv_c',\Pcont}
}{
\Peval{\Pcp,\Ptypedefs,\Pstore,\Penv,\Pcall{\Pexp}{}}
      {\Pstore_3,\Penv,\Pcont}
}
\\
\Pinfer{E-VarDecl}{
\Peval{\Pcp,\Ptypedefs,\Pstore,\Penv,\Pvdecl}
      {\Ptypedefs',\Pstore',\Penv',\Pcont}
}{
\Peval{\Pcp,\Ptypedefs,\Pstore,\Penv,\Pvdecl}
      {\Pstore',\Penv',\Pcont}
}
\and
\Pinfer{E-Call}{
\Pexpreval{\Pcp}{\Ptypedefs}{\Pstore}{\Penv}{\Pgencall{\Pexp}{\Pmany{\rho}}{\Pmany{\Pexp}}}
          {\Pstore'}{sig}
}{
\Peval{\Pcp,\Ptypedefs,\Pstore,\Penv,\Pgencall{\Pexp}{\Pmany{\rho}}{\Pmany{\Pexp}}}
      {\Pstore',\Penv,sig}
}
\and
\end{mathpar}
\caption{Semantics for statements.}
\label{fig:stmt-semantics}
\end{figure}

\subsection{Putting it all together}

The static and dynamic semantics presented thus far omit type
declarations. Full rules are in the appendix, but type declarations are simple.
Aside from the open enum type declarations, which add new members to their type,
type declarations just add new type definitions to the type context.

\subsection{Type soundness and termination}

Big-step semantics fail to distinguish between programs that ``go
wrong'' and programs that run forever. For a language with recursion
or loops, this can complicate the proof of a useful type soundness
result. Fortunately, the parser-free fragment of P4 has neither, so we
can prove that all well-typed expressions and statements evaluate to a
final value of appropriate type. The main theorem shows this for
statements.
\begin{theorem}
Let $\langle \Pcp, \Ptypedefs, \Pstore, \Penv, \Pstmt \rangle$ be an
initial configuration and take contexts
$\Xi, \Sigma, \Sigma', \Gamma, \Gamma', \Delta$.  Suppose
\begin{enumerate}
\item $\Xi,\Sigma,\Delta \vdash \Pstore$,
\item $\Xi \vdash \Penv:\Gamma$, and
\item $\Sigma,\Gamma,\Delta \vdash \Pstmt \dashv \Sigma', \Gamma'$
\end{enumerate}
Then there exists a final configuration
$\langle \Pstore', \Penv', \Psig \rangle$ and a store typing
$\Xi' \supseteq \Xi$ such that
\begin{enumerate}
\item $\Peval{\Pcp,\Ptypedefs,\Pstore,\Penv,\Pstmt} {\Pstore',\Penv',\Psig}$,
\item $\Xi',\Sigma',\Gamma',\Delta \vdash \Pstore'$,
\item $\Xi',\Sigma',\Gamma',\Delta \vdash \Penv':\Gamma'$, and
\item if $\Psig = \Preturn{\Pval}$ then there is a type $\tau$ such that
$\Gamma(\Preturnvoid)=\tau$ and
$\Xi', \Sigma', \Delta \vdash \Pval:\tau$.
\end{enumerate}
\end{theorem}

The proof, given in \Cref{app:safety}, is a simple but tedious proof by logical relations. \Cref{app:safety}
includes additional supporting definitions and analogous theorems for
expressions and variable declarations. Note that this is a ``weak
termination'' result: it states that a final configuration exists, but
does not (and cannot, in the language of big-step semantics) say that
all possible ways of evaluating a program will terminate.

%% file: evaltable2.tex
\renewcommand{\arraystretch}{1.2}
\begin{tabular}{@{} >{\hspace{0pt}}p{9em} p{23em}  c p{3.3em} >{\centering}p{2.3em} c @{}}
	\toprule
	\textbf{Test File} & \textbf{Description (Features Tested)} & \textbf{LoC} & \textbf{headers (\#, bits)} & \textbf{parser states} & \textbf{tables?}\\
	\midrule
	issue2287-bmv2.p4 & apply binary operators to function calls with side-effects (operators, side-effects, copy-in/copy-out) & 95 & (3, 248) & 1 & \xmark\\

	enum-bmv2.p4 & equality test on basic enum (enums) & 44 & (1, 96) & 1 & \xmark\\

	issue1025-bmv2.p4  & call \texttt{lookahead} as argument to \texttt{extract} (\texttt{extract}, \texttt{lookahead}, variable-size bitstrings) & 176 & (3, 468) & 3 & \xmark\\

	subparser-with-header-stack-bmv2.p4 & subparser invocation while parsing a header stack (header stacks, parser application) & 168 & (7, 224) & 5 & \xmark\\

	test-parserinvalidarg-ument-error-bmv2.p4 & variable-size extract triggers parser error (variable-size bitstrings, parser errors, control-flow) & 118 & (2, 128) & 2 & \xmark\\

	table-entries-priority-bmv2.p4 & \texttt{priority} annotation affects constant table entries (table application, priority, constant table entries, \texttt{ternary}) & 89 & (1, 48) & 1 & \cmark\\

	default\_action-bmv2.p4 & table application falls through to non-trivial default action (table application, default action, control-plane interface) & 35 & (1, 64) & 1 & \cmark\\

	table-entires-ser-enum-bmv2.p4 & serializable enum appears in constant table entries (serializable enums, table application, constant table entries) & 85 & (1, 16) & 1 & \cmark\\

	checksum3-bmv2.p4 & compute checksum using \texttt{csum16} (externs) & 195 & (3, 320) & 3 & \xmark\\

	count\_ebpf.p4 & stateful extern from \texttt{ebpf\_model} architecture (stateful externs, target abstraction) & 62 & (2, 272) & 2 & \xmark\\
\bottomrule
	
\end{tabular}

%% file: evaltable1.tex
\renewcommand{\arraystretch}{1.1}
\begin{tabular}{@{} p{5.5em} p{27em}  c p{3.3em} >{\centering}p{2.3em} c @{}}
	\toprule
	\textbf{Test File} & \textbf{Description (Features Tested)} & \textbf{LoC} & \textbf{headers (\#, bits)} & \textbf{parser states} & \textbf{tables?}\\
	\midrule
	bitstrings.p4 & emit results of binary operators on bitstrings (bit-strings, \texttt{emit}) & 97 & (0, 0) & 1 & \xmark\\

	stack.p4 & complex operations on header stacks (header stacks) & 141 & (43, 688) & 1 & \xmark\\

	union.p4 & complex operations on header unions (header unions) & 130 & (6, 72) & 1 & \xmark\\

	scope.p4 & function name shadowing (lexical scope) & 52 & (1, 8) & 1 & \xmark\\

	error2.p4 & triggers parser errors (parser errors, control-flow) & 98 & (2, 32) & 2 & \xmark\\

	subparser.p4 & direct application of sub-parser from main parser (parser application, \texttt{verify}, control-flow) & 133 & (5, 40) & 7  & \xmark\\

	exit.p4 & \texttt{exit} statement in nested calls to actions (control-flow) & 107 & (13, 104) & 3 & \xmark\\

	subcontrol.p4 & direct application of sub-control with \texttt{exit} from egress processing (control application, control-flow) & 71 & (2, 16) & 1 & \xmark\\

	table.p4 & apply control-plane-defined table (control-plane interface, table application) & 65 & (1, 8) & 1 & \cmark\\

	table3.p4 & apply table with constant \texttt{lpm} and \texttt{ternary} entries (constant table entries, \texttt{lpm}, \texttt{ternary}, table application) & 94 & (1, 8) & 1 & \cmark\\

	switch-stmt.p4  & \texttt{switch} statement on table with constant entries (constant table entries, table application, \texttt{switch} statement) & 93 & (2, 16) & 2 & \cmark\\
\bottomrule
	
\end{tabular}

%% file: p4cbugs.tex
\aboverulesep=-1em
\begin{tabular}{@{} m{5em}  m{43em} @{}}
	\toprule
	\textbf{Category} & \textbf{Issues Description} \\
	\hline
	Grammar and Parser &  
	\begin{enumerate}[label=(\alph*), leftmargin=1.1em,align=left,labelsep=-1.1em, series = p4c_issues]
	\item The parser had a conflict with the \texttt{TYPE} token where it could either reduce a \texttt{nonTypeName} out of \texttt{TYPE} or shift to recognize a newtype declaration. This problem was detected before it could manifest in the compiler since at the time, top-level functions were not yet implemented.
	 \item The parser incorrectly resolved names with a ``dot'' using the local context instead of the top-level context. 
	 \item The parser rejected actions with dot prefix.
	 \end{enumerate} \\
	\midrule
	Typing & 
	\begin{enumerate}[label=(\alph*), resume* = p4c_issues]
	\item The type system was sometimes nominal and sometimes structural. The behavior was not consistent across individual programming constructs.
	\item The type of a list expression was a tuple which inadvertently allowed tuples to be assigned to structs since list expressions are allowed to be assigned to structs.
	\item The front-end's constant folding transformed a program that was not well-typed into one that was.
	\item Tuples in set contexts are inconsistently flattened when checking matches against keys.
	\end{enumerate} \\ 
	\midrule
	Other & 
	\begin{enumerate}[label = (\alph*), resume* = p4c_issues]
	\item The compiler did not clearly enforce the constraint that the \texttt{default\_action} must appear after \texttt{actions}.
	\item The compiler did not clearly enforce that values of type \texttt{int} should all be compile-time known values.
	\item An STF test had two lines uncommented that were supposed to be commented.
	\item The compiler rejected any program with headers containing multiple \texttt{varbit} fields even though the spec only states such headers cannot be used in \texttt{extract}.
	\item The compiler did not stop compiling after encountering an error.
	\end{enumerate} \\
	\bottomrule
\end{tabular}

%% file: p4specbugs.tex
\aboverulesep=-1em
\begin{tabular}{@{} m{5em}  m{43em} @{}}
	\toprule
	\textbf{Category} & \textbf{Issues Description} \\
	\hline
	Syntax & \begin{enumerate}[label=(\alph*), leftmargin=1.5em,align=left,labelsep=-1.1em, series = spec_issues]
	\item Annotations could take either an expression list or a keyValuePair list for their arguments but this made the grammar ambiguous because there was no way of telling whether the empty list was an expression list or keyValuePair list. This issue came to light before free-form annotations were added.
	\item The \pfoursix grammar required all optional type parameters to be non-type names even though there were use cases that contradicted this restriction and the compiler did not impose such a restriction.
	\item The spec does not impose a requirement on the placement of table entries even though it seems like it should so that a typechecker can process the properties in order.
	\end{enumerate}\\
	\midrule
	Types & \begin{enumerate}[label=(\alph*), resume*=spec_issues]
	\item The spec stated that the compiler does not insert implicit casts for the arguments to methods or functions. This was an undesirable restriction.
	\item The spec is too restrictive because it only permits division and modulo between positive \texttt{int} values.
	 \item The spec does not explicitly allow header unions in header stacks.
	\item The spec does not clarify whether int values can be cast to bool or not. It also does not state whether assigning a bit value to an int variable is allowed by implicitly casting the value to int.
	\end{enumerate} \\
	\midrule
	Operational & (h) The spec did not define the concatenation operator's behavior on signed and unsigned bitstrings.\\\\
	\bottomrule	
\end{tabular}

%% file: unions.tex
% !TEX root = popl21.tex
In this section, we exercise our formal semantics by adding union
types to P4.
Uncertainty about the safety of language extensions is a perennial
concern for the P4 Language Design Working Group, resulting in
language features which are hamstrung or, worse yet, buggy.
With formal semantics, we can prove adding a feature is safe
by defining and proving correct a translation from the augmented
language back into the original one.

P4 already has a restricted form of unions, but they can only contain
header types and lack a type-safe elimination form.
To address this shortcoming, we define an extension of P4 with tagged
unions and formalize its semantics.  We then define a translation from
P4 with unions into standard P4 and prove that the translation
preserves program semantics.

\paragraph*{Syntax.}
We extended the syntax to allow the declaration of a union type. We
allow assignment to union fields, and extended the statements with a
switch statement where cases are either union fields or
default. Finally, we add union values which consist of the union type,
its ``active'' field, and that field's value.

\[
\begin{grammar}
%\Psyndef{\Ptdecl}{\dots~|~\Punion{X}{\Pmany{\tau~f}}}{}
\Psyndef{\tau}{\dots~|~\Punion{X}{\Pmany{\tau~f}}}{}
\Psyndef{\Pstmt}{\dots~|~\Pswitch{\Pexp}{\Pmany{\Pswitchcase{\Plbl}{\Pblk}}}}{}
\Psyndef{\Plbl}{\dots~|~\kw{default}}{}
\Psyndef{\Pval}{\dots~|~\Punionval{X}{f, \Pval}}{}
\end{grammar}
\]

\paragraph*{Typing rules.} We need two new typing rules for statements that include unions (as well as an auxiliary judgment \kw{switchaseok} to check the branches---see the appendix for details).

\begin{mathpar}
\Pinfer{T-Union}{
\Sigma,\Gamma,\Delta \vdash \Pexp_1 : \Punion{X}{\Pmany{\tau~f}} \\\\
\Sigma,\Gamma,\Delta  \vdash \Pexp_2 : \tau_i
}{\Sigma,\Gamma,\Delta \vdash \Pexp_1.f_i
= \Pexp_2 \dashv \Sigma,\Gamma}

\Pinfer{T-Switch}{
\Sigma,\Gamma,\Delta \vdash \Pexp : \Punion{X}{\Pmany{\tau~f}} \\\\
\Pmany{\Plbl \in \{ \Pmany{f}, \kw{default} \}}  \\\\
\Sigma,\Gamma,\Delta \vdash \kw{switchcaseok}(\Pmany{\tau~f}, \Pmany{\Pswitchcase{\Plbl}{\Pblk}})
}{
\Sigma,\Gamma,\Delta \vdash \Pswitch{\Pexp}{\Pmany{\Pswitchcase{\Plbl}{\Pblk}}} \dashv \Sigma,\Gamma
}
\end{mathpar}

\paragraph*{Evaluation rules.}
Union variables are initialized upon declaration ($\kw{init}_\Ptypedefs~X = \Punionval{X}{f_0, \kw{init}_\Ptypedefs~\tau_0}$). A union's value is modified by assigning a value to one of its fields:
\begin{mathpar}
\Pinfer{E-Union}{
\Pstmtconf{\Pcp,\Ptypedefs}{\Pstore,\Penv}{\Pexp_1}
\Downarrow_{\mathit{\Plval}}
\langle\Pstore_1,\Plval\rangle
\\\\
\Peval{\Pcp,\Ptypedefs,\Pstore_1,\Penv,\Pexp_2}
      {\Pstore_2,\Pval}
\\
\Passignlval{\Pstore_2}{\Penv}{\Passign{\Plval.f}{\Pval}}{\Pstore_3}
}{
\Peval{\Pcp,\Ptypedefs,\Pstore,\Penv,\Passign{\Pexp_1.f_i}{\Pexp_2}}
      {\Pstore_3,\Penv,\Pcont}
}
\end{mathpar}
To evaluate a union switch statement on a union with value $\Punionval{X}{f_i, \Pval_i}$, we match $f_i$ against the labels. If it matches a label other than \kw{default}, we evaluate the corresponding block in an environment where $f_i$ maps to $\Pval_i$ (\textsc{E-UnionSwitch}). Note that we can rewrite the blocks so that they have no local variable declaration with the same name as their corresponding label, so we don't violate our naming convention. If a \kw{default} is provided and $f_i$ matches no other label, we proceed with evaluating the corresponding block in the same environment.
If no \kw{default} is provided and there is no match, we skip.
\begin{mathpar}
\Pinfer{E-UnionSwitch}{
\Pstmtconf{\Pcp,\Ptypedefs}{\Pstore,\Penv}{\Pexp} \Downarrow_{\mathit{\Plval}} \langle\Pstore_1,\Plval\rangle\\
\Penv(\Plval) = \ell_1 \\
\Pstore_1(\ell_1) = \Punionval{X}{f_i, \Pval_i} \\\\
\kw{match\_union\_case}~(\Pmany{\Pswitchcase{\Plbl}{\Pblk}}, f_i) = (f_i, k)\\
\Pfresh{\ell_2}\\\\
\Pstore_2 = \Pstore_1[\Pheapset{\ell_2}{\Pval_i}] \\
\Peval{\Pcp,\Ptypedefs, \Pstore_2,\Penv[\Penvset{f_i}{\ell_2}], \Pbraces{\Pmany{\Pstmt}_k}}
      {\Pstore_3,\Penv[\Penvset{f_i}{\ell_2}],sig}
}{
\Peval{\Pcp,\Ptypedefs,\Pstore,\Penv,\Pswitch{\Pexp}{\Pmany{\Pswitchcase{\Plbl}{\Pblk}}}}
      {\Pstore_3,\Penv,sig}
}
%\end{mathpar}
%
%\begin{mathpar}

%\Pinfer{Eval-UnionSwitch-MatchDefault}{
%\Peval{\Pcp,\Ptypedefs,\Pstore,\Penv,\Pexp}
%      {\Pstore_1, \Punionval{X}{f_i, \Pval_i}}
%\\\\
%\kw{match\_union\_case}~(\Pmany{\Pswitchcase{\Plbl}{\Pblk}}, f_i) = (\kw{default}, k)
%\\\\
%\Peval{\Pcp,\Ptypedefs,\Pstore_1,\Penv, \Pbraces{\Pmany{\Pstmt}_k}}
%      {\Pstore_2,\Penv,sig}
%}{
%\Peval{\Pcp,\Ptypedefs,\Pstore,\Penv,\Pswitch{\Pexp}{\Pmany{\Pswitchcase{\Plbl}{\Pblk}}}}
%      {\Pstore_2,\Penv,sig}
%}
%
%\\\\\\
%%\end{mathpar}
%%
%%\begin{mathpar}
%
%\Pinfer{Eval-UnionSwitch-NoMatch}{
%\Peval{\Pcp,\Ptypedefs,\Pstore,\Penv,\Pexp}
%      {\Pstore_1, \Punionval{X}{f_i, \Pval_i}}
%\\\\
%\kw{match\_union\_case}~(\Pmany{\Pswitchcase{\Plbl}{\Pblk}}, f_i) = (\kw{None}, \kw{None})
%}{
%\Peval{\Pcp,\Ptypedefs,\Pstore,\Penv,\Pswitch{\Pexp}{\Pmany{\Pswitchcase{\Plbl}{\Pblk}}}}
%      {\Pstore_1,\Penv,\Pcont}
%}
\end{mathpar}

\paragraph*{Translation to standard P4} We use records in standard P4 to implement unions. $\Puniontrans{\cdot}$ translates from P4 extended with unions to P4.
Expressions are not changed in the extended language. As such, $\Puniontrans{\Pexp} = \Pexp$.
Declaring a new union type translates to a typedef:
$$\Puniontrans{\Punion{X}{\Pmany{\tau~f}}} = \Ptypedef{\Pbraces{tag:\Pbittype{n}, \Pmany{f:\tau}}}{X}$$
Here, $tag$ keeps track of the ``active'' union field. Declaring a new variable of the union type translates to declaring a new variable of the corresponding record type.
$$\Puniontrans{\Pvardecl{X}{x}} = \Pvardecl{X}{x}, \Passign{x.tag}{0}, \Pmany{\Passign{x.f_i}{\kw{init}_\Ptypedefs~\tau_i}}$$
The only extensions to statements are assignment to union fields, and the union switch. Thus, except for the following cases, $\Puniontrans{\Pstmt} = \Pstmt$. In the clause for assignment the field names $\Pmany{f_j}$ range over all $\Pmany{f}$ except $f_i$.
\begin{mathpar}
\begin{array}{rl}
\Puniontrans{\Passign{\Pexp_1.f_i}{\Pexp_2}}
= & \Passign{\Pexp_1}{\Pbraces{tag:\Pbittype{n} = i, f_i:\tau_i = \Pexp_2, \Pmany{f_j:\tau_j = \kw{init}_\Ptypedefs~\tau_j}}} \\
\Puniontrans{\Pswitch{\Pexp}{\Pmany{\Pswitchcase{\Plbl}{\Pblk}}}} = & \Pvardeclinst{X}{tmp}{\Pexp}; \\
& \kw{if}~(c_0)~\Pbraces{b_0} \dots \kw{else\ if}~(c_n)~\Pbraces{b_n}~\kw{else}~\Pbraces{}
\end{array}
\end{mathpar}
where $(c_i, b_i) = \kw{trans\_union\_case}~(\Plbl_i, \Pmany{\Pstmt}_i)$ and:
\begin{mathpar}
\begin{array}{rl}
\kw{trans\_union\_case}~(\kw{default}, \Pmany{\Pstmt}) & = (\kw{true}, \Puniontrans{\Pmany{\Pstmt}}) \\
\kw{trans\_union\_case}~(f_j, \Pmany{\Pstmt}) & =  (tmp.tag == j, \Pvardeclinst{\tau_j}{f_j}{tmp.f_j}, \Puniontrans{\Pmany{\Pstmt}})
\end{array}
\end{mathpar}
Note that $tmp$ is a fresh variable name. Union values are translated to records:
\begin{mathpar}
\Puniontrans{\Punionval{X}{f_i, \Pval_i}} =  \Pbraces{tag:\Pbittype{n} = i, f_i:\tau_i = \Pval_i, f_j:\tau_j = \kw{init}_\Ptypedefs~\tau_j}
\end{mathpar}
For records, headers, and header stacks that are inductively built from other values, we have:
\begin{mathpar}
\begin{array}{rl}
\Puniontrans{\Prec{\Pmany{f:\tau=\Pval}}} = & \Prec{\Pmany{f:\tau=\Puniontrans{\Pval}}} \\
\Puniontrans{\Pheaderval{\Pmany{f:\tau=\Pval}}} = & \Pheaderval{\Pmany{f:\tau=\Puniontrans{\Pval}}} \\
\Puniontrans{\Pstackval{\tau}{\Pmany{\Pval}}} = & \Pstackval{\tau}{\Pmany{\Puniontrans{\Pval}}}
\end{array}
\end{mathpar}
For all other values, $\Puniontrans{\Pval} = \Pval$. We translate stores by translating their range: $\Puniontrans{\Pstore}$ has the same domain as $\Pstore$. If $\Pstore(l) = \Pval$, then $\Puniontrans{\Pstore}(l) = \Puniontrans{\Pval}$.

\paragraph*{Translation property.} We prove in \Cref{app:union-proof} that the translation function is semantics-preserving. Specifically, we prove the following theorem by induction on the statement evaluation rules, and a case analysis on the last rule in the derivation.
\begin{theorem}
If $\Peval{\Pcp,\Ptypedefs,\Pstore,\Penv,\Pstmt}
      {\Pstore',\Penv', sig}$,
then $\Peval{\Pcp,\Ptypedefs,\Puniontrans{\Pstore},\Penv,\Puniontrans{\Pstmt}}
      {\Pstore_t,\Penv_t, sig}$ and
$\Penvlt{\Puniontrans{\Pstore'}}{\Penv'}{\Pstore_t}{\Penv_t}$.
We say $\Penvlt{\Pstore_1}{\Penv_1}{\Pstore_2}{\Penv_2}$ if $\Penv_1$'s domain is a subset of $\Penv_2$'s domain, and for all $\Plval$ in $\Penv_1$'s domain, if $\Penv_1(\Plval) = \ell_1$ and $\Pstore_1(\ell_1) = \Pval$, then $\Penv_2(\Plval) = \ell_2$ and $\Pstore_2(\ell_2) = \Pval$.
\end{theorem}

%% file: typesafety.tex
\newpage
\section{Additional Judgments}
\label{app:judgment}
\subsection{Type declaration typing}
\begin{figure}
\footnotesize
\begin{mathpar}
\Pinfer{T-TypeDefDecl}{
}{
\Sigma,\Gamma,\Delta \vdash
\Ptypedef{\tau}{X} \dashv
\Sigma,\Gamma,\Delta[X = \tau']
}
\and
\Pinfer{T-EnumDecl}{ }{
\Sigma,\Gamma,\Delta \vdash
\Penumdecl{X}{\overline{f}} \dashv
\Sigma,\Gamma,\Delta[X = \Penumdecl{X}{\overline{f}}]
}
\and
\Pinfer{T-ErrorDecl}{ }{
\Sigma,\Gamma,\Delta \vdash
\Perrordecl{\Pmany{f}} \dashv
\Sigma,\Gamma,\Delta[\Perror = \Perrordecl{\Pmany{f}}]
}
\and
\Pinfer{T-MatchKindDecl}{ }{
\Sigma,\Gamma,\Delta \vdash
\Pmatchkinddecl{\Pmany{f}} \dashv
\Sigma,\Gamma,\Delta[\Pmatchkind = \Pmatchkinddecl{\Pmany{f}}]
}
\end{mathpar}
\caption{Type declaration typing rules.}
\label{fig:type-decl-typing}
\end{figure}
\subsection{Type declaration evaluation}
\begin{figure}
\begin{mathpar}
\Pinfer{E-TypeDefDecl}{
}{
\Peval{\Pcp,\Ptypedefs,\Pstore,\Penv, \Ptypedef{\tau}{X}}
      {\Ptypedefs[X = \tau],\Pstore,\Penv,\Pok}
}
\and
\Pinfer{E-EnumDecl}{ }{
\Peval{\Pcp,\Ptypedefs,\Pstore,\Penv, \Penumdecl{X}{\overline{f}}}
      {\Ptypedefs[X = \Penumdecl{X}{\overline{f}}],\Pstore,\Penv,\Pok}
}
\and
\Pinfer{E-ErrorDecl}{ }{
\Peval{\Pcp,\Ptypedefs,\Pstore,\Penv, \Perrordecl{\Pmany{f}}}
      {\Ptypedefs[\Perror = \Perrordecl{\Pmany{f}}],\Pstore,\Penv,\Pok}
}
\and
\Pinfer{E-MatchKindDecl}{ }{
\Peval{\Pcp,\Ptypedefs,\Pstore,\Penv, \Pmatchkinddecl{\Pmany{f}}}
      {\Ptypedefs[\Pmatchkind = \Pmatchkinddecl{\Pmany{f}}],\Pstore,\Penv,\Pok}
}
\end{mathpar}
\caption{Semantics for type declarations.}
\label{fig:type-decl-semantics}
\end{figure}
\subsection{Value typing}
\begin{figure}
\begin{mathpar}
\Pinfer{TV-Rec}{
\Xi,\Sigma,\Delta \vdash \Pmany{\Pval:\tau}
}{
\Xi,\Sigma,\Delta \vdash \Prec{\Pmany{f=\Pval}}:\Prectype{\Pmany{f:\tau}}
}
\and
\Pinfer{TV-Hdr}{
\Xi,\Sigma,\Delta \vdash \Pmany{\Pval:\tau}
}{
\Xi,\Sigma,\Delta \vdash \Pheaderval{\Pmany{f:\tau=\Pval}}:\Pheader{\Pmany{f:\tau}}
}
\and
\Pinfer{TV-MemEnum}{
\Delta(X)=\Penumtype{X}{\Pmany{f}}
}{
\Xi,\Sigma,\Delta \vdash X.f:\Penumtype{X}{\Pmany{f}}
}
\and
\Pinfer{TV-MemMk}{
\Delta(\Pmatchkind)=\Pmatchkindtype{f}
}{
\Xi,\Sigma,\Delta \vdash \Pmatchkind.f:\Pmatchkind
}
\and
\Pinfer{TV-MemErr}{
\Delta(\Perror)=\Perrortype{f}
}{
\Xi,\Sigma,\Delta \vdash \Perror.f:\Perror
}
\and
\Pinfer{TV-Stack}{
\kw{len}(\Pmany{\Pval})=n \\
\Xi,\Sigma,\Delta \vdash \Pmany{\Pval}:\tau
}{
\Xi,\Sigma,\Delta \vdash \Pstackval{\tau}{\Pmany{\Pval}}: \tau[n]
}
\and
\Pinfer{TV-Clos}{
\Xi,\Sigma,\Delta\vdash\Penv:\Gamma \\
\Sigma,\Gamma[\Pmany{x:\tau}],\Delta[\Pmany{X~\kw{var}}] \vdash \Pmany{\Pdecl} \dashv \Sigma',\Gamma',\Delta' \\
\Sigma',\Gamma'[\kw{return}=\tau],\Delta' \vdash \Pstmt \dashv \Sigma'',\Gamma''
}{
\Xi,\Sigma,\Delta \vdash \Pclos{\Penv}{\Pmany{X}}{\Pmany{d~x:\tau}}{\tau}{\Pmany{\Pdecl}~\Pstmt}:\Pfunctiontype{\Pmany{X}}{\Pmany{d~x:\tau}}{\tau}
}
\and
\Pinfer{TV-Native}{
\mathcal{N} \vdash x: \Psimplefunctiontype{\Pmany{d~x:\tau}}{\tau}
}{
\Xi,\Sigma,\Delta \vdash \Pnative{x}{\Pmany{d~x:\tau}}{\tau}:\Psimplefunctiontype{\Pmany{d~x:\tau}}{{\tau}}
}
\and
\Pinfer{TV-CClos}{
\Xi,\Sigma,\Delta\vdash\Penv:\Gamma
\\\\
\Sigma,\Gamma[\Pmany{x_c:\tau_c}][\Pmany{x:\tau}],\Delta \vdash \Pmany{\Pdecl} \dashv \Sigma_1,\Gamma_1,\Delta_1
\\
\Sigma_1,\Gamma_1[\Preturnvoid:\{\}],\Delta_1 \vdash \Pstmt \dashv \Sigma_2,\Gamma_2
}{
\Xi,\Sigma,\Delta \vdash
\Pcclos{\Penv}{\Pmany{d~x:\tau}}{\Pmany{x_c:\tau_c}}{\Pmany{\Pdecl}}{\Pstmt}
 : \Pctortype{\Pmany{x_c:\tau_c}}{\Psimplefunctiontype{\Pmany{d~x:\tau}}{\{\}}}
}
\and
\Pinfer{TV-Table}{
\Xi,\Sigma,\Delta\vdash\Penv:\Gamma \\
\Sigma,\Gamma,\Delta \vdash \Pmany{\Pexp:\tau_k} \\
\Sigma,\Gamma,\Delta \vdash \Pmany{x:\kw{match\_kind}} \\
\Sigma,\Gamma,\Delta \vdash \Pmany{act~\kw{act\_ok}}
}{
\Xi,\Sigma,\Delta \vdash \Ptableval{\ell}{\Penv}{\Pmany{\Pexp:x}}{\Pmany{act}}:\Ptabletype
}
\and
\Pinfer{Type-Partial-App}{
\Sigma,\Gamma,\Delta \vdash x_{\mathit{act}} : \Pfunctiontype{}{\Pmany{d~x:\tau}, \Pmany{\Pin~x_p:\tau}}{\{\}} \\
\Sigma,\Gamma,\Delta \vdash \Pmany{\Pexp:\tau\Pgoes{d}}
}{
\Sigma,\Gamma,\Delta \vdash x_{\mathit{act}}(\Pmany{\Pexp},\Pmany{x_c:\tau}):~\kw{act\_ok}
}
\end{mathpar}
\caption{Value typing rules.}
\label{fig:value-typing}
\end{figure}

\subsection{Compile Time Evaluation Rules}
\begin{mathpar}
\Pinfer{CTE-Bool}{  }{\Pcteval{\Sigma, b}{b}}
\and
\Pinfer{CTE-Bit}{ }{ \Pcteval{\Sigma, \Pint{n}{w}}{\Pint{n}{w}} }
\and
\Pinfer{CTE-Var}{\Sigma(x) = \Pval}{  \Pcteval{\Sigma, x}{ \Pval } }
\and
\Pinfer{CTE-UOp}{ \Pcteval{\Sigma, \Pexp}{ \Pval } }{  \Pcteval{\Sigma, {\ominus}\Pexp}{ {\ominus}\Pval } }
\and
\Pinfer{CTE-BinOp}{
\Pcteval{\Sigma, \Pexp_1}{\Pval_1} \\
\Pcteval{\Sigma, \Pexp_2}{\Pval_2} }{
\Pcteval{\Sigma, \Pexp_1\oplus\Pexp_2}{\Pval_1\oplus\Pval_2}
}
\end{mathpar}

\subsection{L-value writing.}
\begin{mathpar}
\Pinfer{LW-Var}{
\Penv(x) = \ell
}{
\Passignlval{\Pcp,\Delta,\Pstore}{\Penv}{\Passign{x}{\Pval}}
            {\Pstore[x:=\Pval]}
}
\and
\Pinfer{LW-Rec}{
\langle \Pcp, \Delta, \Pstore, \Penv, \Plval \rangle \Downarrow
\langle \Pstore_1, \{\Pmany{f=\Pval_f}\} \rangle
\\
\Passignlval{\Pcp,\Delta,\Pstore_1}{\Penv}{\Passign{\Plval}{
\Prec{f_i=\Pval, \Pmany{f_{\neq i} = \Pval_f}}
}}{\Pstore_2}
}{
\Passignlval{\Pcp,\Delta,\Pstore}{\Penv}{\Passign{\Plval.f_i}{\Pval}}
            {\Pstore_2}
}
\and
\Pinfer{LW-HdrV}{
\langle \Pcp, \Delta, \Pstore, \Penv, \Plval \rangle \Downarrow
\langle \Pstore_1, \Pheader{\kw{valid}=\Pfalse,\Pmany{f=\Pval_f}} \rangle
\\
\Passignlval{\Pcp,\Delta,\Pstore_1}{\Penv}{\Passign{\Plval}{
\Pheader{\kw{valid}=\Pfalse, f_i=\Pval, \Pmany{f_{\neq i} = \Pval_f}}
}}{\Pstore_2}
}{
\Passignlval{\Pcp,\Delta,\Pstore}{\Penv}{\Passign{\Plval.f_i}{\Pval}}
            {\Pstore_2}
}
\and
\Pinfer{LW-HdrInV}{
\langle \Pcp, \Delta, \Pstore, \Penv, \Plval \rangle \Downarrow
\langle \Pstore_1, \Pheader{\kw{valid}=\Pfalse, \Pmany{f=\Pval_f}} \rangle
}{
\Passignlval{\Pcp,\Delta,\Pstore}{\Penv}{\Passign{\Plval.f_i}{\Pval}}
            {\Pstore_2}
}
\and
\Pinfer{LW-Idx}{
\langle \Pcp, \Delta, \Pstore, \Penv, \Plval \rangle \Downarrow
\langle \Pstore_1, \Pstackval{\tau}{\Pmany{\Pval}} \rangle \\
\Passignlval{\Pcp,\Delta,\Pstore_1}{\Penv}{\Passign{\Plval}{\Pstackval{\tau}{\dots,\Pval_{n-1},\Pval,\Pval_{n+1},\dots}
}}{\Pstore_2}
}{
\Passignlval{\Pcp,\Delta,\Pstore}{\Penv}{\Passign{\Plval[n]}{\Pval}}
            {\Pstore_2}
}
\and
\Pinfer{LW-Slice}{
\langle \Pcp, \Delta, \Pstore, \Penv, \Plval \rangle \Downarrow
\langle \Pstore_1, n_w\rangle \\
\Passignlval{\Pcp,\Delta,\Pstore_1}{\Penv}{\Passign{\Plval}{
\kw{bitset}(n_w, n, m, n_v')
}}{\Pstore_2}
}{
\Passignlval{\Pcp,\Delta,\Pstore}{\Penv}{\Passign{\Plval[n{:}m]}{n_v'}}
            {\Pstore_2}
}
\end{mathpar}

\subsection{L-value evaluation.}
\begin{mathpar}
\Pinfer{LE-Var}{ }{
\Plvaleval{\Pcp}{\Ptypedefs}{\Pstore}{\Penv}{x}{\Pstore}{x}
}
\and
\Pinfer{LE-Rec}{
\Plvaleval{\Pcp}{\Ptypedefs}{\Pstore}{\Penv}{\Pexp}{\Pstore'}{\Plval}
}{
\Plvaleval{\Pcp}{\Ptypedefs}{\Pstore}{\Penv}{\Pexp.f}{\Pstore'}{\Plval.f}
}
\and
\Pinfer{LE-Rec}{
\Plvaleval{\Pcp}{\Ptypedefs}{\Pstore}{\Penv}{\Pexp_1}{\Pstore_1}{\Plval}
\\
\Peval{\Pcp,\Ptypedefs,\Pstore_1,\Penv,\Pexp_2}{\Pstore_2,n_w}
}{
\Plvaleval{\Pcp}{\Ptypedefs}{\Pstore}{\Penv}{\Pexp_1[\Pexp_2]}{\Pstore_2}{\Plval[n]}
}
\and
\Pinfer{LE-Slice}{
\Plvaleval{\Pcp}{\Ptypedefs}{\Pstore}{\Penv}{\Pexp_1}{\Pstore_1}{\Plval}
\\
\Peval{\Pcp,\Ptypedefs,\Pstore_1,\Penv,\Pexp_2}{\Pstore_2,n_\infty}
\\
\Peval{\Pcp,\Ptypedefs,\Pstore_2,\Penv,\Pexp_3}{\Pstore_3,m_\infty}
\\
}{
\Plvaleval{\Pcp}{\Ptypedefs}{\Pstore}{\Penv}{\Pexp_1[\Pexp_2{:}\Pexp_3]}{\Pstore_3}{\Plval[n{:}m]}
}
\end{mathpar}

\subsection{Type simplification, omitting rules for integers, booleans, errors, match kinds, enums, and tables.}
\begin{mathpar}
\Pinfer{TyS-Bit}{
\Pcteval{\Sigma,\Pexp}{n_w}
}{
\Ptypeeval{\Sigma,\Delta}{\Pbittype{\Pexp}}{\Pbittype{n}}
}
\and
\Pinfer{TyS-Rec}{
\Ptypeeval{\Sigma,\Delta}{\Pmany{\tau}}{\Pmany{\tau'}}
}{
\Ptypeeval{\Sigma,\Delta}{\Prectype{\Pmany{f:\tau}}}{\Prectype{\Pmany{f:\tau'}}}
}
\and
\Pinfer{TyS-Var}{
\Delta(X)=\tau \\
\Ptypeeval{\Sigma,\Delta}{\tau}{\tau'}
}{
\Ptypeeval{\Sigma,\Delta}{X}{\tau'}
}
\and
\Pinfer{TyS-Hdr}{
\Ptypeeval{\Sigma,\Delta}{\Pmany{\tau}}{\Pmany{\tau'}}
}{
\Ptypeeval{\Sigma,\Delta}{\Pheader{\Pmany{f:\tau}}}{\Pheader{\Pmany{f:\tau'}}}
}
\and
\Pinfer{TyS-Stack}{
\Ptypeeval{\Sigma,\Delta}{\Pmany{\tau}}{\Pmany{\tau'}}
}{
\Ptypeeval{\Sigma,\Delta}{\Pstacktype{\tau}{n}}{\Pstacktype{\tau'}{n}}
}
\and
\Pinfer{TyS-Fun}{
\Ptypeeval{\Sigma,\Delta[\Pmany{X~\kw{var}}]}{\Pmany{\tau}}{\Pmany{\tau'}}
\\
\Ptypeeval{\Sigma,\Delta[\Pmany{X~\kw{var}}]}{\tau_{\mathit{ret}}}{\tau_{\mathit{ret}}'}
}{
\Ptypeeval{\Sigma,\Delta}{\Pfunctiontype{\Pmany{X}}{\Pmany{d~x:\tau}}{\tau_{\mathit{ret}}}}
                         {\Pfunctiontype{\Pmany{X}}{\Pmany{d~x:\tau'}}{\tau_{\mathit{ret}}'}}
}
\and
\Pinfer{TyS-Ctor}{
\Ptypeeval{\Sigma,\Delta}{\Pmany{\tau}}{\Pmany{\tau'}}
\\
\Ptypeeval{\Sigma,\Delta}{\tau_{\mathit{ret}}}{\tau_{\mathit{ret}}'}
}{
\Ptypeeval{\Sigma,\Delta}{\Pctortype{\Pmany{x:\tau}}{\tau_{\mathit{ret}}}}
                         {\Pctortype{\Pmany{x:\tau'}}{\tau_{\mathit{ret}}'}}
}
\end{mathpar}
\subsection{Runtime type evaluation. We omit the trivial rules for integers, booleans, errors, match kinds, enums, and tables.}
\begin{mathpar}
\Pinfer{TyE-Bit}{
\langle \Delta,\Pstore,\Penv,\Pexp \rangle \Downarrow \langle \sigma',n_w\rangle
}{
\langle \Delta,\Pstore,\Penv,\Pbittype{\Pexp} \rangle \Downarrow_{\tau} \Pbittype{n}
}
\and
\Pinfer{TyE-Rec}{
\langle \Delta,\Pstore,\Penv,\Pmany{\tau} \rangle \Downarrow_{\tau}
\Pmany{\tau'}
}{
\langle \Delta,\Pstore,\Penv,\Prectype{\Pmany{f:\tau}} \rangle \Downarrow_{\tau}
\Prectype{\Pmany{f:\tau'}}
}
\and
\Pinfer{TyE-Var}{
\Delta(X)=\tau \\
\langle \Delta,\Pstore,\Penv,\tau \rangle \Downarrow_{\tau}  \tau'
}{
\langle \Delta,\Pstore,\Penv, X\rangle \Downarrow_{\tau}  \tau'
}
\and
\Pinfer{TyE-Hdr}{
\langle \Delta,\Pstore,\Penv, \Pmany{\tau}\rangle \Downarrow_{\tau}
\Pmany{\tau'}
}{
\langle \Delta,\Pstore,\Penv, \Pheader{\Pmany{f:\tau}}\rangle \Downarrow_{\tau}
\Pheader{\Pmany{f:\tau'}}
}
\and
\Pinfer{TyE-Stack}{
\langle \Delta,\Pstore,\Penv, \Pmany{\tau}\rangle \Downarrow_{\tau}
\Pmany{\tau'}
}{
\langle \Delta,\Pstore,\Penv,\Pstacktype{\tau}{n} \rangle \Downarrow_{\tau}
\Pstacktype{\tau'}{n}
}
\and
\Pinfer{TyE-Fun}{
\langle \Delta[\Pmany{X~\kw{var}}],\Pstore,\Penv, \Pmany{\tau}\rangle \Downarrow_{\tau}
\Pmany{\tau'}
\\
\langle \Delta[\Pmany{X~\kw{var}}],\Pstore,\Penv, \tau_{\mathit{ret}}\rangle \Downarrow_{\tau} \tau_{\mathit{ret}}'
}{
\langle \Delta,\Pstore,\Penv,\Pfunctiontype{\Pmany{X}}{\Pmany{d~x:\tau}}{\tau_{\mathit{ret}}} \rangle \Downarrow_{\tau}
\Pfunctiontype{\Pmany{X}}{\Pmany{d~x:\tau'}}{\tau_{\mathit{ret}}'}
}
\and
\Pinfer{TyE-Ctor}{
\langle \Delta,\Pstore,\Penv, \Pmany{\tau}\rangle \Downarrow_{\tau} \Pmany{\tau'}
\\
\langle \Delta,\Pstore,\Penv,\tau_{\mathit{ret}} \rangle \Downarrow_{\tau} \tau_{\mathit{ret}}'
}{
\langle \Delta,\Pstore,\Penv,\Pctortype{\Pmany{x:\tau}}{\tau_{\mathit{ret}}} \rangle \Downarrow_{\tau}
\Pctortype{\Pmany{x:\tau'}}{\tau_{\mathit{ret}}'}
}
\end{mathpar}

\subsection{Environment and store typing}
\begin{definition}
We say $\Xi,\Delta \vdash \Pstore$ if for all $\ell$ in the domain of
$\Pstore$ there exists a type $\tau$ with $\Xi(\ell)=\tau$ and
$\Xi,\Delta \vdash \Pstore(\ell):\tau$.
\end{definition}

\begin{definition}
We define $\Xi \vdash \Penv:\Gamma$ inductively from the following rules.
\begin{mathpar}
\Pinfer{TEnv-E}{ }{\Xi \vdash []:[]}
\and
\Pinfer{TEnv-C}{
\Xi \vdash \Penv:\Gamma
\\
\Xi(\ell)=\tau
\\
}{
\Xi \vdash (\Penv,x\mapsto\ell):(\Gamma,x:\tau)
}
\and
\Pinfer{TEnv-R}{
\Xi \vdash \Penv:\Gamma
}{
\Xi \vdash \Penv:\Gamma,\kw{return}\mapsto\ell
}
\end{mathpar}
\end{definition}

\begin{definition}
We say $\Sigma \vdash \langle \Pstore, \Penv \rangle$ if for all $x$
in the domain of $\Sigma$, $\Pstore(\Penv(x))$ is defined and
$\Pstore(\Penv(x)) = \Sigma(x)$.
\end{definition}

\subsection{Semantic typing}
\begin{definition}
The abstract variables $\Pvars{\Ptypedefs}$ of a context are its
entries of the form $X~\kw{var}$.
\end{definition}

\begin{definition}
The definitions $\Pdefs{\Ptypedefs}$ of a context are its entries
of the form $X=\rho$.
\end{definition}

Next we define ``semantic typing'' for the purpose of proving
termination. These propositions are written with the double turnstile
$\vDash$ and are defined by mutual induction.

\begin{definition}
Semantic typing of stores $\Xi,\Sigma,\Delta \vDash \sigma$ holds if for
every location $\ell$ in $\sigma$ there exists a type $\tau$ such that
$\Xi(\ell)=\tau$ and $\Xi,\Sigma,\Pdefs{\Delta} \vDash \sigma(\ell):\tau$.
\end{definition}

\begin{definition}
Semantic typing of environments $\Xi,\Delta \vDash \epsilon:\Gamma$
is the same as the ordinary typing
$\Xi,\Delta,\vdash \epsilon:\Gamma$.
\end{definition}

\begin{definition}
Semantic typing of an environment and store tuple 
    $\Xi\Sigma,\Delta \vDash \langle \sigma, \epsilon \rangle:\Gamma$
holds when
$\Xi,\Sigma,\Delta \vDash \sigma$,
$\Xi \vDash \epsilon : \Gamma$,
and $\Sigma \vdash \langle \Pstore, \Penv \rangle$ all hold.
\end{definition}

\begin{definition}[Semantic typing of expressions]
If there are $n$ variables $\Pmany{X}$ in $\Pvars{\Delta}$,
semantic typing 
$\Sigma,\Gamma,\Delta \vDash \Pexp:\tau~\kw{goes}~d$ is a predicate on
length $n$ lists of type arguments $\rho$. It holds when the ordinary typing
$\Sigma,\Gamma,\Delta \vdash \Pexp:\tau~\kw{goes}~d$ holds and when for any
store typing $\Xi$, reduced type $\hat{\tau}$, reduced typing context
$\hat{\Gamma}$, store $\Pstore$, and environment $\Penv$ satisfying
\begin{enumerate}
\item $\Sigma, \Delta[\Pmany{X=\rho}] \vdash \tau \rightsquigarrow \hat{\tau}$
\item $\Sigma, \Delta[\Pmany{X=\rho}] \vdash \Gamma \rightsquigarrow \hat{\Gamma}$
\item $\Xi,\Sigma,\Delta[\Pmany{X=\rho}] \vDash \langle \Pstore, \Penv \rangle:\hat{\Gamma}$
\end{enumerate}
there exists a final configuration $\langle \Pstore', \Psig \rangle$
and a store typing $\Xi' \supseteq \Xi$ such that the following
conditions hold.
\begin{enumerate}
\item $\langle \mathcal{C},\Delta[\Pmany{X=\rho}], \Pstore, \Penv, \Pexp \rangle
\Downarrow \langle \Pstore', \Psig \rangle$
\item $\Xi',\Sigma,\Delta[\Pmany{X=\rho}] \vDash \langle\Pstore', \Penv\rangle:\hat{\Gamma}$
\item If $\Psig = \Pval$ then $\Sigma, \Xi', \Delta[\Pmany{X=\rho}] \vDash \Pval:\hat{\tau}$.
\end{enumerate}
\end{definition}

\begin{definition}[Semantic typing of statements]
If there are $n$ variables $\Pmany{X}$ in $\Pvars{\Delta}$,
semantic typing 
$\Sigma,\Gamma,\Delta \vDash \Pstmt \Pddashv \Sigma', \Gamma'$ is a predicate on
length $n$ lists of type arguments $\rho$. It holds when
the ordinary typing 
$\Sigma,\Gamma,\Delta \vdash \Pstmt \dashv \Sigma', \Gamma'$
holds and for any store
typing $\Xi$, reduced typing contexts $\hat{\Gamma}$ and
$\hat{\Gamma}'$, store $\Pstore$, and environment $\Penv$ satisfying
\begin{enumerate}
\item $\Sigma, \Delta[\Pmany{X=\rho}] \vdash \Gamma \rightsquigarrow \hat{\Gamma}$
\item $\Sigma', \Delta[\Pmany{X=\rho}] \vdash \Gamma' \rightsquigarrow \hat{\Gamma}'$
\item $\Xi,\Sigma,\Delta[\Pmany{X=\rho}] \vDash \langle \Pstore, \Penv \rangle:\hat{\Gamma}$
\end{enumerate}
there exists a final configuration $\langle \Pstore', \Penv', \Psig \rangle$
and a store typing $\Xi' \supseteq \Xi$ such that the following
conditions hold.
\begin{enumerate}
\item $\langle \mathcal{C},\Delta[\Pmany{X=\rho}], \Pstore, \Penv, \Pstmt \rangle \Downarrow \langle \Pstore', \Penv', \Psig \rangle$
\item $\Xi',\Sigma,\Delta[\Pmany{X=\rho}] \vDash \langle\Pstore', \Penv'\rangle:\hat{\Gamma'}$
\item If $\Psig = \Preturn{\Pval}$ then there is a type $\tau$ such that
$\Gamma(\Preturnvoid)=\tau$ and a reduced type $\hat{\tau}$ such that
$\Sigma,\Delta[\Pmany{X=\rho}] \vdash \tau \rightsquigarrow \hat{\tau}$
with $\hat{\Gamma}'(\Preturnvoid)=\hat{\tau}$ and
$\Xi', \Sigma', \Delta[\Pmany{X=\rho}] \vDash \Pval:\hat{\tau}$.
\end{enumerate}
\end{definition}

\begin{definition}[Semantic typing of declarations]
If there are $n$ variables $\Pmany{X}$ in $\Pvars{\Delta}$,
semantic typing 
$\Sigma,\Gamma,\Delta \vDash \Pdecl \Pddashv \Sigma', \Gamma', \Delta'$
is a predicate on length $n$ lists of type arguments $\rho$. It holds
when 
the ordinary typing 
$\Sigma,\Gamma,\Delta \vdash \Pdecl \dashv \Sigma', \Gamma', \Delta'$
holds
and for any store typing $\Xi$, reduced typing contexts
$\hat{\Gamma}$ and $\hat{\Gamma}'$, store $\Pstore$, and environment
$\Penv$ satisfying
\begin{enumerate}
\item $\Sigma, \Delta[\Pmany{X=\rho}] \vdash \Gamma \rightsquigarrow \hat{\Gamma}$
\item $\Sigma', \Delta'[\Pmany{X=\rho}] \vdash \Gamma' \rightsquigarrow \hat{\Gamma}'$
\item $\Xi,\Sigma,\Delta[\Pmany{X=\rho}] \vDash \langle \Pstore, \Penv \rangle:\hat{\Gamma}$
\end{enumerate}
there exists a final configuration $\langle \Delta'', \Pstore', \Penv', \Psig \rangle$
such that the following conditions hold.
\begin{enumerate}
\item $\langle \mathcal{C},\Delta[\Pmany{X=\rho}], \Pstore, \Penv, \Pdecl \rangle \Downarrow \langle \Delta'', \Pstore', \Penv', \Psig \rangle$
\item $\Delta'' = \Delta'[\Pmany{X=\rho}]$
\item $\Xi',\Sigma',\Delta'[\Pmany{X=\rho}] \vDash \langle\Pstore', \Penv'\rangle:\hat{\Gamma'}$
\end{enumerate}
\end{definition}

\begin{definition}[Semantic typing of closures]
Semantic typing 
$\Xi,\Sigma,\Delta \vDash \Pclos{\Penv}{\Pmany{X}}{\Pmany{d~x:\tau}}{\tau_{\mathit{ret}}}{\Pmany{\Pdecl}~\Pstmt}$
for closure values holds if there exists some $\Gamma$ such that the
following conditions all hold.
\begin{enumerate}
\item $\Xi \vDash \Penv : \Gamma$
\item $\Xi,\Sigma,\Delta \vDash \Pclos{\Penv}{\Pmany{X}}{\Pmany{d~x:\tau}}{\tau_{\mathit{ret}}}{\Pmany{\Pdecl}~\Pstmt}$
\item $\Sigma,\Gamma[\Pmany{x:\tau}],\Delta[X~\kw{var}] \vDash \Pdecl \Pddashv \Sigma', \Gamma', \Delta'$
\item $\Sigma',\Gamma'[\Preturnvoid=\tau_{\mathit{ret}}],\Delta' \vDash \Pstmt \Pddashv \Sigma'', \Gamma''$
\end{enumerate}
\end{definition}

\begin{definition}[Semantic typing of constructor closures]
Semantic typing
\(
\Xi,\Ptypedefs,\Sigma \vDash \Pcclos{\Penv}{\Pmany{d~x:\tau}}{\Pmany{x_c:\tau_c}}{\Pmany{\Pdecl}}{\Pstmt}
\)
for constructor closures holds if, given fresh locations
$\Pmany{\ell}$, it is the case that
\[\Xi[\Pmany{\ell:\tau_c}],\Ptypedefs,\Sigma
\vDash
\Pclossimple{\Penv[\Pmany{\Penvset{x_c}{\ell}}]}{\Pmany{d~x:\tau}}{\Pmany{\Pdecl}}{\Pstmt}.\]
and 
\(
\Xi,\Ptypedefs,\Sigma \vdash \Pcclos{\Penv}{\Pmany{d~x:\tau}}{\Pmany{x_c:\tau_c}}{\Pmany{\Pdecl}}{\Pstmt}
\)
\end{definition}

\begin{definition}[Semantic typing of tables]
Semantic typing
\(
\Xi,\Ptypedefs,\Sigma \vDash \Ptableval{\ell}{\Penv}{\Pmany{\Pexp:x}}{\Pmany{act}}:\Ptabletype
\)
for tables holds if all of the following hold.
\begin{align*}
\Xi,\Sigma,\Delta&\vDash\Penv:\Gamma \\
\Sigma,\Gamma,\Delta &\vDash \Pmany{\Pexp:\tau_k} \\
\Sigma,\Gamma,\Delta &\vDash \Pmany{x:\kw{match\_kind}} \\
\Sigma,\Gamma,\Delta &\vDash x_{\mathit{act}} : \Pfunctiontype{}{\Pmany{d~x:\tau}, \Pmany{\Pin~x_p:\tau_p}}{\{\}} &\text{for each $\mathit{act} = x_{\mathit{act}}(\Pmany{\Pexp_a})$}\\
\Sigma,\Gamma,\Delta &\vDash \Pmany{\Pexp_k:\tau_k\Pgoes{d}} &\text{for each $\mathit{act} = x_{\mathit{act}}(\Pmany{\Pexp_a})$}
\end{align*}
\end{definition}

\begin{definition}[Semantic typing of values]
$\Xi,\Ptypedefs,\Sigma \vDash \Pval:\tau$ holds if the syntactic typing
$\Xi,\Ptypedefs,\Sigma \vdash \Pval:\tau$ holds
and when the value $\Pval$ is a closure, constructor closure, or table and the
corresponding special semantic typing holds.
\end{definition}

\section{Type safety}
\label{app:safety}

\subsection{Well-formedness assumptions}
First, we state some well-formedness conditions on contexts which we assume
throughout. We assert without proof that the type system preserves these invariants.
\begin{enumerate}
\item We assume all type variables appearing in a context $\Gamma$ are bound in
    their corresponding $\Delta$. In symbols, $\kw{FTV}(\Gamma) \subseteq
        \kw{vars}(\Delta)$, where
        $\kw{FTV}(x_1:\tau_1,\dots,x_n:\tau_n)$ collects the type
        variables appearing free in all the $\tau_i$. Furthermore, any
        type $\tau$ appearing in $\Gamma$ contains no unevaluated
        expressions.

\item
Types in contexts $\Xi$ contain no free type variables
and no unevaluated expressions.
\end{enumerate}
The choice of $\Delta$ for the well-formedness of some $\Gamma$ or
$\Xi$ should always be clear from context.

There are also some constants which we make assumptions about.
\begin{enumerate}
\item We assume the $\Downarrow_{\kw{match}}$ relation used in table evaluation picks
an action from the table's action list and provides well-typed runtime
arguments. This means that for any configuration
$\langle \mathcal{C},\ell,\Pmany{\Pval:x_k},\Pmany{x(\Pmany{x_c:\tau}})$
there is some $\Pexp_c:\tau_i$ such that \[
\langle \mathcal{C},\ell,\Pmany{\Pval:x_k},\Pmany{x(\Pmany{x_c:\tau}})\rangle
\Downarrow_{\mathit{match}} x_i(\Pmany{\Pexp_c})
\] and
$[],[],[] \vDash \Pmany{\Pexp_c : \tau_i}$.

\item We assume that native functions execute safely if given arguments matching their type.
In particular, suppose $\mathcal{N} \vdash
x:\Psimplefunctiontype{\Pmany{d~x:\tau}}{\tau}$ and
$\Xi,\Sigma,\Delta \vdash \Pstore$ with
$\Xi(\Pmany{\ell})=\Pmany{\tau}$.  Then there exists a store $\Pstore'$, a
value $\Pval$, and a store typing $\Xi'$ such that
$\Xi', \Sigma, \Delta \vdash \Pstore'$,
$\Xi', \Sigma, \Delta \vdash \Pval:\tau$, and
$\mathcal{N}(x,\Pstore,[\Pmany{x\mapsto\ell}])
= \langle \Pstore', \Pval \rangle$.

\item
We assume unary operator typing implies safe evaluation.
If $\mathcal{T}(\Delta,\ominus,\tau)=\tau'$ then for any value $\Pval$
with $\Xi, \Sigma, \Delta \vdash \Pval:\tau$ the function application
$\mathcal{E}(\ominus,\Pval)$ is defined and the
result $\Pval'$ satisfies $\Xi,\Sigma,\Delta \vdash \Pval':\tau'$.

\item
We assume binary operator typing implies safe evaluation.
If $\mathcal{T}(\Delta,\oplus,\tau_1,\tau_2)=\tau_3$ then for any values
$\Pval_1, \Pval_2$ with $\Xi,\Sigma,\Delta \vdash \Pval_i:\tau_i$ the function
application $\mathcal{E}(\oplus, \Pval_1, \Pval_2)$
is defined and the result $\Pval_3$ satisfies
$\Xi,\Sigma,\Delta \vdash \Pval_3:\tau_3$.

\item
We assume all legal casts execute safely.
If $\Delta \vdash \tau \preceq \tau'$ and $\Xi, \Sigma, \Delta \vdash \Pval:\tau$,
the cast $\Pcasteval{\Delta}{\Pval}{\tau'}$ has a defined value
$\Pval'$ satisfying $\Xi, \Sigma,\Delta \vdash \Pval':\tau'$.
\end{enumerate}

\subsection{Typing lemmas}

\begin{lemma}\label{thm:typeevalagree}
If $\Sigma \vdash \langle \Pstore, \Penv \rangle$ and
 $\Ptypeeval{\Sigma, \Delta}{\tau}{\tau'}$, then
$\langle \Delta,\Pstore,\Penv,\tau\rangle \Downarrow_{\tau} \tau'$.
\begin{proof}
By induction on $\tau$, with a use of \cref{thm:ctevalcorrect} in the
$\Pbittype{\Pexp}$ case.
\end{proof}
\end{lemma}

\begin{lemma}\label{thm:exprtypecanonical}
If $\Sigma,\Gamma,\Delta \vdash e:\tau$, then $\tau$ contains no
unevaluated expressions and if $X$ is a free variable of $\tau$ then
$X~\kw{var}$ appears in $\Delta$.
\begin{proof}
By induction on typing derivations.
\end{proof}
\end{lemma}

\begin{lemma}\label{thm:constvaltypeweak}
If $\Sigma \subseteq \Sigma'$ and $\Xi,\Sigma,\Delta \vdash \Penv:\tau$
$\Xi,\Sigma',\Delta \vdash \Penv:\tau$.
\begin{proof}
By induction on derivations of $\Xi,\Sigma,\Delta \vdash \Penv:\tau$.
\end{proof}
\end{lemma}

\begin{lemma}\label{thm:envtypevarstable}
If $\Xi,\Delta \vdash \Penv:\Gamma$ and $X$ is not in $\Delta$,
then $\Xi,\Delta[X=\rho] \vdash \Penv:\Gamma$ for any $\rho$.
\begin{proof}
By induction on derivations of $\Xi,\Delta \vdash \Penv:\Gamma$.
\end{proof}
\end{lemma}

\begin{lemma}\label{thm:exprsubst}
Suppose $\Sigma,\Gamma,\Delta \vdash \Pexp:\tau\mathrel{\kw{goes}}d$
and $\Sigma,\Delta[X=\rho] \vdash \tau \rightsquigarrow \tau'$.
Then $\Sigma,\Gamma,\Delta[X=\rho] \vdash \Pexp:\tau'\mathrel{\kw{goes}}d$.
\end{lemma}
\begin{proof}
By induction on typing derivations.
\end{proof}

\begin{lemma}\label{thm:storesubst}
Suppose $\Xi,\Sigma,\Delta \vdash \Pstore$
and $\Sigma,\Delta[X=\rho] \vdash \Gamma \rightsquigarrow \Gamma'$.
Then $\Xi,\Sigma,\Delta[X=\rho] \vdash \Pstore$.
\end{lemma}
\begin{proof}
By induction on typing derivations.
\end{proof}

\begin{lemma}\label{thm:declsubst}
Suppose $\Sigma,\Gamma,\Delta \vdash \Pdecl \dashv \Sigma',\Gamma',\Delta'$.
Take $X$, $\rho$, $\hat{\Gamma}$, and $\hat{\Gamma}'$ such that
$\Sigma,\Delta[X=\rho] \vdash \Gamma \rightsquigarrow \hat{\Gamma}$
and
$\Sigma,\Delta[X=\rho] \vdash \Gamma' \rightsquigarrow \hat{\Gamma}'$.
Then
$\Sigma,\hat{\Gamma},\Delta[X=\rho] \vdash \Pdecl
\dashv \Sigma',\hat{\Gamma}',\Delta'[X=\rho]$.
\begin{proof}
By induction on derivations of $\Sigma,\Gamma,\Delta \vdash \Pdecl \dashv \Sigma',\Gamma',\Delta'$.
\end{proof}
\end{lemma}

\begin{lemma}\label{thm:stmtsubst}
Suppose $\Sigma,\Gamma,\Delta \vdash \Pstmt \dashv \Sigma',\Gamma'$.
Take $X$, $\rho$, $\hat{\Gamma}$, and $\hat{\Gamma}'$ such that
$\Sigma,\Delta[X=\rho] \vdash \Gamma \rightsquigarrow \hat{\Gamma}$
and
$\Sigma,\Delta[X=\rho] \vdash \Gamma' \rightsquigarrow \hat{\Gamma}'$.
Then
$\Sigma,\hat{\Gamma},\Delta[X=\rho] \vdash \Pstmt
\dashv \Sigma',\hat{\Gamma}'$.
\begin{proof}
By induction on derivations of $\Sigma,\Gamma,\Delta \vdash \Pstmt \dashv \Sigma',\Gamma'$.
\end{proof}
\end{lemma}

\begin{lemma}\label{thm:valtypestable}
If $\Xi,\Sigma,\Delta \vdash \Pval:\tau$ and $\Xi' \supseteq \Xi$,
then $\Xi',\Sigma,\Delta \vdash \Pval:\tau$.
\begin{proof}
By induction on value typing derivations.
\end{proof}
\end{lemma}

\begin{theorem}\label{thm:ctevalcorrect}
Suppose $\Pcteval{\Sigma, \Pexp}{\Pval}$ and take any
$\Pcp,\Ptypedefs,\Pstore,\Penv$ with
$\Sigma \vdash \langle \Pstore, \Penv \rangle$.
For any $\Pval'$, if $\Pexpreval{\Pcp}{\Ptypedefs}{\Pstore}{\Penv}{\Pexp} {\Pstore'}{\Pval'}$
then $\Pval' = \Pval$ and $\Sigma \vdash \langle \Pstore', \Penv \rangle$.
\begin{proof}
We proceed by induction on compile-time evaluation.
\begin{description}
\item[Case \textsc{CTE-Bool}.] Trivial.

\item[Case \textsc{CTE-Bit}.] Trivial.

\item[Case \text{CTE-Var}.]
By inversion on $\Pcteval{\Sigma, x}{\Pval}$ we know that $\Sigma(x) = \Pval$,
and by inversion on $\Pexpreval{\Pcp}{\Ptypedefs}{\Pstore}{\Penv}{x} {\Pstore'}{\Pval'}$
we have that $\Penv(x)=\ell$ and $\Pstore(\ell)=\Pval'$.
By definition of $\Sigma \vdash \langle \Pstore, \Penv \rangle$
we have ${\Pstore(\Penv(x))} = {\Sigma(x)}$, and thus ${\Pval} = {\Pval'} $,
and by assumption $\Sigma \vdash \langle \Pstore, \Penv \rangle$.

\item[Case \text{CTE-UOp}.]
By inversion on $\Pcteval{\Sigma, {\ominus}\Pexp}{{\ominus}\Pval}$ we know that
$\Pcteval{\Sigma, \Pexp}{ \Pval } $, and by inversion on
$\Pexpreval{\Pcp}{\Ptypedefs}{\Pstore}{\Penv}{{\ominus}\Pexp}
          {\Pstore'}{\Puopeval{\ominus}{\Pval'}}$
we have that $\Pexpreval{\Pcp}{\Ptypedefs}{\Pstore}{\Penv}{\Pexp}
          {\Pstore'}{\Pval'}$. By applying the induction hypothesis we get that
$\Pval=\Pval'$ and $\Sigma \vdash \langle \Pstore', \Penv \rangle$,
and thus $\ominus\Pval=\ominus\Pval'$.

\item[Case \text{CTE-BinOp}.]
By inversion on $\Pcteval{\Sigma, \Pexp_1\oplus\Pexp_2}{\Pval_1\oplus\Pval_2}$,
we have that $\Pcteval{\Sigma, \Pexp_1}{\Pval_1}$ and $\Pcteval{\Sigma, \Pexp_2}{\Pval_2}$,
and by inversion on $\Pexpreval{\Pcp}{\Ptypedefs}{\Pstore}{\Penv}{\Pexp_1\oplus\Pexp_2}
          {\Pstore_2}{\Pbinopeval{\oplus}{\Pval_1'}{\Pval_2'}}$
we have that $\Pexpreval{\Pcp}{\Ptypedefs}{\Pstore}{\Penv}{\Pexp_1}{\Pstore_1}{\Pval_1'}$
and $\Pexpreval{\Pcp}{\Ptypedefs}{\Pstore_1}{\Penv}{\Pexp_2}{\Pstore_2}{\Pval_2'}$.
By applying the induction hypothesis we respectively get that $\Pval_1 = \Pval_1'$
and $\Sigma \vdash \langle \Pstore_1, \Penv \rangle$, and applying the
induction hypothesis again we get $\Pval_2 = \Pval_2'$ and
$\Sigma \vdash \langle \Pstore_2, \Penv \rangle$. Thus we have that
$\Pval_1\oplus\Pval_2 = \Pval_1'\oplus\Pval_2'$.
\end{description}
\end{proof}
\end{theorem}

The following theorems show that evaluation does not change constant variables.
\begin{theorem}
Suppose that 
$\Sigma, \Gamma, \Delta \vdash \Pexp:\tau$
and $\Sigma(x)=\Pval$.
If $\Sigma \vdash \langle \Pstore, \Penv \rangle$ and
$\langle \Pcp, \Delta, \Pstore, \Penv, \Pexp \rangle \Downarrow \langle \Pstore', \Psig \rangle$,
then $\Pstore'(\Penv(x)) = \Pval$.
\end{theorem}
\begin{proof}
By induction on the evaluation relation.
\end{proof}

\begin{theorem}
Suppose that 
$\Sigma, \Gamma, \Delta \vdash \Pstmt \dashv \Sigma', \Gamma'$
and $\Sigma(x)=\Pval$.
If $\Sigma \vdash \langle \Pstore, \Penv \rangle$ and
$\langle \Pcp, \Delta, \Pstore, \Penv, \Pstmt \rangle \Downarrow \langle \Pstore', \Penv', \Psig \rangle$,
then $\Pstore'(\Penv'(x)) = \Pval$.
\end{theorem}
\begin{proof}
By induction on the evaluation relation.
\end{proof}

\begin{theorem}
Suppose that 
$\Sigma, \Gamma, \Delta \vdash \Pdecl \dashv \Sigma', \Gamma', \Delta'$
and $\Sigma(x)=\Pval$.
If $\Sigma \vdash \langle \Pstore, \Penv \rangle$ and
$\langle \Pcp, \Delta, \Pstore, \Penv, \Pdecl \rangle \Downarrow \langle \Delta', \Pstore', \Penv', \Psig \rangle$,
then $\Pstore'(\Penv'(x)) = \Pval$.
\end{theorem}
\begin{proof}
By induction on the evaluation relation.
\end{proof}

As a corollary (or really corollaries) to the above, if
$Sigma \vdash \langle \Pstore, \Penv \rangle$ before an evaluation
then $\Sigma \vdash \langle \Pstore', \Penv' \rangle$ after the
evaluation. With this in mind, in the soundness theorems below we only
bother proving agreement between the constant environment and the
run-time environment when we add new values to the constant
environment, since we can now see that old values are left untouched.

\subsection{Lemmas about semantic typing}

\begin{lemma}\label{thm:valsemanticsubst}
Suppose $\Xi,\Delta,\Sigma \vDash \Pval : \tau$.
If $\Sigma,\Delta[X=\rho] \vdash \tau \rightsquigarrow \hat{\tau}$,
then $\Xi,\Delta[X=\rho],\Sigma \vDash \Pval:\hat{\tau}$.
\begin{proof}
By induction on typing derivations.
\end{proof}
\end{lemma}

\begin{lemma}\label{thm:stmtsemanticsubst}
Suppose $\Sigma,\Gamma,\Delta \vDash \Pstmt \Pddashv \Sigma',\Gamma'$.
Take $X$, $\rho$, $\hat{\Gamma}$, and $\hat{\Gamma}'$ such that
$\Sigma,\Delta[\Pmany{X=\rho}] \vDash \Gamma \rightsquigarrow \hat{\Gamma}$
and
$\Sigma,\Delta[\Pmany{X=\rho}] \vDash \Gamma' \rightsquigarrow \hat{\Gamma}'$.
Then $\Sigma,\hat{\Gamma},\Delta[X=\rho] \vDash \Pstmt
\Pddashv \Sigma',\hat{\Gamma}'$.
\begin{proof}
Recall that semantic typing is universally quantified over type
arguments. Let $\Pmany{Y}$ be the variables in $\Delta$ aside from
$\Pmany{X}$ and take a list of type arguments $\Pmany{\rho'}$ for
these variables $Y$. Then because the results of reducing under
$\Delta[\Pmany{X=\rho},\Pmany{Y=\rho}]$ or
$\Delta[\Pmany{X=\rho}]$ and then $\Delta[\Pmany{Y=\rho}]$ are equal,
the desired semantic typing will hold.
\end{proof}
\end{lemma}

\begin{lemma}\label{thm:declsemanticsubst}
Suppose $\Sigma,\Gamma,\Delta \vDash \Pdecl \Pddashv \Sigma',\Gamma',\Delta'$.
Take $X$, $\rho$, $\hat{\Gamma}$, and $\hat{\Gamma}'$ such that
$\Sigma,\Delta[X=\rho] \vDash \Gamma \rightsquigarrow \hat{\Gamma}$
and
$\Sigma,\Delta[X=\rho] \vDash \Gamma' \rightsquigarrow \hat{\Gamma}'$.
Then
$\Sigma,\hat{\Gamma},\Delta[X=\rho] \vDash \Pdecl
\Pddashv \Sigma',\hat{\Gamma}',\Delta'[X=\rho]$.
\begin{proof}
As in \cref{thm:stmtsemanticsubst}.
\end{proof}
\end{lemma}

\begin{lemma}\label{thm:envtypestable}
If $\Xi \vDash \Penv:\Gamma$ and $\Xi' \supseteq \Xi$,
then $\Xi' \vDash \Penv:\Gamma$.
\begin{proof}
By induction on environment typing derivations.
\end{proof}
\end{lemma}

\subsection{Termination theorems}
The following theorems are morally proved by mutual induction,
but we write them as separate proofs. The reader may verify that uses
of theorems are limited to structural subterms or subderivations.

We also avoid proving the safety of exceptional control flow for most
program constructs. We have omitted rules, but the idea is that an exit
or return freezes the state and is propagated up to the top of the
program. The only exception to this is the function call rule, which
will ``catch'' a return statement and also perform copy-out operations
before propagating an exit command. We provide explicit rules for these
situations and prove them correct. The structure of the other cases
is such that the safety for intermediate states of the computation is
proved in passing in order to get to the safety of the final state;
this would not work if freezing the state was unsafe.

\begin{lemma}\label{thm:termstable}
If $\Xi,\Ptypedefs,\Sigma\vDash \Pval:\tau$
and $\Xi' \supseteq \Xi$,  $\Ptypedefs' \supseteq \Ptypedefs$, and $\Sigma' \supseteq \Sigma$,
then $\Xi',\Ptypedefs',\Sigma' \vDash \Pval:\tau$.
\end{lemma}
\begin{proof}
By induction on semantic typing.
\end{proof}

\begin{theorem}\label{thm:exprterm}
If $\Sigma,\Gamma,\Delta \vdash \Pexp:\tau~\kw{goes}~d$,
then
$\Sigma,\Gamma,\Delta \vDash \Pexp:\tau~\kw{goes}~d$.
\end{theorem}
\begin{proof}
We will proceed by induction on typing derivations.
Unfolding the definition of semantic typing for expressions, we take
\begin{enumerate}
\item a list of types $\rho$,
\item a store typing $\Xi$,
\item a reduced type $\hat{\tau}$,
\item a reduced typing context $\hat{\Gamma}$,
\item a store $\Pstore$,
\item and an environment $\Penv$.
\end{enumerate}
We assume the following conditions.
\begin{enumerate}
\item $\Sigma, \Delta[\Pmany{X=\rho}] \vdash \tau \rightsquigarrow \hat{\tau}$
\item $\Sigma, \Delta[\Pmany{X=\rho}] \vdash \Gamma \rightsquigarrow \hat{\Gamma}$
\item $\Xi,\Sigma,\Delta[\Pmany{X=\rho}] \vDash \langle \Pstore, \Penv \rangle:\hat{\Gamma}$
\end{enumerate}
Again unfolding the definition of semantic typing,
it remains to show that there exists a final configuration $\langle \Pstore', \Psig \rangle$
and a store typing $\Xi' \supseteq \Xi$ such that the following
conditions hold.
\begin{enumerate}
\item $\langle \mathcal{C},\Delta[\Pmany{X=\rho}], \Pstore, \Penv, \Pexp \rangle
\Downarrow \langle \Pstore', \Psig \rangle$
\item $\Xi',\Sigma,\Delta[\Pmany{X=\rho}] \vDash \langle\Pstore', \Penv\rangle:\hat{\Gamma}$
\item If $\Psig = \Pval$ then $\Sigma, \Xi', \Delta[\Pmany{X=\rho}] \vDash \Pval:\hat{\tau}$.
\end{enumerate}
We now give the inductive cases.
\begin{description}
\item[Case \textsc{T-Bool}.]
From the typing rule we have $\Pexp=b$, where $b$ is the metavariable
for boolean constants, and $\tau=\Pbool$. Observe any reduction
$\hat{\tau}$ of $\Pbool$ has to be $\Pbool$, since it contains no type
names or expressions that could be changed by reduction. Choose
$\Pstore'=\Pstore$ and $\Psig=b$.  Applying \textsc{E-Bool} proves
$\langle \mathcal{C},\Delta[\Pmany{X=\rho}], \Pstore, \Penv, \Pexp \rangle
\Downarrow \langle \Pstore, b \rangle$.
Because the store is left unchanged, we do not need to give it a new
store typing.  We already have
$\Xi,\Delta[\Pmany{X=\rho}] \vDash \langle \Pstore, \Penv \rangle:\hat{\Gamma}$
from assumption (3). So $\Xi,\Sigma,\Delta[\Pmany{X=\rho}] \vDash
b:\hat{\tau}$.

\item[Case \textsc{T-Bit}.] Trivial.

\item[Case \textsc{T-Integer}.] Trivial.

\item[Case \textsc{T-Var}.]
The typing premises are $x\notin\Pdom{\Sigma}$ and $\Gamma(x)=\tau$.
From $\Xi,\Sigma,\Delta[\Pmany{X=\rho}] \vDash \langle \Pstore, \Penv \rangle:\hat{\Gamma}$
we can conclude that there is a location $\ell$ such that the following all hold.
\begin{align*}
\hat{\Gamma}(x)&=\hat{\tau} \\
\Penv(x)&=\ell \\
\Pstore(\ell) &= \Pval \\
\Xi(\ell) &= \hat{\tau}
\end{align*}
By \textsc{E-Var} there is a reduction \[
\Pexpreval{\Pcp}{\Ptypedefs[\Pmany{X=\rho}]}{\Pstore}{\Penv}{x}
          {\Pstore}{\Pval}.
\]
We know $\Xi,\Sigma,\Delta[\Pmany{X=\rho}] \vDash \Pval:\hat{\tau}$
from the typing of $\Pstore$.
Taking $\Xi'=\Xi$, clearly $\Xi',\Sigma,\Delta[\Pmany{X=\rho}] \vDash \langle\Pstore', \Penv\rangle:\hat{\Gamma}$ and $\Xi',\Sigma,\Delta[\Pmany{X=\rho}] \vDash \Pval:\hat{\tau}$. So we're done.

\item[Case \textsc{T-Var-Const}.]
As in the previous case.

\item[Case \textsc{T-Index}.]
From the typing premises and induction hypothesis we have the following.
\begin{align}
\Sigma,\Gamma,\Delta &\vDash \Pexp_1:\tau[n]\Pgoes{d} \\
\Sigma,\Gamma,\Delta &\vDash \Pexp_2:\Pbittype{32}
\end{align}
So we can evaluate both of these expressions to values $\Pval_1$ and
$\Pval_2$, with new stores $\Pstore_1, \Pstore_2$ and store typings
$\Xi_1,\Xi_2$ satisfying some typing conditions.
\begin{align}
\Xi_1,\Sigma,\Delta[\Pmany{X=\rho}] \vDash \langle\Pstore_1, \Penv\rangle:\hat{\Gamma}
\\
\Sigma, \Xi_1, \Delta[\Pmany{X=\rho}] \vDash \Pval_1:\hat{\tau}[n]
\\
\label{eqn:indexstore2type}
\Xi_2,\Sigma,\Delta[\Pmany{X=\rho}] \vDash \langle\Pstore_2, \Penv\rangle:\hat{\Gamma}
\\
\Sigma, \Xi_2, \Delta[\Pmany{X=\rho}] \vDash \Pval_2:\Pbittype{32}
\end{align}
Inverting the value typings shows $\Pval_1
= \Pstackval{\hat{\tau}}{\Pmany{\Pval}}$ and $\Pval_2 = n_{32}$, where
$n\geq0$. Let $B=\kw{len}(\Pmany{\Pval})$. The index $n$ is either in
bounds or out of bounds.
\begin{description}
\item[Case $n < B$.]
Applying \textsc{E-Index} we have
$\Pexpreval{\Pcp}{\Ptypedefs[\Pmany{X=\rho}]}{\Pstore}{\Penv}{\Pexp_1[\Pexp_2]}
          {\Pstore_2}{\Pval_n}$.
We have a value typing
$\Sigma, \Xi_1, \Delta[\Pmany{X=\rho}] \vdash \Pval_n:\hat{\tau}$ by
inverting the typing of $\Pval_1$. Since arrays can only contain
values of base type (see the syntax in \cref{fig:types}), this amounts
to a semantic typing. It still holds when we move to $\Xi_2$
by \cref{thm:valtypestable}. We have already typed $\Pstore_2$
(\ref{eqn:indexstore2type}), so we're done.

\item[Case $n \geq B$.]
Applying \textsc{E-IndexOOB}, we have
$\Pexpreval{\Pcp}{\Ptypedefs[\Pmany{X=\rho}]}{\Pstore}{\Penv}{\Pexp_1[\Pexp_2]}
          {\Pstore_2}{\Phavoc{\hat{\tau}}}$.
Now the resulting value will typecheck because of the definition of \Phavoc.
As in the previous case, this will be a value of base type and so
its syntactic typing implies semantic typing. Again we have already
typed $\Pstore_2$ (\ref{eqn:indexstore2type}), so we're done.
\end{description}

\item[Case \textsc{T-Slice}.]
Repeated applications of the induction hypothesis yield
evaluations for $\Pexp_1$, $\Pexp_2$, and $\Pexp_3$
along with stores $\Pstore_1, \Pstore_2, \Pstore_3$
and typings $\Xi_1, \Xi_2, \Xi_3$. All these data satisfy appropriate
semantic typings. Inverting typings shows that the values are
respectively $n_w$, $p_\infty$, and $q_\infty$.
Applying \textsc{E-Slice} shows
$\Pexpreval{\Pcp}{\Ptypedefs[\Pmany{X=\rho}]}{\Pstore}{\Penv}{\Pexp_1[\Pexp_2:\Pexp_3]}
          {\Pstore_3}{\kw{slice}(\Pint{n}{w},p,q)}$.
Now, the values of the indices were bounds-checked in
typechecking. These checks are accurate to the run-time values $p$ and
$q$ (\cref{thm:ctevalcorrect}), so the result
$\kw{slice}(\Pint{n}{w},p,q)$ will be well-defined with the expected
type $\Pbittype{p-q+1}$. It typechecks semantically because it is not a
closure. The final environment is the same as the one obtained in the
evaluation of $\Pexp_3$ so we already know it satisfies semantic
typing.

\item[Case \textsc{T-UOp}.]
As in the previous case. here the assumption about the agreement of
$\mathcal{T}$ and $\mathcal{E}$ implies that the final configuration
is defined and safe with type $\tau_2$ under $\Xi_1$.  Since it has
base type (see the $\mathcal{T}$ assumptions) it typechecks
semantically as well.

\item[Case \textsc{T-BinOp}.]
As in the unary operation case.

\item[Case \textsc{T-Cast}.]
P4 only permits safe casts between base types, so this is effectively
the same as the unary operation case after an invocation
of \cref{thm:typeevalagree}.

\item[Case \textsc{T-Record}.]
Repeated applications of the induction hypothesis yield
evaluations for all the expressions $\Pmany{\Pexp}$, new stores, and new store typings.
Call the last store $\Pstore'$, and the last store typing $\Xi'$. In
symbols, we have the following.
\begin{align}
&\langle \mathcal{C},\Delta[\Pmany{X=\rho}], \Pstore, \Penv, \Pmany{\Pexp} \rangle
\Downarrow \langle \Pstore', \Pmany{\Pval} \rangle
\\
&
\Xi',\Sigma,\Delta[\Pmany{X=\rho}] \vDash \langle\Pstore', \Penv\rangle:\hat{\Gamma}
\\
\Sigma, \Xi', \Delta[\Pmany{X=\rho}] \vDash \Pmany{\Pval:\hat{\tau}}
\end{align}
We apply \textsc{E-Rec} to show $\langle \mathcal{C},\Delta[\Pmany{X=\rho}], \Pstore, \Penv, \Prec{\Pmany{f=\Pexp}} \rangle
\Downarrow \langle \Pstore', \Prec{\Pmany{f=\Pval}} \rangle$.
The final store satisfies semantic typing under $\Xi'$, as we have
already seen.  The final value has base type so it only needs to
typecheck syntactically. Applying \textsc{TV-Rec} and using the value
typing works.

\item[Case \textsc{T-MemRec}.]
We reduce $\Pexp$ to a value using the induction hypothesis.
\begin{align*}
&\langle \mathcal{C},\Delta[\Pmany{X=\rho}], \Pstore, \Penv, \Pexp \rangle
\Downarrow \langle \Pstore', \Psig \rangle
\\
&\Xi',\Sigma,\Delta[\Pmany{X=\rho}] \vDash \langle\Pstore', \Penv\rangle:\hat{\Gamma}
\\
&\Sigma, \Xi', \Delta[\Pmany{X=\rho}] \vDash \Pval:\Prectype{\Pmany{f:\hat{\tau}}}
\end{align*}
Inverting the typing of $\Pval$ shows that it has the form $\Pval
= \Prec{\Pmany{f:\hat{\tau}=\Pval}}$ and
$\Sigma, \Xi', \Delta[\Pmany{X=\rho}] \vdash \Pval_i:\hat{\tau}$.  for
each $\Pval_i$ in the sequence $\Pmany{\Pval}$.
Then an application of \textsc{E-RecMem} shows
$\Pexpreval{\Pcp}{\Ptypedefs[\Pmany{X=\rho}]}{\Pstore}{\Penv}{\Pexp.f_i}
           {\Pstore'}{\Pval_i}$.
As a field of a record, $\hat{\tau}$ has to be a base type. This means
that the syntactic typing we already have implies its semantic typing.

\item[Case \textsc{T-MemHdr}.]
Similarly to the previous case, we reduce $\Pexp$ to a value $\Pval_r$
which must be a header. Headers, however, also include a validity
bit. If the header is valid, then things still work out like the
previous case. If the header is not valid, the resulting value is
$\kw{havoc}(\hat{\tau}_i)$ instead of $\Pval_i$. This is a value of
type $\hat{\tau}_i$ in any context, even if it is an arbitrary value.
Since $\hat{\tau}_i$ has to be a base type, as a field type of a
reader, we can conclude that the value typechecks semantically.

\item[Case \textsc{T-Enum}.] Trivial.

\item[Case \textsc{T-Err}.] Trivial.

\item[Case \textsc{T-Match}.] Trivial.

\item[Case \textsc{T-Call}.]
We have the following from the typing rule. We write $\Pmany{Y}$ and
$\Pmany{\rho'}$ for the type arguments in the call to distinguish them from
the variables $\Pmany{X}$ and types $\Pmany{\rho}$ introduced
earlier. Note that because the variables $\Pmany{Y}$ are bound, they
cannot have been in $\Delta$ and must be disjoint from the variables
$\Pmany{X}$.
\begin{align}
\Sigma,\Gamma,\Delta &\vdash \Pexp:\Pfunctiontype{\Pmany{Y}}{\Pmany{d~x:\tau}}{\tau_{\mathit{ret}}} \\
\Sigma,\Delta[\Pmany{Y=\rho'}] &\vdash \Pmany{\tau} \rightsquigarrow \Pmany{\tau'} \\
\Sigma,\Gamma,\Delta &\vdash \Pmany{\Pexp:\tau'\mathrel{\kw{goes}}d} \\
\Sigma,\Delta[\Pmany{Y=\rho'}] &\vdash \tau_{\mathit{ret}} \rightsquigarrow \tau_{\mathit{ret}}'
\end{align}
From the induction hypothesis we immediately obtain the following semantic typings.
\begin{align}
\Sigma,\Gamma,\Delta &\vDash \Pexp:\Pfunctiontype{\Pmany{Y}}{\Pmany{d~x:\tau}}{\tau_{\mathit{ret}}} \\
\label{eqn:firsttypeforargs}
\Sigma,\Gamma,\Delta &\vDash \Pmany{\Pexp:\tau'\mathrel{\kw{goes}}d}
\end{align}

We begin by evaluating $\Pexp$.
Let $\tau_f =
\Pfunctiontype{\Pmany{Y}}{\Pmany{d~x:\tau}}{\tau_{\mathit{ret}}}$
and define $\hat{\tau}_f$
by \[\Sigma,\Delta[\Pmany{X=\rho}] \vdash \tau_f \rightsquigarrow \hat{\tau}_f.\]
This type $\hat{\tau}_f$
is really $\Pfunctiontype{\Pmany{Y}}{\Pmany{d~x:\hat{\tau}_1}}{\hat{\tau}_{\mathit{ret},1}}$,
where the types marked with a hat are reductions in the same way.
The semantic typing of \Pexp yields a final
configuration $\langle \Pstore_1, \Pval\rangle$, a typing context
$\hat{\Gamma}_1$, and a store typing $\Xi_1 \supseteq \Xi$
such that all of the following are true.
\begin{align}
\langle \mathcal{C},\Delta[\Pmany{X=\rho}], \Pstore, \Penv, \Pexp \rangle
&\Downarrow \langle \Pstore_1, \Pval \rangle \\
\label{eqn:firsttypeforenv}
\Xi_1,\Sigma,\Delta[\Pmany{X=\rho}] &\vDash \langle\Pstore_1, \Penv\rangle:\hat{\Gamma}_1 \\
\label{eqn:clossemantictyping}
\Xi_1, \Sigma, \Delta[\Pmany{X=\rho}] &\vDash \Pval:\hat{\tau}_f.
\end{align}

Semantic typing of values implies their ordinary typing,
so (\ref{eqn:clossemantictyping}) implies
\(\Xi_1, \Sigma, \Delta[\Pmany{X=\rho}] \vdash \Pval:\hat{\tau}_f\).
Inverting this judgment shows \Pval must be either a closure or a
builtin. We now proceed by cases.

\begin{description}
\item[Closure case.]
Suppose $\Pval = \Pclos{\Penv_c}{\Pmany{Y}}{\Pmany{d~x:\hat{\tau}_1}}
{\hat{\tau}_{\mathit{ret},1}}{\Pmany{\Pdecl}~\Pstmt}$. We will construct 
an instance of \textsc{E-Call}. The first premise, evaluating $\Pexp$
to a closure, is already done. We address each remaining premise in
turn. For clarity set $\Delta_1 = \Delta[\Pmany{X=\rho},\Pmany{Y=\rho'}]$.
\begin{enumerate}
\item Showing $\langle \Delta_1, \Pstore_1, \Penv, \Pmany{\hat{\tau}_1}\rangle\Downarrow \Pmany{\hat{\tau}_2}$ for some $\Pmany{\hat{\tau}_2}$.

From the \textsc{T-Call} typing rule we have $\Sigma,\Delta[\Pmany{Y=\rho'}]\vdash\Pmany{\tau}\rightsquigarrow\Pmany{\tau'}$ and
$\Sigma,\Delta[\Pmany{X=\rho}]\vdash\Pmany{\tau}\rightsquigarrow\Pmany{\hat{\tau}_1}$.
Let $\Pmany{\hat{\tau}_2}$ be types such that
$\Sigma,\Delta_1\vdash\Pmany{\tau}\rightsquigarrow\Pmany{\hat{\tau}_2}$.
This compile-time type evaluation agrees with the runtime type evaluation
by \cref{thm:typeevalagree}.

\item Showing 
$\langle \Pcp,\Ptypedefs_1,\Pstore_1,\Penv,\Pmany{\Passign{d~x:\hat{\tau}_2}{\Pexp}}\rangle
\Downarrow_{\mathit{copy}}
\langle \Pstore_2,\Pmany{\Penvset{x}{\ell}},\Pmany{\Passign{\Plval}{\ell}}\rangle$.

First, we have $\Xi_1,\Sigma,\Delta_1 \vDash \Pstore$ by \cref{thm:storesubst}.
Second, defining $\Gamma_2$ by
$\Sigma,\Delta_1 \vdash \hat{\Gamma}_1 \rightsquigarrow \hat{\Gamma}_2$,
we have 
$\Xi_1 \vdash \Penv:\hat{\Gamma}_2$
and the typing (\ref{eqn:firsttypeforenv}).
Third, we have
$\Sigma,\hat{\Gamma}_1,\Delta_1 \vDash \Pmany{\Pexp:\hat{\tau}_2}$
from \cref{thm:exprsubst} and the typing (\ref{eqn:firsttypeforargs}).

Those three facts are enough to instantiate \cref{thm:copyterm},
which proves the evaluation step we were after as well as the
following facts.
\begin{align}
\label{eqn:closenvgammatyping}
&\Xi_2,\Sigma,\Delta_1 \vDash \langle \Pstore_2,\Penv \rangle:\hat{\Gamma}_1 \\
&\Xi_2(\Pmany{\ell}) = \Pmany{\hat{\tau}_2} \\
&\Sigma,\hat{\Gamma}_1,\Delta_1 \vdash \Pmany{\Plval:\hat{\tau}_2\Pgoes{\kw{inout}}}
\end{align}

\item Showing
$\langle \Pcp, \Ptypedefs_1, \Pstore_2,\Penv_c[\Pmany{x\mapsto\ell}],\Pmany{\Pdecl}\rangle
\Downarrow \langle\Ptypedefs_2, \Pstore_3, \Penv_2, \Pcont\rangle$.
We apply \cref{eqn:clossemantictyping} to produce a context $\Gamma_c$
and the following semantic typings.
\begin{align}
\label{eqn:envctyping}
&\Xi_1 \vDash \Penv_c:\Gamma_c \\
\label{eqn:closdecltyping}
&\Sigma,\Gamma_c[\Pmany{x:\hat{\tau}_2}],\Delta[\Pmany{X=\rho},\Pmany{Y~\kw{var}}] \vDash \Pmany{\Pdecl}
\Pddashv \Sigma',\Gamma_c',\Delta' \\
\label{eqn:closstmttyping}
&\Sigma',\Gamma_c'[\Preturnvoid:\hat{\tau}_{\mathit{ret,1}}],\Delta' \vDash \Pstmt
\Pddashv \Sigma'', \Gamma''
\end{align}

We set $\Delta_2 = \Delta'[\Pmany{Y=\rho'}]$.

In order to get an evaluation of \Pmany{\Pdecl} out of its semantic typing, we should
first reduce the contexts $\Gamma_c$ and $\Gamma_c'$. So, define
$\hat{\Gamma}_c$ and $\hat{\Gamma}_c'$ from the following reductions.
\begin{align}
&\Sigma, \Delta[\Pmany{X=\rho},\Pmany{Y=\rho'}], \vdash \Gamma_c[\Pmany{x:\hat{\tau}}] \rightsquigarrow \hat{\Gamma}_c \\
&\Sigma', \Delta'[\Pmany{Y=\rho'}] \vdash \Gamma_c' \rightsquigarrow \hat{\Gamma}_c'
\end{align}
We also need to show \[
\Xi_2, \Sigma, \Delta[\Pmany{X=\rho},\Pmany{Y=\rho'}] \vDash
\langle \Pstore_2, \Penv_c[\Pmany{x\mapsto\ell}]\rangle:\hat{\Gamma}_c.
\]
The first part of this, typing $\Pstore_2$, is an immediate
consequence of (\ref{eqn:closenvgammatyping}).
The second part, typing $\Penv_c[\Pmany{x\mapsto\ell}]$, is more involved. We have
$
\Xi_2, \Sigma, \Delta[\Pmany{X=\rho}] \vDash
\langle \Pstore_2, \Penv_c[\Pmany{x\mapsto\ell}]\rangle:\Gamma_c[\Pmany{x:\hat{\tau}}]
$
by splitting the environment into $\Penv_c$ and the new bindings $\Pmany{x\mapsto\ell}$.
The first part holds by \cref{thm:envtypestable} with \ref{eqn:envctyping}.
The second part holds because $\Xi_2(\Pmany{\ell}) = \Pmany{\hat{\tau}}$.

Now we can apply the semantic typing property (\ref{eqn:closdecltyping}) for the declarations,
which proves this evaluation step and produces a new environment
$\Penv_2$, store $\Pstore_3$, and store typing $\Xi_3$ such that
$\Xi_3,\Delta'[\Pmany{Y=\rho'}] \vDash \langle \Pstore_3, \Penv_2 \rangle : \hat{\Gamma}_c'$.

\item Showing $\langle\Pcp, \Ptypedefs_2, \Pstore_3, \Penv_2, \Pstmt\rangle \Downarrow
               \langle \Pstore_4, \Penv_3,\Preturn{\Pval_{\mathit{ret}}} \rangle$.

Define $\hat{\Gamma}''$ by
$\Sigma',\Delta'[\Pmany{Y=\rho'}] \vdash \Gamma'' \rightsquigarrow \hat{\Gamma}''$.
Observe that the only variables in $\Delta'$ are the variables $\Pmany{Y}$,
so we can instantiate the semantic typing (\ref{eqn:closstmttyping})
with type arguments $\Pmany{\rho'}$. We have established all its premises, so we
get the evaluation step we wanted along with a store typing $\Xi_4$ and the
following facts.
\begin{align}
&\Ptypeeval{\Sigma',\Delta'[\Pmany{Y=\rho'}]}{\hat{\tau}_{\mathit{ret},1}}{\hat{\tau}_{\mathit{ret},2}}
\\
&\Xi_4,\Sigma',\Delta_2 \vDash \langle \Pstore',\Penv'\rangle:\hat{\Gamma}''
\\
&\Xi_4,\Sigma',\Delta_2 \vDash \Pval:\hat{\tau}_{\mathit{ret},2}
\end{align}

We have to show that $\hat{\tau}_{\mathit{ret},2}$ is the same as the type $\hat{\tau}_{\mathit{ret}}$
we wanted. First note that the free type variables in $\hat{\tau}_{\mathit{ret},2}$ are exactly $\Pmany{Y}$.
Since $\Delta[\Pmany{X=\rho},Y~\kw{var}] \subseteq \Delta'$ and $\Sigma \subseteq \Sigma'$ we
have $\Ptypeeval{\Delta[\Pmany{X=\rho},\Pmany{Y=\rho'}], \Sigma}{\tau_{\mathit{ret},1}}{\hat{\tau}_{\mathit{ret},2}}$.
Then because reduction under $\Delta[\Pmany{X=\rho},\Pmany{Y=\rho'}]$ is the
same as reduction under $\Delta[\Pmany{X=\rho}]$ followed by reduction under $\Delta[\Pmany{Y=\rho'}]$,
we conclude $\hat{\tau}_{\mathit{ret},2}$ is syntactically equal to $\hat{\tau}_{\mathit{ret}}$.

\item Showing $\Passignlval{\Pcp,\Delta[\Pmany{X=\rho}],\Pstore_4}{\Penv}{\Pmany{\Passign{\Plval}{\Pstore_4(\ell)}}}{\Pstore_5}$.

Follows from \cref{thm:lvalwriteterm}.

\end{enumerate}

\item[Builtin case.]
This is the same as the usual closure case but with the involvement of
$\mathcal{N}$. The value typing rule says $\mathcal{N} \vdash
x: \Psimplefunctiontype{\Pmany{d~x:\tau}}{\tau}$, so the call will
evaluate
safely by assumption.
\end{description}
\end{description}
\end{proof}

\begin{theorem}\label{thm:stmtterm}
If $\Xi, \Sigma, \Gamma, \Delta \vdash \langle \Pcp, \Ptypedefs, \Pstore, \Penv, \Pstmt \rangle \dashv \Sigma', \Gamma'$,
then $\Xi, \Sigma, \Gamma, \Delta \vDash \langle \Pcp, \Ptypedefs, \Pstore, \Penv, \Pstmt \rangle \Pddashv \Sigma', \Gamma'$.
\begin{proof}
We will proceed by induction on typing derivations.
Unfolding the definition of semantic typing for declarations, we take
\begin{enumerate}
\item a list of types $\rho$,
\item a store typing $\Xi$,
\item reduced typing context $\hat{\Gamma}$ and $\hat{\Gamma}'$,
\item a store $\Pstore$,
\item and an environment $\Penv$.
\end{enumerate}
We assume the following conditions on these data.
\begin{enumerate}
\item $\Sigma, \Delta[\Pmany{X=\rho}] \vdash \Gamma \rightsquigarrow \hat{\Gamma}$
\item $\Sigma', \Delta[\Pmany{X=\rho}] \vdash \Gamma' \rightsquigarrow \hat{\Gamma}'$
\item $\Xi,\Sigma,\Delta[\Pmany{X=\rho}] \vDash \langle \Pstore, \Penv \rangle:\hat{\Gamma}$
\end{enumerate}
Now it remains to show that there exists a final configuration
$\langle \Pstore', \Penv', \Psig \rangle$ and store typing $\Xi'$ such
that the following conditions hold.
\begin{enumerate}
\item $\langle \mathcal{C},\Delta[\Pmany{X=\rho}], \Pstore, \Penv, \Pstmt \rangle \Downarrow \langle \Pstore', \Penv', \Psig \rangle$
\item $\Xi',\Sigma',\Delta[\Pmany{X=\rho}] \vDash \langle\Pstore', \Penv'\rangle:\hat{\Gamma}'$
\end{enumerate}
We now consider the inductive cases.
\begin{description}
\item[Case \textsc{TS-Call}.]
Call statements are handled by treating them as call expressions, so
this case follows from \cref{thm:exprterm}.

\item[Case \textsc{TS-TblCall}.]
The induction hypothesis yields a semantic typing
$\Sigma,\Gamma,\Delta \vDash \Pexp:\Ptabletype$.
As usual we will use this to evaluate $\Pexp$ to a value $\Pval_t$,
yielding a new store $\Pstore_1$ and store typing $\Xi_1$.  Inverting
the semantic typing for $\Pval_t$ shows that it must be a table value,
with all that its semantic typing entails. So in summary we have the
following.
\begin{align*}
&\langle \mathcal{C},\Delta[\Pmany{X=\rho}], \Pstore, \Penv, \Pexp \rangle
\Downarrow \langle \Pstore_1, \Ptableval{\ell}{\Penv}{\Pmany{\Pexp:x}}{\Pmany{act}} \rangle
\\
&\Xi,\Ptypedefs[\Pmany{X=\rho}],\Sigma \vDash \langle \Pstore_1, \Penv \rangle:\hat{\Gamma}
\\
&\Xi,\Ptypedefs[\Pmany{X=\rho}],\Sigma \vDash \Ptableval{\ell}{\Penv_t}{\Pmany{\Pexp_k:x_k}}{\Pmany{act}}:\Ptabletype
\\
&\Xi,\Sigma,\Delta \vDash\Penv_t:\hat{\Gamma_t}
\\
&\Sigma,\Gamma,\Delta \vDash \Pmany{\Pexp_k:\tau_k}
\\
&\Sigma,\Gamma,\Delta \vDash \Pmany{x_k:\kw{match\_kind}}
\\
&\Sigma,\Gamma,\Delta \vDash x_{\mathit{act}} : \Pfunctiontype{}{\Pmany{d~x:\tau}, \Pmany{\Pin~x_p:\tau_p}}{\{\}} &\text{for each $\mathit{act} = x_{\mathit{act}}(\Pmany{\Pexp_a})$}
\\
&\Sigma,\Gamma,\Delta \vDash \Pmany{\Pexp_k:\tau_k\Pgoes{d}} &\text{for each $\mathit{act} = x_{\mathit{act}}(\Pmany{\Pexp_a})$}
\end{align*}

Next we evaluate $\Pmany{\Pexp_k}$, producing values
$\Pmany{\Pval_k}$, a store $\Pstore_2$, and a store typing $\Xi_2$. We
now have the following.
\begin{align}
&\Peval{\Pcp,\Ptypedefs[\Pmany{X=\rho}],\Pstore_1,\Penv_c,\Pmany{\Pexp_k}}
      {\Pstore_2,\Pmany{\Pval_k}}
\\
&\Xi_2,\Sigma,\Delta[\Pmany{X=\rho}] \vDash \langle\Pstore_2, \Penv_c\rangle:\hat{\Gamma}
\\
\Sigma, \Xi_2, \Delta[\Pmany{X=\rho}] \vDash \Pmany{\Pval_k:\hat{\tau_k}}
\end{align}

Now that we have values for the keys, we can perform the match.
Our assumption about \Pcp lets us conclude
$\langle\Pcp,\ell,\Pmany{\Pval_k:x_k},\Pmany{x_{\mathit{act}}(\Pmany{x_c:\tau_c})}\rangle
\Downarrow_{\mathit{match}} x_{\mathit{act}}(\Pmany{\Pexp_c})$.

Last, we need to evaluate the function call statement
$x_{\mathit{act}}(\Pmany{\Pexp_c})$.  But we have semantic typings for
the function and all its arguments, so this follows from the same
argument given in \textsc{T-Call}.

We have all the premises of \textsc{E-Call-Table} and semantic typings for the final
state, so we're done with this case.

\item[Case \textsc{TS-Assign}.]
After applying the induction hypothesis to the premises of this rule,
we have the following semantic typings.
\begin{align}
&\Sigma,\Gamma,\Delta \vDash \Pexp_1:\tau\mathrel{\kw{goes}}\kw{inout}
\\
&\Sigma,\Gamma,\Delta \vDash \Pexp_2:\tau
\end{align}
Let $\hat{\tau}$ be defined by
$\Ptypeeval{\Sigma,\Delta[\Pmany{X=\rho}]}{\tau}{\hat{\tau}}$.
We first evaluate the left hand expression to an lvalue
using \cref{thm:lvalterm}. The lemma produces $\Plval$, $\Pstore_1$,
and $\Xi_1$ satisfying the following conditions.
\begin{align}
&\Pstmtconf{\Pcp,\Ptypedefs}{\Pstore,\Penv}{\Pexp_1}
\Downarrow_{\mathit{\Plval}}
\langle\Pstore_1,\Plval\rangle\\
&\Xi_1,\Sigma,\hat{\Gamma},\Delta[\Pmany{X=\rho}] \vDash \langle \Pstore_1,\Penv \rangle:\hat{\Gamma}\\
&\Sigma,\hat{\Gamma},\Delta[\Pmany{X=\rho}] \vDash \Plval:\hat{\tau}\Pgoes{\Pinout}
\end{align}

Next we evaluate the right hand side using \cref{thm:exprterm},
which yields \Pval, $\Pstore_2$, and $\Xi_2$ satisfying the following
conditions.
\begin{align}
&\Peval{\Pcp,\Ptypedefs[\Pmany{X=\rho}],\Pstore_1,\Penv,\Pexp_2}
       {\Pstore_2,\Pval}\\
&\Xi_2,\Sigma,\hat{\Gamma},\Delta[\Pmany{X=\rho}] \vDash \langle \Pstore_2,\Penv \rangle:\hat{\Gamma}\\
&\Xi_2,\Delta[\Pmany{X=\rho}],\Sigma \vDash \Pval:\hat{\tau}
\end{align}

Finally we apply \cref{thm:lvalwriteterm}, which yields a new store
$\Pstore_3$, the evaluation step
$\Passignlval{\Pcp,\Delta[\Pmany{X=\rho}],\Pstore_2}{\Penv}{\Passign{\Plval}{\Pval}}{\Pstore_3}$,
and all the typing side conditions we need after applying \textsc{E-Assign}.

\item[Case \textsc{TS-Exit}.]
Trivial.

\item[Case \textsc{TS-Empty}.]
Trivial.

\item[Case \textsc{TS-If}.]
The induction hypothesis gives us the following semantic typings.
\begin{align}
&\Sigma,\Gamma,\Delta \vDash \Pexp:\Pbool \\
&\Sigma,\Gamma,\Delta \vDash \Pstmt_1 \Pddashv \Sigma_1, \Gamma_1 \\
&\Sigma,\Gamma,\Delta \vDash \Pstmt_2 \Pddashv \Sigma_2, \Gamma_2
\end{align}
We first need to evaluate $\Pexp$ using its semantic typing. This
produces $\Pstore_1$, $\Xi_1$, and $\Pval$ satisfying the following
conditions.
\begin{align}
&\langle \mathcal{C},\Delta[\Pmany{X=\rho}], \Pstore, \Penv, \Pexp \rangle
\Downarrow \langle \Pstore_1, \Psig \rangle \\
&\Xi_1,\Sigma,\Delta[\Pmany{X=\rho}] \vDash \langle\Pstore_1, \Penv\rangle:\hat{\Gamma} \\
&\Sigma, \Xi_1, \Delta[\Pmany{X=\rho}] \vDash \Pval:\Pbool
\end{align}
Inverting the typing of $\Pval$ we see that $\Pval=\Ptrue$ or $\Pval
= \Pfalse$. The proofs are symmetric, we give only the $\Pval=\Ptrue$
case here.
We evaluate $\Pstmt_1$ using its semantic typing, which allows us to
use \textsc{E-IfElse-True} and prove the required semantic typing side
conditions for the final state.

\item[Case \textsc{TS-Block}.]
This is a straightforward application of the induction hypothesis
and \textsc{E-Block}.

\item[Case \textsc{TS-Ret}.]
Subexpression evaluation proceeds as in other cases to obtain a value
$\Pval$ and semantic typing side conditions for the resulting
store. Since $\Psig=\Preturn{\Pval}$, we have to show that $
\Pval:\hat{\Gamma}(\Preturnvoid)$.
Fortunately the type of $\Pval$ is $\hat{\tau}$, defined by
$\Sigma,\Delta[\Pmany{X=\rho}] \vdash \Gamma(\Preturnvoid) \rightsquigarrow \hat{\tau}$
or equivalently by $\hat{\tau} = \hat{\Gamma}(\Preturnvoid)$.

\item[Case \textsc{TS-Decl}.]
Immediate from \cref{thm:declterm}.

\end{description}
\end{proof}
\end{theorem}

\begin{theorem}\label{thm:declterm}
If
$\Sigma,\Gamma,\Delta \vdash \Pdecl \Pddashv \Sigma', \Gamma', \Delta'$
then
$\Sigma,\Gamma,\Delta \vDash \Pdecl \Pddashv \Sigma', \Gamma', \Delta'$.
\end{theorem}
\begin{proof}
We will proceed by induction on typing derivations.
Unfolding the definition of semantic typing for declarations, we take
\begin{enumerate}
\item a list of types $\rho$,
\item a store typing $\Xi$,
\item reduced typing context $\hat{\Gamma}$ and $\hat{\Gamma}'$,
\item a store $\Pstore$,
\item and an environment $\Penv$.
\end{enumerate}
We assume the following conditions on these data.
\begin{enumerate}
\item $\Sigma, \Delta[\Pmany{X=\rho}] \vdash \Gamma \rightsquigarrow \hat{\Gamma}$
\item $\Sigma', \Delta'[\Pmany{X=\rho}] \vdash \Gamma' \rightsquigarrow \hat{\Gamma}'$
\item $\Xi,\Sigma,\Delta[\Pmany{X=\rho}] \vDash \langle \Pstore, \Penv \rangle:\hat{\Gamma}$
\end{enumerate}
Now it remains to show that there exists a final configuration
$\langle \Delta'', \Pstore', \Penv', \Psig \rangle$ and a store typing
$\Xi'$ such that the following conditions hold.
\begin{enumerate}
\item $\langle \mathcal{C},\Delta[\Pmany{X=\rho}], \Pstore, \Penv, \Pdecl \rangle \Downarrow \langle \Delta'', \Pstore', \Penv', \Psig \rangle$
\item $\Delta'' = \Delta'[\Pmany{X=\rho}]$
\item $\Xi',\Sigma',\Delta'[\Pmany{X=\rho}] \vDash \langle\Pstore', \Penv'\rangle:\hat{\Gamma}'$
\end{enumerate}

We now give the inductive cases of the proof.
\begin{description}
\item[Case \textsc{Type-Const}.]
We have some typing hypotheses.
\begin{align}
&\Ptypeeval{\Sigma,\Delta}{\tau}{\tau'} \\
&\Sigma,\Gamma,\Delta \vdash \Pexp:\tau'\\
&\Pcteval{\Sigma, \Pexp}{v}
\end{align}
The induction hypothesis implies the semantic typing
$\Sigma,\Gamma,\Delta \vDash \Pexp:\tau'$.
Let $\hat{\tau}$ be the reduced type
$\Sigma,\hat{\Gamma},\Delta[\Pmany{X=\rho}]\vdash\tau\rightsquigarrow\hat{\tau}$.
Then from the semantic typing we obtain a final configuration
$\langle \Pstore', \Pval \rangle$ and a store typing $\Xi' \supseteq \Xi$
such that the following conditions all hold.
\begin{enumerate}
\item $\langle \mathcal{C},\Delta[\Pmany{X=\rho}], \Pstore, \Penv, \Pexp \rangle \Downarrow \langle \Pstore', \Pval \rangle$
\item $\Xi',\Sigma,\Delta[\Pmany{X=\rho}] \vDash \langle\Pstore', \Penv\rangle:\hat{\Gamma}$
\item $\Sigma, \Xi', \Delta[\Pmany{X=\rho}] \vDash \Pval:\hat{\tau}$.
\end{enumerate}
Let $\ell$ be a fresh location.  Applying \textsc{E-Const}
and \textsc{E-VarInit}, we have
\[
\langle \mathcal{C},\Delta[\Pmany{X=\rho}], \Pstore, \Penv, \Pexp \rangle
\Downarrow \langle \Delta[\Pmany{X=\rho}], \Pstore'[\Passign{\ell}{\Pval}], \Penv[\Penvset{x}{\ell}], \Pok \rangle
\]
as desired. This final configuration typechecks under $\Xi'$ because
$\hat{\Gamma}' = \hat{\Gamma}[x:\hat{\tau}]$.

Additionally we should show
$\Sigma[x=v] \vdash \langle \Pstore'[\Passign{\ell}{\Pval}], \Penv[\Penvset{x}{\ell}] \rangle$.
Certainly $\Sigma \vdash \langle \Pstore', \Penv \rangle$, since
constant variables are never mutated once bound. To justify the new binding for
$x$, observe that $v=\Pval$ by \cref{thm:ctevalcorrect}.

\item[Case \textsc{Type-Inst}.]
We have the following typing premises. We use $C$ for the name of the
constructor to avoid clashing with the type variables $\Pmany{X}$.
\[
\begin{array}{l}
\Sigma,\Gamma,\Delta \vdash C:\Pctortype{\Pmany{x:\tau}}{\tau_{\mathit{inst}}} \\
\Ptypeeval{\Sigma,\Delta}{\tau_{\mathit{inst}}}{\tau_{\mathit{inst}}'}\\
\Sigma,\Gamma,\Delta \vdash \Pmany{\Pexp:\tau}
\end{array}
\]
The induction hypothesis gives us semantic typings for $C$ and
$\Pmany{\Pexp}$.
\[
\begin{array}{l}
\Sigma,\Gamma,\Delta \vDash C:\Pctortype{\Pmany{x:\tau}}{\tau_{\mathit{inst}}} \\
\Sigma,\Gamma,\Delta \vDash \Pmany{\Pexp:\tau}
\end{array}
\]
Define $\Pmany{\hat{\tau}}$ and $\hat{\tau}_{\mathit{inst}}$ by
$\Sigma,\hat{\Gamma},\Delta[\Pmany{X=\rho}] \vDash \Pmany{\tau}\rightsquigarrow\Pmany{\hat{\tau}}$
and $\Sigma,\hat{\Gamma},\Delta[\Pmany{X=\rho}] \vDash \Pmany{\tau_{\mathit{inst}}}\rightsquigarrow\Pmany{\hat{\tau}_{\mathit{inst}}}$.
Then semantic typing of $C$ yields a store $\Pstore_1$, a store typing
$\Xi_1$, and a value $\Pval$ such that
\begin{enumerate}
\item $\langle \mathcal{C},\Delta[\Pmany{X=\rho}], \Pstore, \Penv, C\rangle
\Downarrow \langle \Pstore_1, \Pval \rangle$,
\item $\Xi_1,\Sigma,\Delta[\Pmany{X=\rho}] \vDash \langle\Pstore_1, \Penv\rangle:\hat{\Gamma}$, and
\item $\Sigma, \Xi_1, \Delta[\Pmany{X=\rho}] \vDash \Pval:\Pctortype{\Pmany{x:\hat{\tau}}}{\hat{\tau}_{\mathit{inst}}}$
\end{enumerate}
Moving on, the semantic typing of $\Pmany{\Pexp}$ yields a store
$\Pstore_2$, a store typing $\Xi_2$, and values $\Pmany{\Pval}$ such that
\begin{enumerate}
\item $\langle \mathcal{C},\Delta[\Pmany{X=\rho}], \Pstore, \Penv, \Pmany{\Pexp} \rangle
\Downarrow \langle \Pstore_2, \Pmany{\Pval} \rangle$,
\item $\Xi_2,\Sigma,\Delta[\Pmany{X=\rho}] \vDash \langle\Pstore_2, \Penv\rangle:\hat{\Gamma}$, and
\item $\Sigma, \Xi_2, \Delta[\Pmany{X=\rho}] \vDash \Pmany{\Pval:\hat{\tau}}$.
\end{enumerate}
Semantic typing implies static typing. Inverting the static typing
judgment for $\Pval_{\mathit{clos}}$ shows that it must be a constructor closure
with run-time arguments $\Pmany{x_i:\tau_i}$, as follows.
\begin{align*}
\Pval_{\mathit{clos}} &= \Pcclos{\Penv_c}{\Pmany{d~x_i:\hat{\tau}_i}}{\Pmany{x:\hat{\tau}}}{\Pmany{\Pdecl}}{\Pstmt} \\
\hat{\tau}_{\mathit{inst}} &= \Psimplefunctiontype{\Pmany{d~x_i:\hat{\tau}_i}}{\{\}}
\end{align*}
Take fresh locations $\Pmany{\ell}$ and $\ell_f$. 

Let $\Xi_3=\Xi_2[\Pmany{\Pheapset{\ell}{\hat{\tau}}}][\Pmany{\Pheapset{\ell_f}{\hat{\tau}_{\mathit{inst}}}}]$.
Set \[
\Pval_f = \Pclossimple{\Penv_c[\Pmany{\Penvset{x}{\ell}}]}{\Pmany{d~x_i:\hat{\tau}_i}}{\Pmany{\Pdecl}}{\Pstmt}.
\]
Then we can use \textsc{E-Inst} to prove the evaluation step we want.
\[
\Peval{\Pcp,\Ptypedefs,\Pstore,\Penv,\Pinst{C}{\Pmany{\Pexp}}{x}}
      {\Ptypedefs,\Pstore_2[\Pmany{\Pheapset{\ell}{\Pval}}][\Pheapset{\ell_f}{\Pval_f}],\Penv[\Penvset{x}{\ell_f}],\Pok}
\]
From the semantic typing property of constructor closures and \cref{thm:valsemanticsubst} we have \[
\Xi_3,\Delta[\Pmany{X=\rho}],\Sigma \vDash \Pval_f:\hat{\tau}_{\mathit{inst}},
\]
so the final configuration typechecks (semantically).

\item[Case \textsc{Type-Var}.]
As in the \textsc{Type-Const} case, but with $\kw{init}_\Delta(\hat{\tau})$
instead of $\Pval$.

\item[Case \textsc{Type-VarInit}.]
As in the \textsc{Type-Const} and \textsc{Type-Var} cases.

\item[Case \textsc{T-TypeDefDecl}.]
The \textsc{E-TypeDefDecl} rule has no premises, so we can use it to
establish
$\Peval{\Pcp,\Ptypedefs[\Pmany{X=\rho}],\Pstore,\Penv, \Ptypedef{\tau}{X}}{\Ptypedefs[X
= \tau],\Pstore,\Penv,\Pok}$ immediately. The final configuration here
satisfies semantic typing because we've extended the type context with
a new name but left everything else unchanged (\cref{thm:termstable}).

\item[Case \textsc{T-EnumDecl}.]
As in \textsc{T-TypeDefDecl}.

\item[Case \textsc{T-ErrorDecl}.]
As in \textsc{T-TypeDefDecl}.

\item[Case \textsc{T-MatchKindDecl}.]
As in \textsc{T-TypeDefDecl}.

\item[Case \textsc{T-TableDecl}.]
The induction hypothesis shows the following. We treat the
$\mathsf{act\_ok}$ judgment as an abbreviation so that the induction
hypothesis applies to its premises as well.
\begin{align*}
&\Sigma,\Gamma,\Delta \vDash \Pmany{\Pexp_k : \tau_k}
\\
&\Sigma,\Gamma,\Delta \vDash \Pmany{x_k:\Pmatchkind}
\\
&\Sigma,\Gamma,\Delta \vDash \Pmany{act~\kw{act\_ok}}
\\
\Sigma,\Gamma,\Delta \vDash x_{\mathit{act}} : \Pfunctiontype{}{\Pmany{d~x:\tau}, \Pmany{\Pin~x_p:\tau_p}}{\{\}}
&\text{for each $\mathit{act} = x_{\mathit{act}}(\Pmany{\Pexp_a})$}
\\
\Sigma,\Gamma,\Delta \vDash \Pmany{\Pexp_a:\tau\Pgoes{d}}
&\text{for each $\mathit{act} = x_{\mathit{act}}(\Pmany{\Pexp_a})$}
\end{align*}
Let $\ell$ be a fresh location and let
$\Pval = \Ptableval{\ell}{\Penv}{\Pmany{\Pkey}}{\Pmany{\Pact}}$.
Applying \textsc{E-TableDecl}, we have \[
\Peval{\Pcp,\Ptypedefs[\Pmany{X=\rho}],\Pstore,\Penv,\Ptable{x}{\Pmany{\Pkey}~\Pmany{\Pact}}}
      {\Ptypedefs[\Pmany{X=\rho}],\Pstore[\Pheapset{\ell}{\Pval}],\Penv[\Penvset{x}{\ell}],\Pok}.
\]
Let $\Xi'=\Xi[\ell:\Ptabletype]$.
The final configuration typechecks syntactically because the only new
value is $\Pval$, and we have all the premises of \textsc{TV-Table}
for it. It typechecks semantically for similar reasons---we have all
the premises for table semantic typing.

\item[Case \textsc{T-CtrlDecl}.]
From the typing rule we have the following.
\begin{align}
&\Ptypeeval{\Sigma,\Delta}{\Pmany{\tau_c}}{\Pmany{\tau_c'}} \\
&\Ptypeeval{\Sigma,\Delta}{\Pmany{\tau}}{\Pmany{\tau'}} \\
&\Sigma,\Gamma[\Pmany{x_c:\tau_c}][\Pmany{x:\tau}],\Delta \vdash
 \Pmany{\Pdecl} \dashv \Sigma_1, \Gamma_1, \Delta_1 \\
&\Sigma_1, \Gamma_1[\Preturnvoid:\{\}], \Delta_1 \vdash \Pstmt \dashv \Sigma_2, \Gamma_2
\end{align}
Define $\Pmany{\hat{\tau}_c}$ (resp. \Pmany{\hat{\tau}}) by
$\Ptypeeval{\Sigma,\Delta[\Pmany{X=\rho}]}{\Pmany{\tau_c}}{\Pmany{\hat{\tau}_c}}$
(resp. $\Pmany{\tau}~ \rightsquigarrow \Pmany{\hat{\tau}}$.) These reductions
have corresponding run-time reductions (\cref{thm:typeevalagree}).
Let $\ell$ be a fresh location and define a constructor closure \Pval
by $\Pval
= \Pcclos{\Penv}{\Pmany{d~x:\hat{\tau}}}{\Pmany{x_c:\hat{\tau}_c}}{\Pmany{\Pdecl}}{\Pstmt}$.
We have all the ingredients for an application of \textsc{E-CtrlDecl}, which establishes
the following evaluation step.
\[
\Peval{\Pcp,\Ptypedefs[\Pmany{X=\rho}],\Pstore,\Penv,\Pcontroldecl{X}{\Pmany{d~x:\tau}}{\Pmany{x_c:\tau_c}}{\Pmany{\Pdecl}}{\Pstmt}}
      {\Ptypedefs[\Pmany{X=\rho}],\Pstore[\Pheapset{\ell}{\Pval}],\Penv[\Penvset{C}{\ell}],\Pok}
\]
We now turn to the semantic typing of the final configuration.
First we will make some definitions.
\begin{align*}
\tau_{\mathit{clos}} &= \Pctortype{\Pmany{x_c:\tau_c}}{\Psimplefunctiontype{\Pmany{d~x:\tau}}{\{\}}} \\
\hat{\tau}_{\mathit{clos}} &= \Pctortype{\Pmany{x_c:\hat{\tau}_c}}{\Psimplefunctiontype{\Pmany{d~x:\hat{\tau}}}{\{\}}} \\
\Xi' &= \Xi[\ell\mapsto\hat{\tau}_{\mathit{clos}}]
\end{align*}
We need to show $\Xi',\Sigma,\Delta[\Pmany{X=\rho}] \vDash \langle
\Pstore[\Pheapset{\ell}{\Pval}],\Penv[\Penvset{C}{\ell}]\rangle
:\hat{\Gamma}[C:\tau_{\mathit{clos}}]$.
The typing of $\Penv$ is evident. For $\Pstore$, we just have to show
that the new closure is semantically typed, that is, we need to establish
$\Xi',\Sigma,\Delta[\Pmany{X=\rho}] \vDash \Pval:\hat{\tau}_{\mathit{clos}}$.
Recalling the definition for semantic typing of constructor closures, we
introduce fresh locations $\ell$ and consider the ordinary closure \[
\Pclossimple{\Penv[\Pmany{x_c\mapsto\ell}]}{\Pmany{d~x:\hat{\tau}}}{\Pmany{\Pdecl}}{\Pstmt}.
\]
We need to semantically type this closure under
$\Xi'[\Pmany{\ell\mapsto\hat{\tau}_c}],\Sigma,\Ptypedefs[\Pmany{X=\rho}]$.
First define $\hat{\Gamma}_1$ and $\hat{\Gamma}_2$ by \[
\Ptypeeval{\Sigma,\Delta[\Pmany{X=\rho}]}{\Gamma_1}{\hat{\Gamma}_1}
\Ptypeeval{\Sigma_1,\Delta_1[\Pmany{X=\rho}]}{\Gamma_2}{\hat{\Gamma}_2}.
\]
Setting $\Gamma_0=\hat{\Gamma}[\Pmany{x:\hat{\tau}_c}]$, we address
each typing condition in turn.
\begin{enumerate}
\item $\Xi'[\Pmany{\ell\mapsto\tau_c}] \vDash \Penv:\Gamma_0$. Trivial.
\item $\Xi'[\Pmany{\ell\mapsto\hat{\tau}_c}],\Sigma,\Ptypedefs[\Pmany{X=\rho}] \vdash \Pclossimple{\Penv[\Pmany{x_c\mapsto\ell}]}{\Pmany{d~x:\hat{\tau}}}{\Pmany{\Pdecl}}{\Pstmt}:\hat{\tau}_{\mathit{clos}}$. Apply \textsc{TV-Clos}. Use \cref{thm:declsubst} to type the declarations and use \cref{thm:stmtsubst} to type the statements.
\item $\Xi'[\Pmany{\ell\mapsto\hat{\tau}_c}],\Sigma,\Ptypedefs[\Pmany{X=\rho}] \vDash \Pmany{\Pdecl} \Pddashv \Sigma_1,\hat{\Gamma}_1, \Delta_1[\Pmany{X=\rho}]$.
The induction hypothesis gives us a typing without the substitution for $\Pmany{X}$ and \cref{thm:declsemanticsubst} closes the gap.
\item $\Xi'[\Pmany{\ell\mapsto\hat{\tau}_c}],\Sigma_1,\Ptypedefs_1[\Pmany{X=\rho}] \vDash \Pstmt \Pddashv \Sigma_2,\hat{\Gamma}_2, \Delta_2[\Pmany{X=\rho}]$.
Again the induction hypothesis gives us a typing without $\Pmany{X=\rho}$ but in this case \cref{thm:stmtsemanticsubst} closes the gap.
\end{enumerate}

\item[Case \textsc{T-FuncDecl}.]
To avoid confusion we rename the bound type variables in the function
declaration from $\Pmany{X}$ to $\Pmany{Y}$. We first show we can
evaluate the function declaration, placing a closure into the store,
and then we give an appropriate semantic typing for the closure.

From the typing rule we have the following hypotheses.
\begin{align}
&\Gamma_1 = \Gamma[\Pmany{x_i:\tau_i}] \\
&\Delta_1 = \Delta[\Pmany{Y~\kw{var}}] \\
\label{eqn:funcdeclstmttyping}
&\Sigma, \Gamma_1[\Preturnvoid:\tau], \Delta_1 \vdash \Pstmt \dashv \Sigma_2, \Gamma_2 \\
&\Gamma'=\Gamma[x:\Pfunctiontype{\Pmany{Y}}{\Pmany{d~x_i:\tau_i'}}{\tau'}] \\
&\Delta'=\Delta \\
&\Ptypeeval{\Sigma, \Delta[\Pmany{Y~\kw{var}}]}{\Pmany{\tau_i}}{\Pmany{\tau_i'}} \\
&\Ptypeeval{\Sigma, \Delta[\Pmany{Y~\kw{var}}]}{\tau}{\tau'}
\end{align}

Define $\hat{\tau}'$ (resp. $\Pmany{\hat{\tau}_i'}$) by
$\Ptypeeval{\Sigma, \Delta[\Pmany{X=\rho},\Pmany{Y~\kw{var}}]}{\tau}{\hat{\tau}'}$
(resp $\Pmany{\tau_i} \rightsquigarrow \Pmany{\hat{\tau}_i'}$.)
These evaluations agree with corresponding runtime type evaluations
by \cref{thm:typeevalagree}.  Take a fresh location $\ell$ and set
$\Pval = \Pclos{\Penv}{\Pmany{Y}}{\Pmany{d~x:\hat{\tau_i}'}}{\hat{\tau}'}{\Pstmt}$.
Now we can apply \textsc{E-FuncDecl} to show
\[
\Peval{\Pcp,\Ptypedefs[\Pmany{X=\rho}],\Pstore,\Penv,\Pfunctiondecl{\tau}{x}{\Pmany{Y}}{\Pmany{d~x_i:\tau_i}}{\Pstmt}}
      {\Ptypedefs[\Pmany{X=\rho}],\Pstore[\Pheapset{\ell}{\Pval}],\Penv[\Penvset{x}{\ell}],\Pok}.
\]
We have to show $\Delta'' = \Delta'[\Pmany{X=\rho}]$, but that is immediate
in this case where $\Delta'' = \Delta[\Pmany{X=\rho}]$ and
$\Delta' = \Delta$. We also need to produce some $\Xi' \supseteq \Xi$
such that
\[
\Xi',\Sigma,\Delta'[\Pmany{X=\rho}] \vDash \langle\Ptypedefs[\Pmany{X=\rho}],\Pstore[\Pheapset{\ell}{\Pval}],\Penv[\Penvset{x}{\ell}],\Pok\rangle: \hat{\Gamma}'.
\]
Let $\tau_f = \Pfunctiontype{\Pmany{Y}}{\Pmany{d~x_i:\tau_i'}}{\tau'}$
and let
$\hat{\tau}_f=\Pfunctiontype{\Pmany{Y}}{\Pmany{d~x_i:\hat{\tau}_i'}}{\hat{\tau}'}$.
Then in particular $\hat{\Gamma}'=\hat{\Gamma}[x:\hat{\tau}_f]$. So
set $\Xi'=\Xi[\ell\mapsto\hat{\tau}_f]$ and observe
$\Xi' \vDash \Penv[\Penvset{x}{\ell}]$. Finally, proving 
$\Xi',\Sigma,\Delta'[\Pmany{X=\rho}] \vDash \Pstore[\Pheapset{\ell}{\Pval}]$
amounts to showing $\Xi',\Sigma,\Delta'[\Pmany{X=\rho}] \vDash \Pval:\hat{\tau}_f$,
which is going to be somewhat involved.

To prove semantic typing of our newly created closure, we will have to find
contexts $\Sigma_c$ and $\Gamma_c$ such that
\begin{enumerate}
\item $\Xi' \vDash \Penv : \hat{\Gamma}$,
\item $\Xi',\Sigma,\Delta[\Pmany{X=\rho}] \vdash \Pval:\hat{\tau}_f$, and
\item $\Sigma',\hat{\Gamma}[\Pmany{x:\hat{\tau}_i'},\Preturnvoid=\hat{\tau}'],\Delta[\Pmany{X=\rho},Y~\kw{var}] \vDash \Pstmt \Pddashv \Sigma_c, \Gamma_c$.
\end{enumerate}
We address each of these goals in turn.
\begin{enumerate}
\item By hypothesis.
\item Apply \textsc{TV-Clos}. The typing of $\Pstmt$ comes from (\ref{eqn:funcdeclstmttyping}) and \cref{thm:stmtsubst}.
\item From the inductive hypothesis, we know there are contexts $\Sigma_2$ and $\Gamma_2$ such that
\[
\Sigma, \Gamma[\Pmany{x_i:\tau_i},\Preturnvoid: \tau], \Delta[\Pmany{Y~\kw{var}}] \vDash \Pstmt \dashv \Sigma_2, \Gamma_2.
\]
Substituting $\rho$ for $\Pmany{X}$ preserves typing by \cref{thm:stmtsemanticsubst}.
\end{enumerate}
\end{description}
\end{proof}

\begin{lemma}\label{thm:lvalterm}
Let $\langle\Pcp,\Ptypedefs,\Pstore,\Penv,\Pexp\rangle$ be an initial
expression configuration.
Assume $\Sigma,\Gamma,\Delta \vdash \Pexp:\tau\Pgoes{\Pinout}$ and take
variables $\Pmany{\rho}$, a context $\hat{\Gamma}$, and a type
$\hat{\tau}$. Suppose that
\begin{enumerate}
\item $\Sigma, \Delta[\Pmany{X=\rho}] \vdash \tau \rightsquigarrow \hat{\tau}$,
\item $\Sigma, \Delta[\Pmany{X=\rho}] \vdash \Gamma \rightsquigarrow \hat{\Gamma}$, and
\item $\Xi,\Sigma,\Delta[\Pmany{X=\rho}] \vDash \langle \Pstore, \Penv \rangle:\hat{\Gamma}$.
\end{enumerate}
Then there exists a store $\Pstore'$, a typing $\Xi'\supseteq \Xi$,
and an l-value $\Plval$ such that
\begin{enumerate}
\item $\Plvaleval{\Pcp}{\Ptypedefs}{\Pstore}{\Penv}{\Pexp}{\Pstore'}{\Plval}$,
\item $\Xi',\Sigma,\Gamma,\Delta \vDash \langle \Pstore',\Penv \rangle:\Gamma$, and
\item $\Sigma,\Gamma,\Delta \vDash \Plval:\tau\Pgoes{\Pinout}$.
\end{enumerate}
\begin{proof}
By induction on the typing of $\Pexp$ with applications
of \cref{thm:exprterm}.
\end{proof}
\end{lemma}

\begin{lemma}\label{thm:copyterm}
Consider an initial configuration
$\langle \Pcp,\Ptypedefs,\Pstore,\Penv,
\Pmany{\Passign{d~x:\tau}{\Pexp}}\rangle$.
Take contexts $\Xi,\Sigma,\Gamma$ and suppose
\begin{itemize}
\item $\Xi,\Delta \vDash \Pstore$,
\item $\Xi,\Delta \vDash \Penv:\Gamma$,
\item and $\Sigma,\Gamma,\Delta \vDash \Pmany{\Pexp:\tau\Pgoes{d}}$.
\end{itemize}
Then there exists a store $\Pstore'$, a new store typing $\Xi'$,
locations $\Pmany{\ell}$, and L-values $\Pmany{\Plval}$ such that
all of the following conditions are true.
\begin{enumerate}
\item $\langle \Pcp,\Ptypedefs,\Pstore,\Penv,
\Pmany{\Passign{d~x:\tau}{\Pexp}}\rangle
\Downarrow_{\mathit{copy}}
\langle \Pstore',\Pmany{\Penvset{x}{\ell}},\Pmany{\Passign{\Plval}{\ell}}\rangle$
\item $\Xi',\Delta \vdash \langle\Pstore',\Penv\rangle:\Gamma$
\item $\Xi'(\Pmany{\ell}) = \Pmany{\tau}$
\item $\Sigma,\Gamma,\Delta \vdash \Pmany{\Plval:\tau\Pgoes{\kw{inout}}}$.
\end{enumerate}
\begin{proof}
By cases on $d$ using \cref{thm:lvalterm} and \cref{thm:exprterm}.
\end{proof}
\end{lemma}

\begin{lemma}\label{thm:lvalwriteterm}
Let $\langle \Pcp,\Delta, \Pstore, \Penv, \Passign{\Plval}{\Pval}\rangle$ be an
initial write configuration and take contexts
$\Xi,\Sigma,\Gamma,\Delta$ and a type $\tau$. If
\begin{enumerate}
\item $\Xi,\Delta \vDash \langle \Pstore, \Penv \rangle:\Gamma$,
\item $\Sigma,\Gamma,\Delta \vDash \Plval:\tau$,
\item and $\Sigma,\Xi,\Delta \vDash \Pval:\tau$,
\end{enumerate}
then there exists a store $\Pstore'$ such that
\begin{enumerate}
\item $\Passignlval{\Pcp, \Delta, \Pstore}{\Penv}{\Passign{\Plval}{\Pval}}{\Pstore'}$
\item and $\Xi,\Delta \vDash \langle \Pstore',\Penv \rangle:\Gamma$.
\end{enumerate}
\begin{proof}
By induction on $\Xi,\Sigma,\Gamma,\Delta \vdash \Plval:\tau~\Pgoes{\Pinout}$.
\end{proof}
\end{lemma}

%% file: unions-proof.tex
% !TEX root = popl21.tex
\section{Union Translation Proof}
\label{app:union-proof}

\begin{figure}
\begin{mathpar}
\Pinfer{TU-Case-Field}{
\Sigma,\Gamma[f_i:\tau_i],\Delta \vdash \Pblk \dashv \Sigma',\Gamma'
}{
\Sigma,\Gamma,\Delta \vdash \kw{switchcaseok}(\Pmany{\tau~f}, \Pswitchcase{f_i}{\Pblk})
}

\and

\Pinfer{TU-Case-Default}{
\Sigma,\Gamma,\Delta \vdash \Pblk \dashv \Sigma',\Gamma'
}{
\Sigma,\Gamma,\Delta \vdash \kw{switchcaseok}(\Pmany{\tau~f}, \Pswitchcase{\kw{default}}{\Pblk}~\kw{switchcaseok})
}
\end{mathpar}
\caption{Omitted switch case typing rules.}
\end{figure}

\begin{lemma}\label{thm:union_value_lemma}
If $\Peval{\Pcp,\Ptypedefs,\Pstore,\Penv,\Pexp}
     {\Pstore',\Pval}$, 
then $\Puniontrans{\Pval} = \Pval$.
\begin{proof}
By induction on expression evaluation rules.
\end{proof}
\end{lemma}

\begin{lemma}\label{thm:union_expression_lemma}
If $\Peval{\Pcp,\Ptypedefs,\Pstore,\Penv,\Pexp}
     {\Pstore',\Pval}$, 
then $\Peval{\Pcp,\Ptypedefs,\Puniontrans{\Pstore},\Penv,\Pexp}
     {\Puniontrans{\Pstore_1},\Pval}$.
\begin{proof}
By induction on expression evaluation rules and lemma~\ref{thm:union_value_lemma}.
\end{proof}
\end{lemma}

\begin{lemma}\label{thm:union_lval_lemma}
If $\Pstmtconf{\Pcp,\Ptypedefs}{\Pstore,\Penv}{\Pexp}
\Downarrow_{\mathit{\Plval}}
\langle\Pstore',\Plval\rangle$, 
then $\Pstmtconf{\Pcp,\Ptypedefs}{\Puniontrans{\Pstore},\Penv}{\Pexp}
\Downarrow_{\mathit{\Plval}}
\langle\Puniontrans{\Pstore'},\Plval\rangle$.
\begin{proof}
By induction on expression evaluation rules and lemma~\ref{thm:union_value_lemma}.
\end{proof}
\end{lemma}

\begin{definition}
We say $\Penvlt{\Pstore_1}{\Penv_1}{\Pstore_2}{\Penv_2}$ if $\Penv_1$'s domain is a subset of $\Penv_2$'s domain, and for all $\Plval$ in $\Penv_1$'s domain, if $\Penv_1(\Plval) = \ell_1$ and $\Pstore_1(\ell_1) = \Pval$, then $\Penv_2(\Plval) = \ell_2$ and $\Pstore_2(\ell_2) = \Pval$.
\end{definition}

\begin{theorem}\label{thm:union_semantics_commutes}
If $\Peval{\Pcp,\Ptypedefs,\Pstore,\Penv,\Pstmt}
      {\Pstore',\Penv', sig}$,
then $\Peval{\Pcp,\Ptypedefs,\Puniontrans{\Pstore},\Penv,\Puniontrans{\Pstmt}}
      {\Pstore_t,\Penv_t, sig}$ and 
$\Penvlt{\Puniontrans{\Pstore'}}{\Penv'}{\Pstore_t}{\Penv_t}$.

\begin{proof}
The proof works by induction on the statement evaluation rules, and a case analysis on the last rule in the derivation.

\begin{description}

\item[Case \textsc{E-Union-Assign}.]We know
\begin{equation}
\Pstmt = \Passign{\Pexp_1.f_i}{\Pexp_2}
\end{equation}
\begin{equation}
\Pstmtconf{\Pcp,\Ptypedefs}{\Pstore,\Penv}{\Pexp_1}
\Downarrow_{\mathit{\Plval}}
\langle\Pstore_1,\Plval\rangle
\end{equation}
\begin{equation}
\Peval{\Pcp,\Ptypedefs,\Pstore_1,\Penv,\Pexp_2}
      {\Pstore_2,\Pval}
\end{equation}
\begin{equation}\label{eq:union-assign-store-write}
\Passignlval{\Pstore_2}{\Penv}{\Passign{\Plval.f}{\Pval}}{\Pstore_3}
\end{equation}
\begin{equation}
\Peval{\Pcp,\Ptypedefs,\Pstore,\Penv,\Passign{\Pexp_1.f_i}{\Pexp_2}}
      {\Pstore_3,\Penv,\Pcont}
\end{equation}

By lemma~\ref{thm:union_expression_lemma} and lemma~\ref{thm:union_lval_lemma}, we have
\begin{equation}
\Pstmtconf{\Pcp,\Ptypedefs}{\Puniontrans{\Pstore},\Penv}{\Pexp_1}
\Downarrow_{\mathit{\Plval}}
\langle\Puniontrans{\Pstore_1},\Plval\rangle
\end{equation}
\begin{equation}
\Peval{\Pcp,\Ptypedefs,\Puniontrans{\Pstore_1},\Penv,\Pexp_2}
      {\Puniontrans{\Pstore_2},\Pval}
\end{equation}

By definition of translation we know
\begin{equation}
\Puniontrans{\Pstmt} = \Passign{\Pexp_1}{\Pbraces{tag:\Pbittype{n} = i, f_i:\tau_i = \Pexp_2, f_j:\tau_j = \kw{init}_\Ptypedefs~\tau_j}}
\end{equation}

Suppose $\Penv(\Plval.f) = \ell$. Using evaluation rules, it is straightforward to show
\begin{equation}
\Peval{\Pcp,\Ptypedefs,\Puniontrans{\Pstore},\Penv,\Puniontrans{\Pstmt}}
      {\Pstore_t,\Penv,\Pcont}
\end{equation}
\begin{equation}
\Pstore_t  = \Puniontrans{\Pstore_2}[\Pheapset{\ell}{\Pbraces{tag:\Pbittype{n} = i, f_i:\tau_i = \Pval, f_j:\tau_j = \kw{init}_\Ptypedefs~\tau_j}}]
\end{equation}

We know from (\ref{eq:union-assign-store-write}) that $\Pstore_3 = \Pstore_2[\Pheapset{\ell}{\Pval}]$. By definition of translation for stores, we have:
\begin{equation}\label{eq:union_assign_store_equivalent}
\Puniontrans{\Pstore'} = \Puniontrans{\Pstore_3} = \Puniontrans{\Pstore_2[\Pheapset{\ell}{\Punionval{X}{f_i, \Pval}}]} = \Puniontrans{\Pstore_2}[\Pheapset{\ell}{\Puniontrans{\Punionval{X}{f_i, \Pval}}}] = \Pstore_t
\end{equation}

$\Penv' = \Penv_t = \Penv$, so it follows from (\ref{eq:union_assign_store_equivalent}) that $\Penvlt{\Puniontrans{\Pstore'}}{\Penv'}{\Pstore_t}{\Penv_t}$.

\item[Case \textsc{E-UnionSwitch-Match}.] We know
\begin{equation}
\Pstmt = \Pswitch{\Pexp}{\Pmany{\Pswitchcase{\Plbl}{\Pbraces{\Pmany{\Pstmt'}}}}}
\end{equation}
\begin{equation}
\Pstmtconf{\Pcp,\Ptypedefs}{\Pstore,\Penv}{\Pexp}
\Downarrow_{\mathit{\Plval}}
\langle\Pstore_1,\Plval\rangle
\end{equation}
\begin{equation}
\Penv(\Plval) = \ell_1 \;\;\;\;\;
\Pstore_1(\ell_1) = \Punionval{X}{f_i, \Pval_i} \\
\end{equation}
\begin{equation}\label{eq:union_match_case_eq}
\kw{match\_union\_case}~(\Pmany{\Pswitchcase{\Plbl}{\Pbraces{\Pmany{\Pstmt'}}}}, f_i) = (f_i, k)
\end{equation}
\begin{equation}\label{eq:union_match_fresh}
\Pfresh{\ell_2} \;\;\;\;\;
\Pstore_2 = \Pstore_1[\Pheapset{\ell_2}{\Pval_i}]
\end{equation}
\begin{equation}\label{eq:union_match_stmt_eval}
\Peval{\Pcp,\Ptypedefs, \Pstore_2,\Penv[\Penvset{f_i}{\ell_2}], \Pbraces{\Pmany{\Pstmt'}_k}}
      {\Pstore_3,\Penv[\Penvset{f_i}{\ell_2}],sig}
\end{equation}
\begin{equation}
\Peval{\Pcp,\Ptypedefs,\Pstore,\Penv,\Pswitch{\Pexp}{\Pmany{\Pswitchcase{\Plbl}{\Pbraces{\Pmany{\Pstmt'}}}}}}
      {\Pstore_3,\Penv,sig}
\end{equation}
\\
It follows from (\ref{eq:union_match_case_eq}) and the definition of translation that
\begin{equation}
\begin{array}{rl}
\Puniontrans{\Pswitch{\Pexp}{\Pmany{\Pswitchcase{\Plbl}{\Pbraces{\Pmany{\Pstmt'}}}}}} = & \Pvardeclinst{X}{tmp}{\Pexp}, \\
& \kw{if}~(tmp.tag == j_1)~\Pbraces{\Pvardeclinst{\tau_{j_1}}{f_{j_1}}{tmp.f_{j_1}}, \Puniontrans{\Pmany{\Pstmt'}_1}} \\
& ... \\
& \kw{else\ if}~(tmp.tag == i)~\Pbraces{\Pvardeclinst{\tau_{i}}{f_{i}}{tmp.f_{i}}, \Puniontrans{\Pmany{\Pstmt'}_k}} \\
& ... \\
& \kw{else}~\Pbraces{} \\
\end{array} 
\end{equation} \\
Using evaluation rules and assuming $\Pfresh{\ell}$, we can show
\begin{equation}
\Peval{\Pcp, \Ptypedefs, \Puniontrans{\Pstore}, \Penv, \Pvardeclinst{X}{tmp}{\Pexp}}
{\Puniontrans{\Pstore_1}[\Pheapset{\ell}{\Puniontrans{\Punionval{X}{f_i, \Pval_i}}}], \Penv[\Penvset{tmp}{\ell}], \Pcont}
\end{equation}
We can also show
\begin{equation}
\Puniontrans{\Pstore_1}[\Pheapset{\ell}{\Puniontrans{\Punionval{X}{f_i, \Pval_i}}}](tmp.f_i) = \Puniontrans{\Pval_i}
\end{equation}\\
Let $\Pstore_{t1} = \Puniontrans{\Pstore_1}[\Pheapset{\ell}{\Puniontrans{\Punionval{X}{f_i, \Pval_i}}}]$, and  $\Penv_{t1} = \Penv[\Penvset{tmp}{\ell}]$. Note that $\Pfresh{\ell_2}$ holds from (\ref{eq:union_match_fresh}). Using evaluation rules, it is straightforward to show
\begin{equation}
\Peval{\Pcp,\Ptypedefs, \Pstore_{t1}, \Penv_{t1}, \Pvardeclinst{\tau_{i}}{f_{i}}{tmp.f_{i}}}
      {\Pstore_{t1}[\Pheapset{\ell_2}{\Puniontrans{\Pval_i}}],\Penv_{t1}[\Penvset{f_i}{\ell_2}],\Pcont}
\end{equation}
\begin{equation}
\Pstore_{t2} = \Pstore_{t1}[\Pheapset{\ell_2}{\Puniontrans{\Pval_i}}]
\end{equation}
\begin{equation} \label{eq:union_match_helper_1}
\Peval{\Pcp,\Ptypedefs, \Pstore_{t2}, \Penv_{t1}[\Penvset{f_i}{\ell_2}], \Pbraces{\Puniontrans{\Pmany{\Pstmt'}_k}}}
      {\Pstore_{t3}, \Penv_{t1}[\Penvset{f_i}{\ell_2}], sig_{t1}}
\end{equation}
\begin{equation}
\Peval{\Pcp, \Ptypedefs, \Pstore_{t1}, \Penv_{t1}, \Pbraces{\Pvardeclinst{\tau_{i}}{f_{i}}{tmp.f_{i}}, \Puniontrans{\Pmany{\Pstmt'}_k}}}
      {\Pstore_{t3}, \Penv_{t1}, sig_{t1}}
\end{equation}
\begin{equation}
\Peval{\Pcp, \Ptypedefs, \Puniontrans{\Pstore}, \Penv, \Puniontrans{\Pswitch{\Pexp}{\Pmany{\Pswitchcase{\Plbl}{\Pbraces{\Pmany{\Pstmt'}}}}}}}
{\Pstore_{t3}, \Penv_{t1}, sig_{t1}}
\end{equation}
\begin{equation}  \label{eq:union_match_helper_2}
\langle \Pstore_{t}, \Penv_{t}, sig_{t} \rangle = 
\langle \Pstore_{t3}, \Penv_{t1}, sig_{t1} \rangle
\end{equation}
\\
Suppose $\Pstore_2' = \Pstore_2[\Pheapset{\ell}{\Punionval{X}{f_i, \Pval_i}}]$, and $\Penv_2' = \Penv[\Penvset{tmp}{\ell}]$. Note that by definition of translation, $tmp$ is not used elsewhere in the switch statement. Thus, it follows from (\ref{eq:union_match_stmt_eval}) that
\begin{equation}
\Peval{\Pcp,\Ptypedefs,\Pstore_2',\Penv_2'[\Penvset{f_i}{\ell_2}], \Pbraces{\Pmany{\Pstmt'}_k}}
      {\Pstore_3[\Pheapset{\ell}{\Punionval{X}{f_i, \Pval_i}}], \Penv_2'[\Penvset{f_i}{\ell_2}],sig}
\end{equation}

By the induction hypothesis, we know that
\begin{equation}
\Peval{\Pcp,\Ptypedefs,\Puniontrans{\Pstore_2'},\Penv_2'[\Penvset{f_i}{\ell_2}], \Puniontrans{\Pbraces{\Pmany{\Pstmt'}_k}}}
      {\Pstore_{t3}',\Penv_{t3}',sig}
\end{equation}
such that $\Penvlt{\Puniontrans{\Pstore_3[\Pheapset{\ell}{\Punionval{X}{f_i, \Pval_i}}]}}
                            {\Penv_2'[\Penvset{f_i}{\ell_2}]}
                            {\Pstore_{t3}'}{\Penv_{t3}'}$.
\\
Moreover, by definition of translation, we have
\begin{equation}
\begin{array}{rl}
\Puniontrans{\Pstore_2'} & = \Puniontrans{\Pstore_2}[\Pheapset{\ell}{\Puniontrans{\Punionval{X}{f_i, \Pval_i}}}] \\
& = (\Puniontrans{\Pstore_1} [\Pheapset{\ell_2}{\Puniontrans{\Pval_i}}]) [\Pheapset{\ell}{\Puniontrans{\Punionval{X}{f_i, \Pval_i}}}] \\ 
& = (\Puniontrans{\Pstore_1} [\Pheapset{\ell}{\Puniontrans{\Punionval{X}{f_i, \Pval_i}}}]) [\Pheapset{\ell_2}{\Puniontrans{\Pval_i}}]  \\
& = \Pstore_{t1}[\Pheapset{\ell_2}{\Puniontrans{\Pval_i}}] \\
& = \Pstore_{t2}
\end{array}
\end{equation}

Thus, from (\ref{eq:union_match_helper_1}) and (\ref{eq:union_match_helper_2}), it follows that:
\begin{equation}
\Pstore_{t3}' = \Pstore_{t}
\end{equation}
\begin{equation}
\Penv_{t3}'  = \Penv_{t1}[\Penvset{\ell_2}{f_i}]
\end{equation}
\begin{equation}
sig_{t} = sig
\end{equation}
\begin{equation}
\Penvlt{\Puniontrans{\Pstore_3[\Pheapset{\ell}{\Punionval{X}{f_i, \Pval_i}}]}}
                            {\Penv_2'[\Penvset{f_i}{\ell_2}]}
                            {\Pstore_{t}}{\Penv_{t1}[\Penvset{f_i}{\ell_2}]}
\end{equation}
\begin{equation}
\Penvlt{\Puniontrans{\Pstore_3[\Pheapset{\ell}{\Punionval{X}{f_i, \Pval_i}}]}}
                            {\Penv_2'}
                            {\Pstore_{t}}{\Penv_{t1}} 
\end{equation}
\begin{equation}
\Penvlt{\Puniontrans{\Pstore_3[\Pheapset{\ell}{\Punionval{X}{f_i, \Pval_i}}]}}
                            {\Penv[\Penvset{tmp}{\ell}]}
                            {\Pstore_{t}}{\Penv_{t1}}
\end{equation}
\begin{equation}
\Penvlt{\Puniontrans{\Pstore_3}[\Pheapset{\ell}{\Puniontrans{\Punionval{X}{f_i, \Pval_i}}]}}
                            {\Penv[\Penvset{tmp}{\ell}]}
                            {\Pstore_{t}}{\Penv_{t1}}
\end{equation}
\begin{equation}
\Penvlt{\Puniontrans{\Pstore_3}}
                            {\Penv}
                            {\Pstore_{t}}{\Penv_{t1}}
\end{equation}
\begin{equation}
\Penvlt{\Puniontrans{\Pstore'}}
                            {\Penv'}
                            {\Pstore_{t}}{\Penv_{t}}
\end{equation}

\item[Other Cases.] The case when default is matched can be proven similar to the case above. The case when there is no match and the cases for other P4 statements are trivial.
\end{description}
\end{proof}

\end{theorem}